\documentclass[
a4paper
]{article}

\usepackage{authblk}
\usepackage{fullpage}
\usepackage[utf8]{inputenc}
\usepackage{amsmath}
\usepackage{amssymb}
\usepackage{amsthm}
\usepackage{mathtools}
\usepackage{fullpage}
\usepackage[labelfont=bf]{caption}

\renewcommand{\epsilon}{\varepsilon}

\usepackage{graphicx}
\usepackage{amsmath}
\usepackage{amsfonts}
\usepackage{subcaption}
\usepackage{footmisc}
\usepackage{algorithm}
\usepackage[noend]{algpseudocode}
\algrenewcommand\algorithmicdo{}
\algrenewcommand\algorithmicthen{}
\usepackage{booktabs}
\usepackage{todonotes}
\usepackage{xcolor}
\usepackage{bm}
\usepackage{xspace}
\usepackage[normalem]{ulem}
\usepackage{stmaryrd}
\usepackage{thm-restate}
\usepackage{pifont}
\usepackage{wrapfig}
\usepackage{oubraces}
\usepackage{nicefrac}
\allowdisplaybreaks

\usepackage{qtree}

\usepackage[normalem]{ulem}

\usepackage{microtype}

\newcommand{\set}[1]{\{#1\}}                    %
\newcommand{\setof}[2]{\{{#1}\mid{#2}\}}        %
\newcommand{\bigsetof}[2]{\left\{{#1}\ \biggr|\ {#2}\right\}}        %

\newcommand{\N}{\mathbb N} %

\newcommand{\calC}{\mathcal C}

\newcommand{\calI}{\mathcal I}

\newcommand{\calP}{\mathcal P}

\newcommand{\eat}[1]{}

\renewcommand{\epsilon}{\varepsilon}

\newcommand{\br}{\mathit{branch}}
\newcommand{\parent}{\mathit{parent}}
\newcommand{\children}{\mathit{child}}
\newcommand{\var}{\mathit{var}}
\newcommand{\ra}{\mathit{\overline{ra}}}       %
\newcommand{\raexc}{\mathit{ra}}       %
\newcommand{\ria}{\mathit{ria}}       %

\newcommand{\con}{\mathit{con}}         %
\newcommand{\conn}{\overline{con}}
\newcommand{\scon}{\mathit{conSto}}     %
\newcommand{\icon}{\mathit{conInst}}    %
\newcommand{\out}{\mathit{out}}         %
\newcommand{\iv}{\boldsymbol{I}}           %
\newcommand{\caliv}{\calI}

\newcommand{\supp}{\mathit{supp}}         %

\newcommand{\yanc}{\ensuremath{\boldsymbol{y}_{\text{anc}}}}

\newcommand{\ysto}{\ensuremath{\boldsymbol{y}_{\text{sto}}}}
\newcommand{\yinst}{\ensuremath{\boldsymbol{y}_{\text{inst}}}}

\newcommand{\fhw}{\textit{fhw}\xspace}        %

\newcommand{\ptree}{\text{PT}\xspace}       
\newcommand{\ptrees}{\text{PTs}\xspace}       
\newcommand{\pt}{\ensuremath{\mathcal{PT}}\xspace}       
\newcommand{\gj}{\ensuremath{\mathcal{GJ}}\xspace}       
\newcommand{\ptc}{\ensuremath{\mathcal{PTC}}\xspace}     
\newcommand{\ptcr}{\ensuremath{\mathcal{PTCR}}\xspace}   
\newcommand{\rptcr}{\ensuremath{\mathcal{RPT}}\xspace} 
\newcommand{\td}{\ensuremath{\mathcal{TD}}\xspace}   
\newcommand{\tdgj}{\ensuremath{\mathcal{TD}^\mathcal{GJ}}\xspace}   
\newcommand{\tdpt}{\ensuremath{\mathcal{TD}^\mathcal{PT}}\xspace}  
\newcommand{\tdptc}{\ensuremath{\mathcal{TD}^\mathcal{PTC}}\xspace}  
\newcommand{\tdptcr}{\ensuremath{\mathcal{TD}^\mathcal{PTCR}}\xspace} 
\newcommand{\tdrptcr}{\ensuremath{\mathcal{TD}^\mathcal{RPT}}\xspace}     
\newcommand{\tdp}{\ensuremath{\mathcal{TD}^\mathcal{C}}\xspace}   
\newcommand{\tdc}{\ensuremath{\mathcal{TD}^\mathcal{C}}\xspace}   

\newcommand{\anc}{\mathit{anc}}     
\newcommand{\ancc}{\overline{anc}}     
\newcommand{\desc}{\mathit{desc}}     
\newcommand{\descc}{\overline{desc}}   
\newcommand{\outt}{\overline{out}}

\newcommand{\0}{\mathbf{0}}
\newcommand{\1}{\mathbf{1}}
\newcommand{\K}{\mathbb{K}}
\newcommand{\dom}{\mathbf{dom}}

\newcommand{\sem}[1]{{\llbracket{}{#1}\rrbracket}}

\newcommand{\keys}[1]{{keys({#1})}}

\newcommand{\ptroot}{root} %
\newcommand{\Root}{root} %

\usetikzlibrary{fit}
\usetikzlibrary{calc}
\usetikzlibrary{trees}
\newcommand{\convexpath}[2]{
  [   
  create hullcoords/.code={
    \global\edef\namelist{#1}
    \foreach [count=\counter] \nodename in \namelist {
      \global\edef\numberofnodes{\counter}
      \coordinate (hullcoord\counter) at (\nodename);
    }
    \coordinate (hullcoord0) at (hullcoord\numberofnodes);
    \pgfmathtruncatemacro\lastnumber{\numberofnodes+1}
    \coordinate (hullcoord\lastnumber) at (hullcoord1);
  },
  create hullcoords
  ]
  ($(hullcoord1)!#2!-90:(hullcoord0)$)
  \foreach [
  evaluate=\currentnode as \previousnode using \currentnode-1,
  evaluate=\currentnode as \nextnode using \currentnode+1
  ] \currentnode in {1,...,\numberofnodes} {
    let \p1 = ($(hullcoord\currentnode) - (hullcoord\previousnode)$),
    \n1 = {atan2(\y1, \x1) + 90}, %
    \p2 = ($(hullcoord\nextnode) - (hullcoord\currentnode)$),
    \n2 = {atan2(\y2, \x2) + 90}, %
    \n{delta} = {Mod(\n2-\n1,360) - 360}
    in 
    {arc [start angle=\n1, delta angle=\n{delta}, radius=#2]}
    -- ($(hullcoord\nextnode)!#2!-90:(hullcoord\currentnode)$) 
  }
}

\newcommand{\customsubscript}{
  \tikz[baseline={(0,0)}]{
    \draw[scale=0.1] (.1,0) -- (1,1.5) -- (1.9,0) -- cycle;
    \draw[scale=0.1] (0, 1.6) -- (2,1.6);
    \draw[scale=0.1] (1, 1.6) -- (1,0);
  }
}

\newcommand{\customsubscriptt}{
  \tikz[baseline={(0,0)}]{
    \draw[scale=0.1] (.5,0) -- (1,1) -- (2,1) -- (2.5,0) -- cycle;
  }
}

\newcommand{\customsubscripttt}{
  \tikz[baseline={(0,0)}]{
    \draw[scale=0.1] (.5,0) -- (1,1) -- (4,1) -- (4.5,0) -- cycle;
    \draw[scale=0.1] (1.66, 1) -- (2.16,0);
    \draw[scale=0.1] (3.33, 1) -- (2.83,0);
  }
}

\newcommand{\customsubscriptttt}{
  \tikz[baseline={(0,0)}]{
    \draw[scale=0.1] (.1,0) -- (1,1.5) -- (1.9,0) -- cycle;
    \draw[scale=0.1] (.5, .75) -- (1.5,.75);
  }
}

\newcommand{\customsubscripttttt}{
  \tikz[baseline={(0,0)}]{
    \draw[scale=0.1] (.03,.93) -- (-2.03,.93);
    \draw (0,.1) arc[start angle=0, end angle=-180, radius=1mm];
    \draw (-.05,.1) arc[start angle=0, end angle=-180, radius=.5mm];
  }
}

\newcommand{\pentagon}{%
  \tikz[baseline={(0,0)}] \draw[scale=0.07] (90:1) -- (162:1) -- (234:1) -- (306:1) -- (18:1) -- cycle;
}
\newcommand{\QT}{Q_{\customsubscript}}
\newcommand{\QShort}{Q_{\customsubscriptt}}
\newcommand{\QLong}{Q_{\customsubscripttt}}
\newcommand{\QTri}{Q_{\customsubscriptttt}}
\newcommand{\QCirc}{Q_{\customsubscripttttt}}
 
\theoremstyle{plain}
\newtheorem{theorem}{Theorem}[section]
\newtheorem*{theorem*}{Theorem}

\newtheorem*{proposition*}{Proposition}
\newtheorem{lemma}[theorem]{Lemma}
\newtheorem*{lemma*}{Lemma}

\newtheorem*{corollary*}{Corollary}

\newtheorem*{claim*}{Claim}
\newtheorem{conjecture}[theorem]{Conjecture}
\newtheorem*{conjecture*}{Conjecture}

\theoremstyle{definition}
\newtheorem{definition}[theorem]{Definition}
\newtheorem*{definition*}{Definition}

\theoremstyle{remark}

\newtheorem*{observation*}{Observation}
\newtheorem{example}[theorem]{Example}
\newtheorem*{example*}{Example}

\newif\ifArxiv
\Arxivtrue

\begin{document}

\title{{The Space-Time Complexity of Sum-Product Queries}}

\author[1]{Kyle Deeds}
\author[2]{Timo Camillo Merkl}
\author[2]{Reinhard Pichler}
\author[1]{Dan Suciu}
\affil[1]{University of Washington, United States, \texttt{\{kdeeds,suciu\}@cs.washington.edu}}
\affil[2]{TU Wien, Austria, \texttt{\{timo.merkl,reinhard.pichler\}@tuwien.ac.at}}

\date{} 

\maketitle

\begin{abstract}
While extensive research on query evaluation
has achieved consistent improvements in the time complexity of algorithms, 
the space complexity of query evaluation has been largely ignored. 
This is a particular challenge in settings with strict pre-defined space constraints.
In this paper, we examine the combined space-time complexity of conjunctive queries (CQs) and, 
more generally, of sum-product queries (SPQs). We propose
several classes of space-efficient algorithms for evaluating SPQs, and 
we show that the optimal time complexity is almost always achievable with asymptotically lower space complexity than traditional approaches.
\end{abstract}

\section{Introduction}
\label{sec:introduction}

Algorithms for answering conjunctive queries (CQs), often generalized to sum-product queries (SPQs), 
have been extensively studied. Prior work has identified tight bounds on their \emph{time complexity} relative to a variety of structural parameters of the query, e.g. treewidth, (generalized or fractional) hypertree width, 
or submodular width~\cite{DBLP:journals/jacm/Grohe07,DBLP:journals/jcss/GottlobLS02,DBLP:journals/talg/GroheM14,DBLP:journals/jacm/Marx13}. However, no attention has been paid to the \emph{space complexity} of these algorithms which can often equal the time complexity.

This is a major challenge for end-users who typically run these algorithms in settings with strict pre-defined space constraints, e.g. GPU memory, main memory, or SSD size.  If the algorithm has a large space complexity, the user has two unsatisfactory options; 1) reserve a moderate amount of space and risk an out-of-memory error when inputs produce large intermediates or 2) reserve a larger, more expensive server to guarantee robustness. Developers typically place a high value on stability which pushes them towards the latter, and this conservative impulse is further exacerbated by the challenges of estimating space utilization ahead-of-time~\cite{DBLP:conf/cluster/BaderSLSWTK24}. In the cloud setting, this has resulted in the well-known problem of over-provisioning memory with over 90\% of jobs in the Google Cluster Dataset using less than 20\% of the provisioned memory~\cite{liu2025public, Google_2019}. 

In this paper, we examine the combined space-time complexity of SPQs to address these space-constrained settings -- illustrated here in the introduction for CQs. 
We begin by formally defining these notions of complexity. A {\em query plan} $\Pi$ for a query $Q$ is a structure that is associated with a specific algorithm for evaluating $Q$, e.g. a tree decomposition, or join-plan, or a variable order for generic-join.  We denote a class of plans by $\calC$, e.g. all tree decompositions, and the plans of $\calC$ for a query $Q$ as $\calC(Q)$. Each plan $\Pi$ is associated with a \emph{space exponent} $s(\Pi)$ and
a \emph{time exponent} $t(\Pi)$. The latter bounds 
the associated algorithm's runtime by $\tilde{O}(|D|^{t(\Pi)})$. 
The former bounds the 
space used by the algorithm by $O(|D|^{s(\Pi)})$, excluding the space required to store the input relations. We will often refer to these jointly as the space-time exponents $e(\Pi) = (s(\Pi), t(\Pi))$.

\begin{definition}{(Plan Domination).} \label{def:domination}
  Let $\Pi_1,\Pi_2$ be two plans for the same query $Q$. We say $\Pi_1$ \emph{improves over} $\Pi_2$, or that $\Pi_1$ {\em
    dominates} $\Pi_2$, denoted $\Pi_1 \preceq \Pi_2$, if
  $t(\Pi_1) \leq t(\Pi_2)$ and $s(\Pi_1) \leq s(\Pi_2)$.
  We say that $\Pi_1$ \emph{strictly dominates} $\Pi_2$,
  denoted $\Pi_1 \prec \Pi_2$, if $\Pi_1 \preceq \Pi_2$ but not
  vice versa.

  A class of plans $\calC_1$ \emph{improves over} (\emph{dominates}) another class
  $\calC_2$, denoted $\calC_1 \preceq \calC_2$, 
  if for every query $Q$, it holds that $\forall \,\Pi_2 \in \calC_2(Q)$, $\exists\, \Pi_1 \in \calC_1(Q)$
  such that $\Pi_1 \preceq \Pi_2$. Notice that $\preceq$ over classes of plans forms a preorder.
  We say that $\calC_1$ \emph{strictly dominates} $\calC_2$, denoted $\calC_1 \prec \calC_2$,
  if $\calC_1 \preceq \calC_2$ holds but not vice versa.  
\end{definition}

This paper studies, compares, and improves the space-time exponents of various query evaluation methods proposed in the literature, both for database queries and for probabilistic inference in graphical models.  As explained, we distinguish between the \emph{algorithm} used to evaluate the query, and the \emph{plan}, which is a syntactic structure (e.g. a tree, or a variable order), on which we can define simple measures (e.g. depth of a tree).  Every plan is canonically associated with an algorithm, but the set of plans for a given query is finite, while that of algorithms is infinite.  Our choice to distinguish these two notions may be less common in the theory community, but it is standard in database systems, where a ``plan'' refers to a relational algebra expression.  With this distinction in mind, let's examine what types of plans have been considered in the literature.

Prior work on conjunctive query answering has focused only on the time complexity. Standard {\em relational algebra plans}, described in all textbooks on database systems, are known to be suboptimal, hence we do not discuss them in this paper.  For a full CQ (query without projection), a \emph{worst-case optimal join} (WCOJ) algorithm runs in time $O(|D|^{\rho^*})$, where $\rho^*$ is the fractional edge cover number of $Q$; this time is proportional to the worst-case output size of the query~\cite{DBLP:conf/focs/AtseriasGM08}.  While the first WCOJ algorithms were first introduced in~\cite{DBLP:conf/icdt/Veldhuizen14,DBLP:conf/pods/NgoPRR12}, the best known variant is Generic Join (GJ)~\cite{DBLP:journals/sigmod/NgoRR13}.  A \emph{GJ plan} consists of a total order on the query variables, and the associated algorithm consists of nested \texttt{for}-loops, one for each variable.  Somewhat surprisingly, GJ can be proven to run in time $O(|D|^{\rho^*})$ independently of the plan\footnote{In practice, the choice of the plan (i.e.  of the variable order) makes a huge difference for the instance-specific runtime, see~\cite{DBLP:journals/pvldb/WangT0O23}, but we do not discuss instance-optimal algorithms in this paper.}.  GJ can easily be adapted to compute conjunctive queries with projections (i.e. non-full queries), however it is no longer guaranteed to be worst-case optimal.  The space required by a GJ plan consists of the space needed to store the iteration variables of the \texttt{for}-loops, which is $O(1)$ since each variable stores a single domain value, plus the space required to store the query's output. In the case of a Boolean query, the output also has size $O(1)$, therefore the space-time exponents of a GJ plan are $(0,\rho^*)$.  GJ is always space-optimal, but its time complexity is in general suboptimal for queries with projections.

Handling projections efficiently, and in particular handling Boolean queries, requires new techniques. All solutions proposed in the theoretical database community are based on tree decompositions~\cite{DBLP:conf/pods/GottlobLS99, DBLP:conf/pods/KhamisNR16, DBLP:conf/pods/GottlobGLS16, DBLP:journals/tods/AbergerLTNOR17, DBLP:conf/stoc/ArenasCJR21a, DBLP:conf/pods/KimH0H23, DBLP:conf/pods/BremenM23}.  As the name implies, a tree decomposition is a formalism for splitting the query into small, manageable sub-queries, called \emph{bags}, and composing these bags into a tree.  Execution proceeds by computing the result of each bag, then semi-joining the results bottom-up trough the tree decomposition.  The overall time complexity is given by the time required to solve every bag, which is generally referred to as the \emph{width} of the decomposition.  In the literature one finds different approaches on how to evaluate the bags, leading to various notions of width, such as \emph{tree width}, \emph{generalized hypertree width}, and \emph{fractional hypertree width} ($\fhw$)~\cite{DBLP:conf/pods/FischlGP18}.  A \emph{tree decomposition plan} consists of both a tree decomposition, and a choice of a plan for every bag of the tree decomposition.  When we choose generic join to compute the bags, then the space-time exponents of the tree decomposition plan are $(\fhw,\fhw)$.

Many more inference algorithms have been described in the field of probabilistic graphical models (PGMs); we direct the reader to~\cite{DBLP:series/synthesis/2019Dechter} for a comprehensive overview.  While PGM inference can be expressed as a scalar sum-product query studied in this paper, the runtime analysis of PGM inference algorithms differs from that done for query evaluation because the input data in PGMs is assumed to be dense, and the runtime is always expressed as an integer power of the domain size.  For example, in the case of a tree decomposition, the time exponent is the \emph{tree width}, instead of the (much smaller) fractional hypertree width used in the analysis of query evaluation. One of the contributions of our paper consists in adapting some of the PGM inference algorithms to query evaluation, and providing their runtime analysis.  The algorithms of interest to us here are the pseudo-tree based algorithm, and its refinement to caches and resets (called \emph{purges} in~\cite{DBLP:series/synthesis/2019Dechter}).

A \emph{Pseudo-Tree} (PT) for a query is a tree whose nodes are the query variables, such that the variables of every query atom are contained in some path from the root to a leaf (formal definition in Sec.~\ref{sec:pt}).  A PT is a plan for the query, and its natural algorithm consists of \texttt{for}-loops, whose nesting structure is given by the PT.  A Generic Join plan is a special case of a PT, where the tree consists of a single path, but, in general, a PT can have an improved time exponent because \texttt{for}-loops for independent variables can be executed sequentially, instead of nested.  The space required remains optimal and is $O(1)$ (plus the space required to store the output, as we discuss in this paper). 
The term \emph{pseudo-tree} was coined by Freuder and Quinn in the context of constraint optimization~\cite{DBLP:conf/ijcai/FreuderQ85}.

The runtime of a search algorithm can often be improved by the addition of a \emph{cache}.  Dechter~\cite{DBLP:series/synthesis/2019Dechter} adds a cache to each node of a pseudo-tree, leading to an improved time complexity, at the cost of using more space.  The cache associated to a query variable is a hash table, whose key consists of certain ancestor variables in the pseudo-tree (formal definition in Sec.~\ref{sec:ptc}).  To allow some tradeoff between the space and time complexity, Dechter describes a refinement, by which the size of a cache can be reduced by simply removing some of these ancestors from the hash table, and resetting the cache when their value changes (details Sec.~\ref{sec:ptcr}).  But no complexity analysis is provided for these techniques, even in the simplified complexity model of the probabilistic graphical models.

While no prior work has examined the end-to-end space-time complexity of CQ evaluation, two related lines of research should be acknowledged. For one, research on \textit{factorized databases} aims to create a space-efficient data structure from which the answers can be enumerated efficiently \cite{DBLP:journals/tods/OlteanuZ15}. Second, under the name \textit{conjunctive queries with access patterns}, prior work has explored how to materialize a space-efficient set of views to speed up subsequent query execution~\cite{DBLP:conf/pods/ZhaoDK23}.

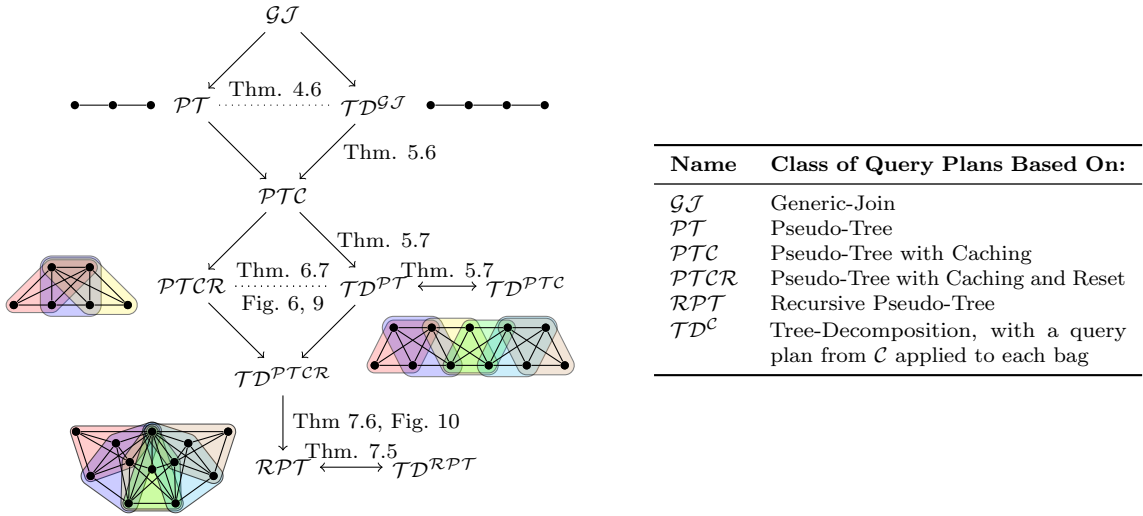
\begin{figure}
    \begin{minipage}{0.57\textwidth}
            \begin{center}
                \begin{tikzpicture}[scale = 0.8, font=\footnotesize]
                    \node (GJ) at (0,10) { \shortstack{ $\gj$ }};%
                    \node (TD[GJ]) at (1.5,8.5) {\shortstack{ $\tdgj$ }};%
                    \node (PT) at (-1.5,8.5) { \shortstack{$\pt$ }};%
                    \draw[->] (GJ)--(TD[GJ]) node[midway, above] {};
                    \draw[->] (GJ)--(PT) node[midway, above] {};
                    \node (PTC) at (0,7) { \shortstack{ $\ptc$ }};%
                    \draw[->] (TD[GJ])--(PTC) node[midway, right] {$\;$Thm. \ref{thm:tdgj-ptc}};
                    \draw[->] (PT)--(PTC) node[midway, above] {};
                    \node (PTCR) at (-1.5,5.5) { \shortstack{$\ptcr$ }};%
                    \node (TD[PT]) at (1.5,5.5) { \shortstack{$\tdpt$ }};%
                    \node (TD[PTC]) at (4,5.5) { \shortstack{$\tdptc$ }};%
                    \draw[->] (PTC)--(TD[PT]) node[midway, right] {Thm. \ref{thm:tdpt-tdptc}};
                    \draw[<->] (TD[PT])--(TD[PTC]) node[midway, above] {Thm. \ref{thm:tdpt-tdptc}};
                    \draw[->] (PTC)--(PTCR) node[midway, left] {};
                    \node (TD[PTCR]) at (0,4) {\shortstack{ $\tdptcr$ }};%
                    \draw[->] (TD[PT])--(TD[PTCR]) node[midway, above] {};
                    \draw[->] (PTCR)--(TD[PTCR]) node[midway, above] {};
                    \node (RPTCR) at (0,2.5) { \shortstack{$\rptcr$ }};%
                    \node (TD[RPTCR]) at (2.5,2.5) { \shortstack{$\tdrptcr$ }};%
                    \draw[->] (TD[PTCR])--(RPTCR) node[midway, right] {Thm \ref{thm:rptcr-tdptcr}, Fig.~\ref{fig:rptQuery}};
                    \draw[<->] (TD[RPTCR])--(RPTCR) node[midway, above] {Thm. \ref{thm:rptcr-tdrptcr}}; 
                    \draw[dotted] (PT)--(TD[GJ]) node[midway, above] {Thm. \ref{thm:tdgj-pt}}; 
                    \draw[dotted] (PTCR)--(TD[PT]) node[midway, above] {Thm. \ref{thm:ptcr-tdpt}}; 
                    \draw[dotted] (PTCR)--(TD[PT]) node[midway, below] {Fig.~\ref{fig:ptcrQuery},~\ref{fig:tdptQuery}}; 
                    \node[anchor=east] at ($(PT.west) + (0,0)$) (EXPTCR) {
                        \begin{tikzpicture}[scale=0.5]

                            \node[circle, fill=black, inner sep=1.1pt] (A) at (-1,0) {};
                            \node[circle, fill=black, inner sep=1.1pt] (C) at (0,0) {};
                            \node[circle, fill=black, inner sep=1.1pt] (E) at (1,0) {};

                            \draw (A) -- (C);
                            \draw (C) -- (E);
                        \end{tikzpicture}
                    };
                    \node[anchor=west] at ($(TD[GJ].east)$) (EXPTCR) {
                        \begin{tikzpicture}[scale=0.5]
                            \node (1) at (0,0) [circle, fill=black, inner sep=1.1pt] {};
                            \node (2) at (1,0) [circle, fill=black, inner sep=1.1pt] {};
                            \node (3) at (2,0) [circle, fill=black, inner sep=1.1pt] {};
                            \node (4) at (3,0) [circle, fill=black, inner sep=1.1pt] {};

                            \draw (1) -- (2);
                            \draw (2) -- (3);
                            \draw (3) -- (4);
                        \end{tikzpicture}
                    };
                    \node[anchor=east] at ($(PTCR.west)$) (EXPTCR) {
                        \begin{tikzpicture}[scale=0.5]
                            \node (1) at (0,0) [circle, fill=black, inner sep=1.1pt] {};
                            \node (2) at (1,0) [circle, fill=black, inner sep=1.1pt] {};
                            \node (3) at (2,0) [circle, fill=black, inner sep=1.1pt] {};
                            \node (4) at (3,0) [circle, fill=black, inner sep=1.1pt] {};
                            \node (5) at (1,1) [circle, fill=black, inner sep=1.1pt] {};
                            \node (6) at (2,1) [circle, fill=black, inner sep=1.1pt] {};

                            \draw[fill=red,opacity=0.5, fill opacity=0.2] \convexpath{1,5,6,2}{.2cm};
                            \draw[fill=blue,opacity=0.5, fill opacity=0.2] \convexpath{2,5,6,3}{.3cm};
                            \draw[fill=yellow,opacity=0.5, fill opacity=0.2] \convexpath{3,5,6,4}{.25cm};

                            \node (1) at (0,0) [circle, fill=black, inner sep=1.1pt] {};
                            \node (2) at (1,0) [circle, fill=black, inner sep=1.1pt] {};
                            \node (3) at (2,0) [circle, fill=black, inner sep=1.1pt] {};
                            \node (4) at (3,0) [circle, fill=black, inner sep=1.1pt] {};
                            \node (5) at (1,1) [circle, fill=black, inner sep=1.1pt] {};
                            \node (6) at (2,1) [circle, fill=black, inner sep=1.1pt] {};

                            \draw (1) -- (2);
                            \draw (2) -- (3);
                            \draw (3) -- (4);
                            \draw (1) -- (5);
                            \draw (2) -- (5);
                            \draw (3) -- (5);
                            \draw (1) -- (6);
                            \draw (2) -- (6);
                            \draw (3) -- (6);
                            \draw (5) -- (6);
                            \draw (4) -- (5);
                            \draw (4) -- (6);
                        \end{tikzpicture}
                    };
                    \node[anchor=west] at ($(TD[PT].east) + (-1,-1)$) (EXTDPT) {
                        \begin{tikzpicture}[scale=0.5]
                            \node[circle, fill=black, inner sep=1.1pt] (A) at (-1,0) {};
                            \node[circle, fill=black, inner sep=1.1pt] (B) at (-.5,1) {};
                            \node[circle, fill=black, inner sep=1.1pt] (C) at (0,0) {};
                            \node[circle, fill=black, inner sep=1.1pt] (D) at (.5,1) {};
                            \node[circle, fill=black, inner sep=1.1pt] (E) at (1,0) {};
                            \node[circle, fill=black, inner sep=1.1pt] (F) at (1.5,1) {};
                            \node[circle, fill=black, inner sep=1.1pt] (G) at (2,0) {};
                            \node[circle, fill=black, inner sep=1.1pt] (H) at (2.5,1) {};
                            \node[circle, fill=black, inner sep=1.1pt] (I) at (3,0) {};
                            \node[circle, fill=black, inner sep=1.1pt] (J) at (3.5,1) {};
                            \node[circle, fill=black, inner sep=1.1pt] (K) at (4,0) {};
                    
                            \draw[fill=red,opacity=0.5, fill opacity=0.2] \convexpath{A,B,D,C}{.2cm};
                            \draw[fill=blue,opacity=0.5, fill opacity=0.2] \convexpath{B,D,E,C}{.3cm};
                            \draw[fill=yellow,opacity=0.5, fill opacity=0.2] \convexpath{D,F,G,E}{.25cm};
                            \draw[fill=green,opacity=0.5, fill opacity=0.2] \convexpath{F,H,G,E}{.2cm};
                            \draw[fill=cyan,opacity=0.5, fill opacity=0.2] \convexpath{H,J,I,G}{.3cm};
                            \draw[fill=brown,opacity=0.5, fill opacity=0.2] \convexpath{H,J,K,I}{.25cm};

                            \node[circle, fill=black, inner sep=1.1pt] (A) at (-1,0) {};
                            \node[circle, fill=black, inner sep=1.1pt] (B) at (-.5,1) {};
                            \node[circle, fill=black, inner sep=1.1pt] (C) at (0,0) {};
                            \node[circle, fill=black, inner sep=1.1pt] (D) at (.5,1) {};
                            \node[circle, fill=black, inner sep=1.1pt] (E) at (1,0) {};
                            \node[circle, fill=black, inner sep=1.1pt] (F) at (1.5,1) {};
                            \node[circle, fill=black, inner sep=1.1pt] (G) at (2,0) {};
                            \node[circle, fill=black, inner sep=1.1pt] (H) at (2.5,1) {};
                            \node[circle, fill=black, inner sep=1.1pt] (I) at (3,0) {};
                            \node[circle, fill=black, inner sep=1.1pt] (J) at (3.5,1) {};
                            \node[circle, fill=black, inner sep=1.1pt] (K) at (4,0) {};

                            \draw (A) -- (B);
                            \draw (A) -- (C);
                            \draw (A) -- (D);
                            \draw (B) -- (C);
                            \draw (B) -- (D);
                            \draw (B) -- (E);
                            \draw (C) -- (D);
                            \draw (C) -- (E);
                            \draw (D) -- (E);
                            \draw (D) -- (F);
                            \draw (D) -- (G);
                            \draw (E) -- (F);
                            \draw (E) -- (G);
                            \draw (E) -- (H);
                            \draw (F) -- (G);
                            \draw (F) -- (H);
                            \draw (G) -- (H);
                            \draw (I) -- (G);
                            \draw (I) -- (H);
                            \draw (I) -- (J);
                            \draw (I) -- (K);
                            \draw (J) -- (G);
                            \draw (J) -- (H);
                            \draw (J) -- (K);
                            \draw (H) -- (K);
                        \end{tikzpicture}
                    };
                    \node[anchor=east] at ($(RPTCR.west)$) (EXRPT) {
                        \begin{tikzpicture}[scale=0.5]
                            \node[circle, fill=black, inner sep=1.1pt] (L) at (0,0) {};
                    
                            \foreach \angle/\label in {0/K, -36/I, -72/G, -108/E, -144/C, -180/A}
                            {
                                \node[circle, fill=black, inner sep=1.1pt] (\label) at (\angle:2) {};
                            }
                    
                            \foreach \angle/\label in {-18/J, -54/H, -90/F, -126/D, -162/B}
                            {
                                \node[circle, fill=black, inner sep=1.1pt] (\label) at (\angle:1) {};
                            }

                            \draw[fill=red,opacity=0.7, fill opacity=0.2] \convexpath{A,L,D,C}{.15cm};
                            \draw[fill=blue,opacity=0.7, fill opacity=0.2] \convexpath{B,L,E,C}{.25cm};
                            \draw[fill=yellow,opacity=0.7, fill opacity=0.2] \convexpath{D,L,G,E}{.2cm};
                            \draw[fill=green,opacity=0.7, fill opacity=0.2] \convexpath{L,H,G,E}{.15cm};
                            \draw[fill=cyan,opacity=0.7, fill opacity=0.2] \convexpath{L,J,I,G}{.25cm};
                            \draw[fill=brown,opacity=0.7, fill opacity=0.2] \convexpath{H,L,K,I}{.2cm};
                    
                            \node[circle, fill=black, inner sep=1.1pt] (L) at (0,0) {};
                    
                            \foreach \angle/\label in {0/K, -36/I, -72/G, -108/E, -144/C, -180/A}
                            {
                                \node[circle, fill=black, inner sep=1.1pt] (\label) at (\angle:2) {};
                            }
                    
                            \foreach \angle/\label in {-18/J, -54/H, -90/F, -126/D, -162/B}
                            {
                                \node[circle, fill=black, inner sep=1.1pt] (\label) at (\angle:1) {};
                            }

                            \draw (A) -- (B);
                            \draw (A) -- (C);
                            \draw (A) -- (D);
                            \draw (B) -- (C);
                            \draw (B) -- (D);
                            \draw (B) -- (E);
                            \draw (C) -- (D);
                            \draw (C) -- (E);
                            \draw (D) -- (E);
                            \draw (D) -- (F);
                            \draw (D) -- (G);
                            \draw (E) -- (F);
                            \draw (E) -- (G);
                            \draw (E) -- (H);
                            \draw (F) -- (G);
                            \draw (F) -- (H);
                            \draw (G) -- (H);
                            \draw (I) -- (G);
                            \draw (I) -- (H);
                            \draw (I) -- (J);
                            \draw (I) -- (K);
                            \draw (J) -- (G);
                            \draw (J) -- (H);
                            \draw (J) -- (K);
                            \draw (H) -- (K);
                            \draw (L) -- (A);
                            \draw (L) -- (B);
                            \draw (L) -- (C);
                            \draw (L) -- (D);
                            \draw (L) -- (E);
                            \draw (L) -- (F);
                            \draw (L) -- (G);
                            \draw (L) -- (H);
                            \draw (L) -- (I);
                            \draw (L) -- (J);
                            \draw (L) -- (K);
                        \end{tikzpicture}
                    }; 
                \end{tikzpicture}
            \end{center}
    \end{minipage}
    \hfill
    \begin{minipage}{0.42\textwidth}        
            \centering
            {\footnotesize
            \begin{tabular}{l p{0.7\textwidth}}
                \toprule
                \textbf{Name} & \textbf{Class of Query Plans Based On:} \\
                \midrule $\gj$ & Generic-Join\\
                $\pt$ & Pseudo-Tree\\
                $\ptc$ & Pseudo-Tree with Caching\\
                $\ptcr$ & Pseudo-Tree with Caching and Reset\\
                $\rptcr$ & Recursive Pseudo-Tree\\
                $\tdp$ & Tree-Decomposition, with a query plan from $\mathcal{C}$ applied to each bag\\
                \bottomrule
            \end{tabular}
            }
    \end{minipage}
    \caption{Classes of query plans, ordered by their space-time exponents: lower classes have smaller exponents and are better.  An arrow $\calC_2 \rightarrow \calC_1$ means that every plan in $\calC_2$ can be mapped to a plan in $\calC_1$ with space-time exponents at least as good; in particular, $\calC_1 \preceq \calC_2$ (see Def.~\ref{def:domination}).  A missing arrow, i.e.  $\calC_2 \not\rightarrow \calC_1$, means $\calC_1\not\preceq \calC_2$.  In particular, all downward arrows indicate strict domination (Cor.~\ref{cor:ptcr-tdpt-ptc}, Thm.~\ref{thm:rptcr-tdptcr}).  The depicted graphs are examples of scalar queries with binary predicates that separate the classes; the colors represent maximal cliques.  }
    \label{fig:plan-families}
\end{figure}

\paragraph{Our Contributions}
In this paper we study the space-time tradeoff of several classes of query plans, of increased sophistication: Generic Join and Pseudo-Trees (Sec.~\ref{sec:pt}), Tree Decomposition (Sec.~\ref{sec:td}), Pseudo-Trees with Caching (Sec.~\ref{sec:ptc}), Pseudo-Trees with Caching and Reset (Sec.~\ref{sec:ptcr}), and finally Recursive Pseudo-Trees (Sec.~\ref{sec:rpt}).  We fully characterize their domination relationships (Def.~\ref{def:domination}) and represent the resulting hierarchy in Fig.~\ref{fig:plan-families}: lower classes have smaller exponents, and thus are better.  We describe now our results in more detail, referring to this figure.

At the top of the figure is Generic Join ($\gj$), which is dominated by all other classes. From there, we generalize $\gj$ along two main axes. First, we consider pseudo-tree plans $\pt$.  These were originally introduced for constraint satisfaction problems, where they correspond to Boolean queries, or, more generally, to scalar sum-product queries.  We extend $\pt$'s to handle arbitrary outputs, and characterize their space-time complexity: unsurprisingly, they strictly dominate $\gj$.  We then revisit the plans based on tree decompositions, noticing that such a plan must consist of both the tree decomposition \emph{and} the plans used to compute each bag.  Thus, $\tdgj$ are tree decompositions plus generic join, while $\tdpt$ are tree decompositions plus pseudo-trees.  The figure shows that $\tdgj$ dominates $\gj$ (as expected); we discuss $\tdpt$ shortly.

Next we study the extension of $\pt$ with caches, $\ptc$; we slightly extend the original definition in~\cite{DBLP:series/synthesis/2019Dechter} by allowing caches to be added to any subset of the nodes of the 
pseudo-tree, instead of all nodes.  As the figure shows, $\ptc$ strictly dominate $\tdgj$.  This is somewhat surprising, because a tree decomposition allows a query to be computed ``in small pieces'', by computing one bag at a time, while a pseudo-tree consists of nested \texttt{for}-loops, requiring a global approach. Yet, by using caches, a pseudo-tree can simulate what a tree decomposition does, and, for some queries, strictly improve the space-time exponents.  However, if we use pseudo-trees instead of generic join to compute the bags of a tree decomposition, then the order reverses: $\tdpt$ strictly dominates $\ptc$.  Interestingly, if we try to improve tree decompositions by computing the bags using pseudo-trees \emph{with caches}, $\tdptc$, we don't gain any improvements over not using caches, $\tdpt$.

Next, we further refine pseudo-trees by allowing caches to be reset (and thus reduce their memory usage), and denote the resulting plans by $\ptcr$.  As explained, cache reset was already discussed in~\cite{DBLP:series/synthesis/2019Dechter}, but no complexity analysis was provided.  It turns out that computing the time complexity is more difficult in this case, because of the interaction between the various caches in the pseudo-tree.  Instead, we modified the algorithm in~\cite{DBLP:series/synthesis/2019Dechter}, thereby both improving its time complexity, and making the analysis possible.  As expected, adding resets improves the space-time complexity, and combining it with a tree decomposition, $\tdptcr$, further improves this complexity.

Lastly, we describe a new type of query plans to dominate them all, called Recursive Pseudo-Trees.  These appear to represent the best space-time tradeoff, because even by extending them with tree decompositions, the space-time exponents do not improve.

The reader may have noticed that all our results concern only upper bounds, and no lower bounds.  It turns out that very few space lower bounds are known in the literature, and none of them applies to the sum-product queries studied in this paper.  We discuss lower bounds in Sec.~\ref{sec:beyond}, were we also conjecture the space-time hardness of a specific query.

Further, all discussed methods aim at being more space efficient than existing tree decomposition based methods.
As such, naturally, the time exponent is always at least as large as the fractional hypertree width (fhw).
To go beyond this barrier, fundamentally different techniques are needed and are left as future work.
Some intricacies of this are hinted at in Sec.~\ref{sec:beyond}.

\ifArxiv 

We omitted proof details in the main body and defers them to the appendix. 
Additionally, the appendix contains further examples.
For some of the theorems, we need to prove the non-existence of certain structures. This is done by computer-assisted exhaustive enumeration. To that end, we developed software for computing the optimal time exponent of the \pt and \ptcr classes when given a space exponent (https://github.com/kylebd99/submodular-width).

\fi

\section{Preliminaries}
\label{sec:preliminaries}

Throughout this paper we fix an infinite domain $\dom$.  We denote (sets of) variables by capital letters $A,B,C,\dots$ ($\boldsymbol{X},\boldsymbol{Y},\boldsymbol{Z}, \dots$) and (tuples of) domain values by lowercase letters $a,b,c,\dots$ ($\boldsymbol{x},\boldsymbol{y},\boldsymbol{z}, \dots$).  We will also refer to variables as attributes when this is more appropriate.  If $\bm X, \bm Y$ are two sets of variables and $\bm x \in \dom^{\bm X}$, then we denote by $\bm x[\bm Y]$ the projection of $\bm x$ on the variables $\bm X \cap \bm Y$.

Fix a commutative semi-ring $(\K, \oplus, \otimes, \0, \1)$.  A $\K$-relation is a function $R\colon \dom^{\bm X} \rightarrow \K$ with finite support, meaning that $\supp(R):= \setof{\bm x}{R(\bm x)\neq \0}$ is finite.  When $\bm X=\emptyset$ then we identify $R$ with the scalar value $s := R() \in \K$.
By the cardinality of $R$, in notation $|R|$, we mean the cardinality of its support, and we write $\bm x \mapsto s \in R$ when $\bm x \in \supp(R)$ and $R(\bm x)=s$. A $\K$-database $D$, or simply a database when $\K$ is clear from the context, is a tuple of $\K$-relations $D=(R_1^D,\ldots,R_m^D)$.  Its size is $|D| := \sum_{i=1}^m|R^D_i|$.

A \emph{sum-product query} (SPQ or query for short) is an expression of the form:
\begin{align}
 Q(\boldsymbol{X})\leftarrow & \bigotimes_{i=1,\dots,m} R_i(\boldsymbol{X}_i) \label{eq:def:cq}
\end{align}
where $R_i$ are unique relation symbols and $\boldsymbol{X}, \boldsymbol{X}_1\dots$ are sets of variables such that $\boldsymbol{X}\subseteq \var(Q) := \bigcup_i \boldsymbol{X}_i$.  We refer to this query simply as $Q$, and call $\bm X$ the \emph{output variables} or \emph{head variables} of $Q$.  When $\bm X = \var(Q)$ then we say that $Q$ is a \emph{full query}, and when $\bm X= \emptyset$ then we call it a \emph{scalar query}.  
If the semi-ring $\K$ is the set of Booleans $\mathbb{B}$, then a scalar query is called a \emph{Boolean query}.

The semantics of $Q$ on a database $D$ is the $\K$-relation $\sem{Q}^D\colon \dom^{\boldsymbol{X}} \rightarrow \K$ defined as follows.  Let $\bm V = \var(Q)$ be the set of variables of the query $Q$ in Eq.~\eqref{eq:def:cq}.  Then:
\begin{align*}
  \sem{Q}^D(\bm x) := & \bigoplus_{\bm v \in \dom^{\bm V}\colon \bm v[\bm X]=\bm x} \left( \bigotimes_{i=1,\dots,m} R^D_i(\bm v[\bm X_i])\right)
\end{align*}
We may omit the superscript $D$ when it is clear from the context, and write $R_i$, $\sem{Q}$ for $R_i^D$, $\sem{Q}^D$.

\begin{example}
  \label{ex:3path} Most of our examples (in particular those used for separating families of query plans) will feature queries where all relations are binary, which allows for an intuitive representation as graphs. For instance, the 3-path on the top left of Fig.~\ref{fig:plan-families} represents the scalar query $Q() \leftarrow R_1(A,B) \otimes R_2(B,C)$, whose semantics is $\sum_{a,b,c \in \dom}R_1^D(a,b)\cdot R_2^D(b,c)$ over the natural numbers $\mathbb{N}$. 
\end{example}

Let $\bm Y \subseteq \var(Q)$ be a set of variables. A \emph{fractional edge cover} of $\bm Y$ (with respect to $Q$) is a sequence of non-negative weights $w_i$, one for each relation $R_i$, such that, for every variable $A \in \bm Y$, $\sum_{i: A \in \bm X_i}w_i \geq 1$.  The \emph{fractional edge cover number of $\bm Y$}, $\rho^*(\bm Y)$, is the minimum value of $\sum_i w_i$, when the weights $w_i$ range over fractional edge covers of $\bm Y$.  

\paragraph{Complexity analysis of algorithms.}
We assume the RAM model, where each cell can hold a single element from the domain or from the semi-ring.
As far as (normal and semi-ring) arithmetic is concerned, we assume that each operation requires constant time. 
For query evaluation, we only study the data complexity, i.e., the (size of the) query is considered as constant. 
Several of our algorithms assume a particular order of the tuples of the relations. We assume that 
sorting a  relation does not use additional space.  We will also ignore the additional log-factor due to sorting and/or lookups, and write our upper bounds using  $O$  instead of $\Tilde{O}$.

\section{Constant Space Query Evaluation}
\label{sec:pt}

In this section, we briefly review Generic Join and its space-time exponents.  Then, we discuss {\em pseudo-trees}~\cite{DBLP:conf/ijcai/FreuderQ85,DBLP:series/synthesis/2019Dechter} and extend them in two significant ways: we generalize them  to non-scalar queries, and analyze their space-time exponents using techniques  similar to those used for Generic Join.

\paragraph*{Generic Join} Consider a query $Q(\bm X)$ in~Eq.~\eqref{eq:def:cq}, with variables $\var(Q)=\set{A_1, \ldots, A_k}$.  The Generic Join (GJ) algorithm~\cite{DBLP:conf/pods/000118} computes $Q(D)$ in worst-case optimal time given by the AGM bound $O(|D|^{\rho^*(\var(Q))})$~\cite{DBLP:conf/focs/AtseriasGM08}.  GJ fixes an arbitrary order on the variables, $A_1,\dots A_k$, and computes iteratively partial assignments $\boldsymbol{y}_j = (a_1, \ldots, a_j)$ on $\boldsymbol{Y}_j=(A_1,\dots A_j)$, for all $0 \leq j \leq k$.  It starts with the empty assignment $\boldsymbol{y}_0:= ()$ and, in $k$ nested loops, it extends it to one variable after the other, as follows. Assuming a partial tuple $\bm y_{j-1}=(a_1, \ldots, a_{j-1})$, the $j$'th nested loop is:

\begin{align}
  & \textbf{for } a_j \textbf{ in } \bigcap_{i: A_j \in \bm X_i}\supp(R_i[A_j | \boldsymbol{y}_{j-1}]) \textbf{ do} \ldots \label{eq:gj:for}
\end{align}
where $\supp(R_i[A_j|\boldsymbol{y}_{j-1}])$ represents the $A_j$-values in the relation $R_i$, restricted to tuples that agree with the values\footnote{Formally, $\supp(R_i[A_j|\boldsymbol{y}_{j-1}]) = \setof{\bm z[A_j]}{\bm z \in \supp(R_i), \bm z[\boldsymbol{Y}_{j-1}]=\boldsymbol{y}_{j-1}[\bm X_i]}$.} $\boldsymbol{y}_{j-1}$.  GJ computes the intersection above in time proportional to the smallest set, for example by iterating over the smallest set and probing (using hash tables) in all the other sets.
If $Q$ is a full SPQ, then in the inner-most loop, GJ simply outputs the assignment $\boldsymbol{y}_k \mapsto s$, where $s := \bigotimes_{i}R_i(\boldsymbol{y}_k[\bm X_i]) \in \K$.  If $Q$ is a scalar query, then the innermost loop computes the sum of these values $s \in \K$.  In the general case, $\emptyset \subseteq \bm X \subseteq \var(Q)$, GJ  maintains a hash-table $OUT\colon \dom^{\bm X} \rightarrow \K$ to store the current output, and the inner-most loop updates $OUT(\bm y_k[\bm X]) := OUT(\bm y_k[\bm X])\oplus s$.  The time complexity of GJ is $O(|D|^{\rho^*(\var(Q))})$.  Its space complexity, i.e. the memory size required in addition to the input database, consists of the $k$ variables $a_1, \ldots, a_k$, with a total size of $O(1)$, plus the space required to store the output $OUT$, whose size is $O(|D|^{\rho^*(\bm X)})$.

\begin{definition}
  \label{def:gj:plans} A \emph{Generic Join Plan} of a query $Q$ is a total order on its variables, $\Pi = (A_1, \ldots, A_k)$.  We denote by $\gj$ (resp. $\gj(Q)$) the set of all Generic Join plans (of $Q$).  The space-time exponents of any $\Pi \in \gj(Q)$ are $e(\Pi) = (\rho^*(\boldsymbol{X}), \rho^*(\var(Q)))$; in particular, if $Q$ is a scalar query, then the space-time exponents are $(0,\rho^*(\var(Q)))$.
\end{definition}

\begin{theorem}[\cite{DBLP:journals/sigmod/NgoRR13}] \label{th:generic:join}
  If $\Pi \in \gj(Q)$, then $\Pi$ computes $\sem{Q}$ in space and time given by $e(\Pi)$.  Concretely, the space used is $O(|D|^{\rho^*(\bm X)})$, and the time spent is $O(|D|^{\rho^*(\var(Q))})$.
\end{theorem}

When the cardinalities of individual relations are known, $N_1=|R_1|, N_2=|R_2|, \ldots$, then tighter space-time bounds are given by $O(\prod_i N_i^{w_i^*})$, where $w_1^*, w_2^*, \ldots$ is the fractional edge cover of $\bm X$ (or $\var(Q)$ respectively) that minimizes $\prod_i N_i^{w_i^*}$.  In the special case when $N_1=N_2=\cdots = N$ this is equal to $O(|D|^{\rho^*})$.  In this paper we prefer to use the simpler formula, and will note where the tighter formula is needed.

\paragraph*{Pseudo-trees} For a \emph{full} query, any GJ plan is worst-case optimal, because the output size can be as large as $O(|D|^{\rho^*(\var(Q))})$~\cite{DBLP:conf/focs/AtseriasGM08}. But when $Q$ is not full, GJ is no longer optimal.  A pseudo-tree, defined below, improves the time exponent of GJ, without increasing its space exponent. 

\begin{example} \label{ex:query:3path}
  For a simple intuition, consider the 3-path scalar query in Example~\ref{ex:3path}.  GJ computes it using 3 nested loops, corresponding to the variables $A, B, C$; intuitively it computes the expression $\sum_a \sum_b \sum_c R_1(a,b) R_2(b,c)$, with runtime $O(|R_1|\cdot|R_2|)$.  A pseudo-tree based algorithm, in contrast, iterates over $B$ first, then performs two independent loops that iterate over $A$ and $C$ respectively; this corresponds to the expression $\sum_b\left(\sum_a R_1(a,b)\right)\cdot\left(\sum_c R_2(b,c)\right)$, and the runtime is $O(|R_1|+|R_2|)$.
\end{example}

If $T=(\bm V, E)$ is a directed tree and $A \in \bm V$, then we denote by $\anc(A)$ the set of ancestors of $A$ excluding $A$, and write $\ancc(A) = \anc(A) \cup \set{A}$.  Similarly, we write $\desc(A), \descc(A)$ for the set of descendants of $A$, without and with $A$ respectively. 

Fix an SPQ $Q(\boldsymbol{X}) \leftarrow \bigotimes_i R_i(\boldsymbol{X}_i)$ (see Eq.~\eqref{eq:def:cq}), and let $\bm V = \var(Q)$.  

\begin{definition}[\cite{DBLP:conf/ijcai/FreuderQ85,DBLP:series/synthesis/2019Dechter}]
\label{def:pt}
A {\em pseudo-tree} (PT) of $Q$ is a directed tree $P=(\boldsymbol{V},E)$, satisfying:
\begin{itemize}
\item Every atom $R_i(\boldsymbol{X}_i)$ is contained in a branch:  formally, $\exists A\in \bm V$ such that $\boldsymbol{X}_i\subseteq \ancc(A)$.
\end{itemize}

The term \emph{pseudo-tree} was introduce by Freuder and Quinn in the context of constraint optimization~\cite{DBLP:conf/ijcai/FreuderQ85}, and studied extensively by Dechter~\cite{DBLP:series/synthesis/2019Dechter}.  Pseudo-trees are a generalization of \emph{normal trees}, also called \emph{Tr\'emaux trees} used in graph theory~\cite[Ch.1]{DBLP:books/daglib/0030488}, which are pseudo-trees where every tree edge is also an edge in the graph.
\end{definition}

For any variable $A \in \bm V$, we denote by $\out(A) := \desc(A) \cap \bm X$ and by $\outt(A) := \descc(A) \cap \bm X$.

\begin{definition} \label{def:space:time:pt} 
The class of query plans $\pt(Q)$ consists of pseudo-trees $P$ of $Q$ and their space and time exponents are defined as:
  $$s(P) := \max_{A\in \var(Q)}\rho^*(\outt(A)),\quad t(P) := \max_{A\in \var(Q)}\rho^*(\ancc(A) \cup \out(A)).$$
\end{definition}

When $Q$ is a scalar query, then the space exponent is $s(P)=0$.  Pseudo-trees strictly dominate GJ, i.e. $\pt \prec \gj$ as per Def.~\ref{def:domination} because, any variable order of a GJ can be converted into a linear \ptree $A_1 - A_2 - \cdots - A_k$, and the two plans have the same space-time complexity thus $\pt \preceq \gj$.  On the other hand, Example~\ref{ex:query:3path} shows that $\gj \not\preceq \pt$.  This establishes the first arrow in Fig.~\ref{fig:plan-families}.

It remains to describe an algorithm that, given a pseudo-tree $P$ for $Q$, computes $\sem{Q}$ in space-time given by the exponents in Def.~\ref{def:space:time:pt}.  First we need to introduce some notations.  If $R(\bm X) : \dom^{\bm X} \rightarrow \K$ is a $\K$-relation and $A \in \bm X$, then we write $\supp(R[A]) := \setof{\bm x[A]}{\bm x \in \supp(R)}$ for the projection of $\supp(R)$ on the attribute $A$; in other words, this is the $A$-column of $R$.  We generalize this as follows.  Let $\bm Y$ be a set of variables s.t. $A \not\in \bm Y$, and $\bm y \in \dom^{\bm Y}$. We write $\supp(R[A|\bm y])$ for the projection on $A$ of the  tuples in $\supp(R)$ that agree with $\bm y$:  $\supp(R[A|\bm y]):=\setof{\bm x[A]}{\bm x \in \supp(R), \bm x[\bm Y]=\bm y[\bm X]}$.

Algorithm~\ref{alg:pt} (ignore the gray lines for now) evaluates $Q(\bm X)$ recursively, by following the structure of the pseudo-tree $P$.  We start from the root $A_1$, and proceed recursively in the tree.  Assume we have followed a path $\bm Y = (A_1, A_2, \ldots, A_{j-1})$, and have bound these variables to the tuple $\bm y \in \dom^{\bm Y}$.  For a child $A$ of $A_{j-1}$, $\textsc{solve}(A,\bm y)$ first computes the following $\K$-relation $\sem{Q[A|\bm y]}: \dom^A \rightarrow \K$:

\newcommand{\refalgptc}{3\xspace}

\newcommand{\TypedFunction}[3]{\textbf{function} \textsc{#1}(#2)$\colon$#3}

\begin{figure}
\begin{minipage}{0.49\textwidth}
\begin{algorithm}[H]
    \caption{ PT Algorithm (excl. gray parts)
    \textcolor{gray}{\textbf{Algorithm \refalgptc} PTC Algorithm (incl. gray parts)}}
    \label{alg:pt}
    \begin{algorithmic}[1] %
      \Statex \textbf{Input:} $\text{Query } Q(\bm X), \text{pseudo-tree } P,$
      \Statex \hspace{10mm} $\textcolor{gray}{\text{caches } \boldsymbol{C}\subseteq \var(Q)}$
        \Statex \textbf{Output:} $\sem{Q}\colon \dom^X \rightarrow \K$
        \State \Return \Call{solve}{$\Root(P),()$}
                
    \algrenewcommand\algorithmicfor{\textcolor{gray}{\textbf{for}}}
        \For{\textcolor{gray}{$A\in \boldsymbol{C}$}}
            \State \textcolor{gray}{$M_A \gets \emptyset$} 
        \EndFor  
    \algrenewcommand\algorithmicfor{\textcolor{black}{\textbf{for}}}
        \State \TypedFunction{solve}{$A,\boldsymbol{y}$}{$\dom^{\outt(A)}\rightarrow \K$}
    \algrenewcommand\algorithmicif{\textcolor{gray}{\textbf{if}}}
        \If{\textcolor{gray}{$A\in \boldsymbol{C} \land {\boldsymbol{y}}[\con(A)]\in \keys{M_A}$}}
    \algrenewcommand\algorithmicreturn{\textcolor{gray}{\textbf{return}}}
            \State \Return \textcolor{gray}{$M_A({\boldsymbol{y}}[\con(A)])$}
    \algrenewcommand\algorithmicreturn{\textcolor{black}{\textbf{return}}}
        \EndIf
    \algrenewcommand\algorithmicif{\textcolor{black}{\textbf{if}}}
        \State OUT $\gets \{\boldsymbol{z} \mapsto \0 \mid \boldsymbol{z}\in\dom^{\outt(A)}\}$ \label{line:ptc:miss}
        \For{$a\mapsto s \in \sem{Q[A|\boldsymbol{y}]}$} \label{line:for:a}
        \Comment{see Eq.~\eqref{eq:qay}}
            \State $\boldsymbol{y}'\gets (\boldsymbol{y}, a)$
            \State TMP $\gets \{a[\bm X\cap \{A\}] \mapsto s\}$
            \For{$B\in child(A)$}
                \State TMP $\gets$ TMP $\otimes$ \Call{solve}{$B,\boldsymbol{y}'$} \label{line:pt:recurse}
            \EndFor 
            \State OUT $\gets$ OUT $\oplus$ TMP 
        \EndFor
    \algrenewcommand\algorithmicif{\textcolor{gray}{\textbf{if}}}
        \If{\textcolor{gray}{$A\in \boldsymbol{C}$}}
            \State \textcolor{gray}{$M_A \gets M_A \cup \{{\boldsymbol{y}}[\con(A)] \mapsto $ OUT $\}$}
        \EndIf
    \algrenewcommand\algorithmicif{\textcolor{black}{\textbf{if}}}
        \State \Return OUT
    \end{algorithmic}
\end{algorithm}
\vspace{-1em}

\end{minipage}
\hfill
\begin{minipage}{0.5\textwidth}
\begin{algorithm}[H]
    \begin{algorithmic}[1] %
    \State $\text{OUT}^{B}\gets 0$
    \For{$b \mapsto s^b \in \sem{Q[B]}$}  \Comment{$ s^b= 1$}\\
            \Comment{$b\in R_1[B]\cap \dots \cap R_5[B]$}
        \State $\text{OUT}^{A}\gets 0$
        \For{$a \mapsto s^a\in \sem{Q[A|b]}$}         \Comment{$s^a=R_1(a,b)$}
            \State $\text{OUT}^{A}\gets \text{OUT}^{A} + s^a$
        \EndFor
        \State $s^b \gets s^b \cdot \text{OUT}^A, \text{OUT}^{E}\gets 0$
        \For{$e\mapsto s^e\in \sem{Q[E|b]}$} \Comment{$s^e=R_4(b,e)$}\\
                \Comment{$e\in R_4[E|b]\cap R_6[E]\cap R_7[E]$} %
            \State $\text{OUT}^{D}\gets 0$
            \For{$d\mapsto s^d\in \sem{Q[D|be]}$} \\
                    \Comment{$s^d=R_3(b,d)\cdot R_6(e,d)$}
                \State $\text{OUT}^{D}\gets \text{OUT}^{D} + s^d$
            \EndFor
            \State $s^e\gets s^e \cdot \text{OUT}^D, \text{OUT}^{F}\gets 0$
            \For{$f\mapsto s^f\in \sem{Q[F|be]}$} %
                \State $\text{OUT}^{F}\gets \text{OUT}^{F} + s^f$
            \EndFor
            \State $s^e \gets s^e \cdot \text{OUT}^F$,  $\text{OUT}^{E}\gets \text{OUT}^{E} + s^e$
        \EndFor
        \State $s^b \gets s^b \cdot \text{OUT}^E, \text{OUT}^{C}\gets 0$
        \For{$c\mapsto s^c\in \sem{Q[C|b]}$} %
            \State $\text{OUT}^{C}\gets \text{OUT}^{C} + s^c$
        \EndFor
        \State $s^b \gets s^b \cdot \text{OUT}^C$, $\text{OUT}^{B}\gets \text{OUT}^{B} + s^b$
    \EndFor
    \State \Return $\text{OUT}^{B}$
    \end{algorithmic}
    \caption{Alg. \ref{alg:pt} for the \ptree in Fig.~\ref{fig:pt}}
    \label{alg:loopsPt}
\end{algorithm}
\end{minipage}

\end{figure}

\addtocounter{algorithm}{1} 
\begin{align}
  \sem{Q[A|\bm y]} := & \bigsetof{a \mapsto \bigotimes_{i: \bm X_i \setminus \bm Y = \set{A}}R_i(\bm y'[\bm X_i])}{a \in \bigcap_{i: A \in \bm X_i}\supp(R_i[A|\bm y]), \bm y' = (\bm y, a)}\label{eq:qay}
\end{align}

Intuitively, $\sem{Q[A|\bm y]}$ contains the possible values of $A$ that extend $\bm y$ in a manner consistent with~$Q$. 
To compute~\eqref{eq:qay}, the algorithm iterates over all values $a \in \bigcap_i \supp(R_i[A|\bm y])$, by intersecting the $A$-attributes of all relations that contain the attribute $A$ (similarly to Generic Join in Eq.~\eqref{eq:gj:for}), then maps each such value $a$ to $s \in \K$, where $s$ is the product of all $\K$-values of the relations whose last attribute (in the order of the pseudo-tree) is $A$: the condition $\bm X_i \setminus \bm Y = \set{A}$ checks that $A$ is the last attribute of $R_i(\bm X_i)$, while $R_i(\bm y'[X_i]) \in \K$ is its value associated to $\bm y' := (\bm y, a)$.  Like GJ, the algorithm computes the intersection of $\supp(R_i[A|\bm y])$ on the fly in time proportional to the smallest set.  That is, the algorithm iterates over the values $a\mapsto s$ in $\sem{Q[A|\bm y]}$, and performs a recursive call on each child $B$ of $A$.  When $Q$ is a scalar query, then both TMP and OUT are scalars (because $\outt(A)=\emptyset$), and the algorithm simply multiplies the values of all children $B$, then adds up these values over all $a$'s.  When $Q$ has output variables $\bm X$, then both OUT and TMP are $\K$-relations.  OUT has type $\dom^{\outt(A)}\rightarrow \K$.  Initially, TMP has attributes $\bm X \cap \set{A}$ (i.e. either $\emptyset$ or $\set{A}$), while the natural join\footnote{The natural join $R_1 \otimes R_2$ of relations $R_i\colon \dom^{\bm U_i} \rightarrow \K$ is defined as $(R_1 \otimes R_2)(\bm u) := R_1(\bm u[\bm U_1]) \otimes R_2(\bm u[\bm U_2])$ where $u\in \dom^{\bm U_1 \cup \bm U_2}$. Note the schemas of TMP and $\textsc{solve}(B,\boldsymbol{y}')$ are disjoint and, hence, it degenerates to a Cartesian product.}  $TMP \leftarrow TMP \otimes \textsc{solve}(B,\bm y')$ extends TMP with $\outt(B)$, so that, after processing all children $B$ of $A$, TMP has the same schema as OUT, and the algorithm adds up these values over all $a$'s.  
\ifArxiv
We prove the following in the appendix:
\fi

\begin{restatable}{theorem}{thmalgopt}    
\label{thm:algo:pt}
If $P\in \pt(Q)$, then Algorithm \ref{alg:pt} computes $\sem{Q}$ in time $ {O}(|D|^{t(P)})$ and uses $O(|D|^{s(P)})$ space (where $s,t$ are given by Def.~\ref{def:space:time:pt}).
\end{restatable}

\begin{example}
\label{ex:ptalgo}
Consider the scalar query $\QT()$ depicted in Fig.~$\ref{fig:ptQuery}$, over 7 $\N$-relations.  Its fractional edge cover number is $\rho^*:=\rho^*(\var(Q))=4$.  Generic Join computes the query using 6 nested loops, one for each variable $A,B,\ldots,F$, and its runtime is $O(|D|^4)$.  Consider now the \ptree $P$ in Fig.~\ref{fig:pt}.  We check that it satisfies the condition in Def.~\ref{def:pt}: indeed, each relation is included in a branch, e.g., relation $R_5(B,F)$ is included in the branch $B-E-F$.  Its space-time exponents are $e(P)=(s(P),t(P))=(0,\nicefrac{3}{2})$; indeed, $s(P)=\rho^*(\emptyset)=0$ because $\QT$ is a scalar query, and $t(P)=\max(\rho^*(\ancc(A)), \rho^*(\ancc(B)), \rho^*(\ancc(C)), \rho^*(\ancc(D)), \dots) = \max(1,1,1,\nicefrac{3}{2}, 1, \nicefrac{3}{2}) = \nicefrac{3}{2}$.  Algorithm~\ref{alg:loopsPt} (the expansion of Algorithm~\ref{alg:pt}) computes $\QT$ following the \ptree.  It starts with a loop for the variable $B$, but, unlike Generic Join, it continues with three independent loops, one for each variable $A,E,C$ respectively.  The for-loop for $E$ contains another two, independent for-loops, for $D$ and $F$.  Notice that the first loop, for $B$, is over the intersection of the attributes $B$ of all relations that contain $B$ (to reduce clutter we omitted $\supp$ in the comment), while the associated value $s^b$ is $1$, because there are no relations that ``end'' at $B$.  On the other hand, the loop for $A$ associates to $s^a$ the value $R_1(a,b)$, because $A$ is the last attribute of the relation $R_1(A,B)$.  Similarly, the loop for $e$ is over the intersection of the $E$-columns $R_4, R_6, R_7$ (where $R_4$ is restricted to the value $b$), and the value we associate to $e$ is $s^e=R_4(b,e)$, because $R_4(B,E)$ is the only relation that ends at $E$.  The runtime of Algorithm~\ref{alg:loopsPt} is $O(|D|^{\nicefrac{3}{2}})$, because it is dominated by the nested loops $B-E-D$ and $B-E-F$, each requiring only $O(|D|^{\nicefrac{3}{2}})$ steps.  Since $\QT$ is scalar, both TMP and OUT are scalars (to reduce clutter we did not include the variables TMP but used $s$ instead).  Consider what happens if we modify the query to have output variables $\bm X = \set{D,F}$.  Then, $\text{OUT}^A, \text{OUT}^C$ are still scalars while $\text{OUT}^D$ is a $\K$-relation with attribute $D$, $\text{OUT}^F$ has the attribute $F$, and $\text{OUT}^B, \text{OUT}^E$ have attributes $DF$.  Further, the space-time exponents become $(\rho^*(DF),\rho^*(BEDF))=(2,2)$ computed, e.g., at $E$.

\begin{figure}

\begin{minipage}[b]{0.23\textwidth}
\begin{figure}[H]
    \centering
    \begin{tikzpicture}[scale = .5,font=\footnotesize]
            \node (A) at (-2,2) {$A$};
            \node (B) at (0,2) {$B$};
            \node (C) at (2,2) {$C$};
            \node (D) at (-2,0) {$D$};
            \node (E) at (0,0) {$E$};
            \node (F) at (2,0) {$F$};

            \draw (A) -- (B) node[midway, above] {$R_1$};
            \draw (B) -- (C) node[midway, above] {$R_2$};
            \draw (B) -- (D) node[midway, left] {$R_3$};
            \draw (B) -- (E) node[midway, right] {$R_4$};
            \draw (B) -- (F) node[midway, right] {$R_5$};
            \draw (D) -- (E) node[midway, below] {$R_6$};
            \draw (E) -- (F) node[midway, below] {$R_7$};

    \end{tikzpicture}
    \caption[]{Query $\QT$}
    \label{fig:ptQuery}
\end{figure}
\end{minipage}
\hfill
\begin{minipage}[b]{0.26\textwidth}
\begin{figure}[H]
    \centering
    \begin{tikzpicture}[scale = .5,font=\footnotesize]
    \node {$B$}
        child {node {$A$}}
        child { node {$E$}
            child {node {$D$}}
            child {node {$F$}}
        }
        child {node {$C$}};                    
    \end{tikzpicture}
    \caption[]{Pseudo-tree for $\QT$}
    \label{fig:pt}
\end{figure}
\end{minipage}
\hfill
\begin{minipage}[b]{0.2\textwidth}
    \begin{figure}[H]
        \centering
            \begin{tikzpicture}[scale = 0.35,font=\footnotesize]
                \node (A) at (4,2) {$A$};
                \node (B) at (6,2) {$B$};
                \node (F) at (5,4) {$F$};
                \node (E) at (3,0) {$E$};
                \node (D) at (5,0) {$D$};
                \node (C) at (7,0) {$C$};

                \draw (A) -- (B);
                \draw (A) -- (C);
                \draw (A) -- (D);
                \draw (A) -- (E);
                \draw (A) -- (F);
                \draw (B) -- (C);
                \draw (B) -- (D);
                \draw (B) -- (E);
                \draw (B) -- (F);
                \draw (E) -- (D);
                \draw (D) -- (C);
            \end{tikzpicture}
        \caption[]{Query $\QTri$}
        \label{fig:tdptQ}
    \end{figure}
\end{minipage}
\hfill
\begin{minipage}[b]{0.27\textwidth}
\begin{figure}[H]
\centering
\begin{center}
    \begin{tikzpicture}[scale = 0.9, font=\footnotesize]
        
        \node[rectangle, draw] (ABCDE) at (0,0) {
        \begin{tikzpicture}[scale=0.3]
            \node[inner sep=1.5pt] {$A$}
                child { node[inner sep=1.5pt] {$B$}
                    child { node[inner sep=1.5pt] {$D$}
                        child {node[inner sep=1.5pt] {$E$}}
                        child {node[inner sep=1.5pt] {$C$}}
                    }
                };
        \end{tikzpicture}
        };
        \node[rectangle, draw] (DEFGH) at (2,.3) {
        \begin{tikzpicture}[scale=0.3]
            \node[inner sep=1.5pt] {$B$}
                child { node[inner sep=1.5pt] {$A$}
                    child { node[inner sep=1.5pt] {$F$}
                    }
                };
        \end{tikzpicture}
        };

        \draw (ABCDE) -- (DEFGH) node[midway, above] {$AB$};
    \end{tikzpicture}
\end{center} 
\caption[]{$\tdpt$ plan for $\QTri$}
    \label{fig:tdptex}
\end{figure}
\end{minipage}
\end{figure}
 
\end{example}

\section{Revisiting Tree Decompositions}
\label{sec:td}

Tree decompositions (TDs) have been extensively studied in the literature~\cite{DBLP:journals/jacm/Grohe07,DBLP:journals/jcss/GottlobLS02,DBLP:journals/talg/GroheM14,DBLP:journals/jacm/Marx13}.  Like a \ptree, a TD may decrease the time exponent of generic join, but it may increase the space exponent.  We briefly review the definition of TDs , then show that they are incomparable to \ptrees.

In this section we fix a query $Q(\boldsymbol{X})\leftarrow \bigotimes_i R_i(\boldsymbol{X}_i)$, as in Eq.~\eqref{eq:def:cq}.

\begin{definition} \label{def:td}
  A tree decomposition of $Q(\boldsymbol{X})$ is a tuple $TD = (T, \chi)$ where $T=(V,E)$ is a directed
  tree and $\chi: V\rightarrow 2^{\var(Q)}$ is a function from the nodes to sets of variables, satisfying:
    \begin{enumerate}
    \item \label{item:td:root} $\bm X\subseteq \chi(\Root(T))$,
    \item $\forall A\in \var(Q)$, the nodes $v$ with $A \in \chi(v)$
      must  form a connected subset of $V$,
    \item $\forall \boldsymbol{X}_i \exists v\in V \text{ s.t. } \boldsymbol{X}_i\subseteq \chi(v)$. We pick an arbitrary such $v$ and say $R_i(\bm X_i)$ is covered by $v$.
    \end{enumerate}
\end{definition}

The sets $\chi(v)$ are called \emph{bags}.  Readers familiar with free-connex tree decomposition may notice that condition (\ref{item:td:root}) is more restrictive, but that's OK for our purpose, because we only consider worst-case optimal algorithms, and do not consider constant-delay algorithms~\cite{DBLP:conf/csl/BaganDG07}.  Any free-connex tree decomposition can be converted into a tree that satisfies Def.~\ref{def:td}, by adding a bag with all output variables $\bm X$, without increasing its worst-case total runtime.

To each tree vertex $v \in V$ we associate a query, as follows.  Let $\bm Y^v := \chi(v)$ and $\bm Z^v := \chi(v) \cap \chi(\parent(v))$ (when $v$ is the root node, then $\bm Z^v := \bm X$) and define the $\K$-relation $R_i^v : \dom^{\bm X_i \cap \bm Y^v}\rightarrow \K$ as $R_i^v := R_i$ when $R_i$ is covered by $v$, and $R_i^v(\bm z) := \left(\{\boldsymbol{z}\mapsto \1 \mid \exists \boldsymbol{x}\in \supp(R_i) \text{ s.t.\ } \boldsymbol{z} = \boldsymbol{x}[\boldsymbol{Y}^v]\}\right)$ otherwise.  In other words, for all $v\in V$ but one, the relation $R^v_i$ is a $\set{\bm 0, \bm 1}$-relation consisting of the projection of $\supp(R_i)$ on the variables $\bm X_i \cap \bm Y^v$.  Define the \emph{sub-query} $Q^v$ at node $v$ as:
\begin{align}
  Q^v(\bm Z^v) \leftarrow & \bigotimes_i R_i^v(\bm X_i \cap \bm Y^v) \otimes  \bigotimes_{w \in \children(v)}Q^w(\bm Z^w)\label{eq:sub:query}
\end{align}
Notice that $\var(Q^v)=\chi(v)=\bigcup_i\var(R_i^v)=\bm Y^v$, since $\boldsymbol{Z}^w \subseteq \chi(v)$ for all $w \in \children(v)$.

We use the TD to compute $Q(\bm X)$ by computing all subqueries $Q^v$, bottom-up.
For each $v \in V$, first compute recursively the subqueries $Q^w(\bm Z^w)$ of its children $w$, materialize these results, then compute $Q^v$ as in Eq.~\eqref{eq:sub:query}.  The materialized results $Q^w(\bm Z^w)$ are called \emph{messages} in the literature.

It remains to decide what plan we use to compute the subqueries $Q^v$. 
This justifies the following:

\begin{definition}%
\label{def:time:space:td}
Let $\calC$ be a class of query plans for evaluating SPQs.  The class of plans $\tdc(Q)$ consists 
of pairs $(TD,\pi)$, where $TD$ is a tree decomposition of $Q$, and $\pi$ is a function that maps each vertex $v$ in $TD$ to a plan $\pi^v \in \calC(Q^v)$.  The space and time exponents are
  \begin{align}
      s(TD, \pi) = \max_v s(\pi^v),\quad t(TD, \pi) = \max_v t(\pi^v). \label{eq:tdst}
  \end{align}
\end{definition}

\begin{restatable}{theorem}{thmalgotdc}    
  \label{thm:algo:tdc} Let $s, t$ be the space and time exponents of the plans $\calC$, such that for every query $Q'$ and plan $\pi \in \calC(Q')$, $\sem{Q'}$ can be computed in space $O(|D|^{s(\pi)})$ and time $O(|D|^{t(\pi)})$.  Then, for every plan $(TD,\pi) \in\tdc(Q)$, $\sem{Q}$ can be computed in space $O(|D|^{s(TD,\pi)})$ and time $O(|D|^{t(TD,\pi)})$, where $s(TD,\pi)$ and $t(TD,\pi)$ are given by Eq.~\eqref{eq:tdst}.
\end{restatable}

Two remarks are in order.  First, we notice that the space needed to store the message $Q^v$ that is sent to the parent is already accounted for by the space exponent $s(\pi^v)$.  Second, when computing the space-time exponents of the query $Q^v$, we need to account for both the sizes of the input relations $R_i$, and for the sizes of the incoming messages $Q^w$.  The latter can be asymptotically larger: if the bags are computed using Generic Join, or Pseudo-Trees, then we need to use the tighter upper bound expression $O(\prod_i N_i^{w_i^*})$ rather than $O(|D|^{\rho^*})$, see Sec.~\ref{sec:pt}.  We illustrate with an example.

\begin{example} \label{ex:4cycle}
  Consider the 4-cycle query $Q_\square() \leftarrow E_1(A_1,A_2) \otimes E_2(A_2,A_3) \otimes E_3(A_3,A_4) \otimes E_4(A_1,A_4)$, and the tree decomposition with two bags $\chi(v)=\set{A_1A_2A_3}$, $\chi(w)=\set{A_3A_4A_1}$, where $v$ is the root.  Assume $|E_1|=|E_2|=|E_3|=|E_4|=N$. 
  The sub-query at $w$ is $Q^w(A_1A_3) = E_3(A_3A_4) \otimes E_4(A_1A_4)$.  Its space-time exponents are $(2,2)$, because the optimal fractional edge cover (for both $\var(Q)$ and $\set{A_1,A_3}$) is $w_3=w_4=1$.  This means that GJ can compute it in time $O(N^2)$, and its output takes space $O(N^2)$.
  Consider next the sub-query $Q^v() = E_1(A_1A_2)\otimes E_2(A_2A_3) \otimes Q^w(A_1A_3)$.
  Although $Q^v$ has the shape of a triangle query, GJ does not compute it in time $O(N^{1.5})$ but rather in time $O(N^2)$, because the message $Q^w$ can be as large as $O(N^2)$.  The optimal fractional edge cover, which minimizes $|E_1|^{w_1}\cdot |E_2|^{w_2}\cdot |Q^w|^{w_0}$, is $w_1=w_2=1$ and $w_0=0$, and GJ will compute this query in time $O(N^2)$ and space $O(1)$.   Therefore, the space-time exponents of this tree decomposition are $(2,2)$.
  
  In this simple 4-cycle example, using pseudo-trees instead of GJ to compute the bags does not improve either the space or time exponent.  However, we will see in Example~\ref{ex:TDPTvsTDGJ} that replacing GJ with PT can lead to asymptotic improvements.
\end{example}

When all cardinalities are equal $|R_1|=|R_2|=\cdots$ then  the simplified formula $O(|D|^{\rho^*})$ still gives an upper bound on the time and space complexity, assuming that we ignore the messages $Q^w$ when computing  $\rho^*$; equivalently, we assign each of them the weight $0$.
Hence, for a $\td^{\gj}$ plan with tree decomposition $(T,\chi)$,
we get the following space-time exponents (known as folklore): 
\begin{align}
\label{eq:tdgj}
  s(T,\chi) = &   \max\left(\rho^*(\bm X), \max_{(u,v) \in E(T)}\rho^*(\chi(u)\cap \chi(v))\right),
 & 
  t(T,\chi) = & \max_{v \in V(T)} \rho^*(\chi(v)) = fhw(T,\chi).
\end{align}
In words, the space exponent is the maximal fractional edge cover number of the intersection of adjacent bags (or the maximal output size) while the time exponent is simply the fraction hypertree width of the tree decomposition.

We end this section by comparing $\tdc$ for different $\calC$ and proving some of the domination relations in Fig.~\ref{fig:plan-families}.

\begin{example}
\label{ex:TDPTvsTDGJ}
  $\td^{\pt}$ plans can strictly improve both the time and the space exponents of $\td^{\gj}$ plans.  For example, consider the query $\QTri()$ in Fig. \ref{fig:tdptQ}, and the $\td^{\pt}$ plan  in Fig. \ref{fig:tdptex}.  Using GJ to process the bags results in space-time exponents of $(1,\nicefrac{5}{2})$ while using the PTs results in $(1,2)$.  We can improve the exponents even further, by using a TD with a single bag $ABCDEF$.  To compute it, add the node $F$ as a child of $B$ to the left PT in Fig.~\ref{fig:tdptex}.  The space-time exponents decreased to $(0,2)$.
\end{example}

\begin{restatable}{theorem}{thmtdgjpt}   
\label{thm:tdgj-pt}
The following hold:  $\pt \not\preceq \tdgj, \tdgj \not\preceq \pt,$ and $\forall \calC\colon \td^{\calC}\preceq \calC$.
\end{restatable}

\begin{proof}[Proof Sketch.]
  $\pt \not\preceq \tdgj$ follows from the fact that the 4-path query $R(AB)\otimes S(BC)\otimes T(CD)$ admits a $\tdgj$ query plan $\Pi$ with $e(\Pi) = (1,1)$, while time exponent 1 is not feasible in \pt.  $\tdgj \not\preceq \pt$ is proved by showing that the 3-path query in Example~\ref{ex:3path} has a plan $\Pi \in \pt$ with $e(\Pi) = (0,1)$, while \tdgj allows for plans $\Pi'$ with $e(\Pi') = (1,1)$ or $e(\Pi') = (0,2)$ but not $(0,1)$; none dominate $e(\Pi)=(0,1)$.
  $ \td^{\calC}\preceq \calC$ follows as $Q^v=Q$ for the single-bag TD.
\end{proof}

\section{Caching}
\label{sec:ptc}

In this section, we describe the addition of caching to pseudo-trees.  While this method was introduced in~\cite{DBLP:series/synthesis/2019Dechter}, we extend it here to handle output variables, and perform its (non-obvious!) space-time analysis.  A \emph{cache} is a data structure that maps from a set of \emph{keys} to a set of \emph{values}.

\begin{example}
To motivate caching, consider the 4-path query $Q()\leftarrow R(AB)\otimes S(BC) \otimes T(CD)$, over the semiring $\N$, and consider the linear PT $A-B-C-D$, where $A$ is the root.  The runtime of this plan is given by the AGM bound, $O(N^2)$.  Intuitively, this query plan corresponds to the summation $\sum_a\left(\sum_b R(ab)\cdot \left(\sum_c S(bc) \cdot \left(\sum_d T(cd)\right)\right)\right)$.  We note that the subexpression $M_C(b) := \sum_c S(bc) \cdot \sum_d T(cd)$ is independent of $a$, and, by caching the values $M_C(b)$, we can avoid recomputing this expression.  Similarly, we can cache $M_D(c) := \sum_dT(cd)$.  By adding caches to Algorithm~\ref{alg:pt} we can trade off space for time.  In our example, the two caches decrease the runtime of the PT above from $O(N^2)$ to $O(N)$, while the space increases from $O(1)$ to $O(N)$.
\end{example}

Throughout this section we fix a query $Q(\bm X) \leftarrow \bigotimes_i R_i(\boldsymbol{X}_i)$ and a pseudo-tree $P=(\boldsymbol{V},E)$, where $\bm V=\var(Q)$.   Assume we decide to cache the values returned by $\textsc{solve}(A, \bm y)$ of the recursive Algorithm~\ref{alg:pt}, for some $A \in \bm V$.  The \emph{key} of the cache $M_A$ at $A$ is called the \emph{context} of $A$.

\begin{definition}[\cite{DBLP:series/synthesis/2019Dechter}] \label{def:context}
  The {\em context} of a variable $A\in \boldsymbol{V}$ is defined as
$$\con(A) = \setof{B \in \anc(A)}{\exists C \in  \descc(A) \mbox{, s.t. } B,C \in \boldsymbol{X}_i \mbox{ for some atom $R_i(\bm X_i)$ of $Q$}},$$
and the {\em closed context} of $A$ is $\conn(A)= \con(A)\cup \{A\}$.
\end{definition}

The main property of $\con(A)$ is that the value returned by $\textsc{solve}(A,\bm y)$ depends only on $\bm y[\con(A)]$ and not on the entire tuple $\bm y$.  Therefore, we can cache these values in a cache $M_A$ with key $\con(A)$, whose values are $\K$-relations of type $\dom^{\outt(A)} \rightarrow \K$ (the type returned by $\textsc{solve}(A,\bm y)$).  The type of this cache is $M_A\colon \dom^{\con(A)} \rightarrow (\dom^{\outt(A)}\rightarrow \K)$, which is equivalent, through curry-uncurry, to $M_A\colon \dom^{\con(A) \cup \outt(A)}\rightarrow \K$.  Therefore, the space usage of the cache $M_A$ is given by $\rho^*(\con(A)\cup \outt(A))$.  The time spent by the algorithm at node $A$ will be reduced, because it only needs to call $\textsc{solve}(A,\boldsymbol{y})$ once for each distinct value $\boldsymbol{y}[\con(A)]$.
We are now ready to define a pseudo-tree with caching:

\begin{definition}
\label{def:ptc}
A {\em pseudo-tree with caching} (PTC) of $Q$ is a pair $(P,\boldsymbol{C})$, where $P=(\boldsymbol{V},E)$ is a \ptree of $Q$, and $\boldsymbol{C} \subseteq \boldsymbol{V}$ is a subset of the variables for which we add a cache.  We require $\ptroot(P) \in \boldsymbol{C}$.
\end{definition}

The gray lines in Algorithm~\ref{alg:pt} represent its extension to caching, which we call Algorithm~\refalgptc.  We now compute its space-time complexity. We have already seen the space requirement for a cache $M_A$ and, hence, know the space complexity.
So, let us focus on the time complexity.  If we do not use any caches, then the time complexity of the \texttt{for}-loop in line~\ref{line:for:a} of the Algorithm is $\rho^*(\ancc(A)\cup \out(A))$: this is what we used in Def.~\ref{def:space:time:pt}.  But if some $B \in \ancc(A)$ uses a cache, then it suffices to consider only the set\footnote{$[A,B]$ denotes the set of nodes between $A$ and $B$; $(A,B]$ denotes $[A,B]\setminus\set{A}$.} $\con(B)\cup [A,B]$ instead of $\ancc(A)$.  Indeed, consider two calls to $\textsc{solve}(A,-)$: first, $\textsc{solve}(A,\bm x)$, followed at some later time by a second call $\textsc{solve}(A,\bm y)$.  If $\bm x[\con(B)\cup (A,B]]=\bm y[\con(B) \cup (A,B]]$ then the second call will not happen, because of the cache at $B$.  The only variables in $\ancc(A)$ relevant to the time consumption in line~\ref{line:for:a} are $\con(B)$, and $[A,B]$.  If $A$ has multiple ancestors $B$ with a cache, then we will only consider the lowest one (closest to $A$).  This justifies the following generalization of Def.~\ref{def:space:time:pt}:

\begin{definition} \label{def:ptc:exp}
The class of query plans $\ptc(Q)$ consists of pseudo-trees with caching $(P,\boldsymbol{C})$ of $Q$ and their space and time exponents are defined as:
\begin{align*}
  s(P,\boldsymbol{C}) :=  \max_{A \in \boldsymbol{C}} \rho^*(\con(A) \cup \outt(A)), \quad  
  t(P,\boldsymbol{C}) :=  \max_{A \in V(P)} \rho^*(\con(B_A)\cup [A, B_A]\cup \out(A)), 
\end{align*}
where\footnote{Given a nonempty set $\boldsymbol{S}$ on some branch of the tree, i.e.  $\boldsymbol{S} \subseteq \ancc(A)$ for some $A$, we denote by $\min(\boldsymbol{S})$ its smallest element, i.e. $\min(\boldsymbol{S}) \in \boldsymbol{S}$ and $\boldsymbol{S}\subseteq \ancc(\min(\boldsymbol{S}))$.} $B_A:=\min(\boldsymbol{C}\cap \ancc(A))$. %
\end{definition}

The reader may check that, when there are no caches (i.e., $\boldsymbol{C} = \{\ptroot(P)\}$), then the space and time exponents of the PTCs coincide with those of PTs in Def.~\ref{def:space:time:pt}.
In general, we prove:

\begin{restatable}{theorem}{thmalgoptc}   
\label{thm:algo:ptc}
If $(P,\bm C) \in \ptc(Q)$, then Algorithm \refalgptc computes $\sem{Q}$ in time $ O(|D|^{t(P,\bm C)})$ and uses space $O(|D|^{s(P, \bm C)})$.
\end{restatable}

We end this section by establishing the domination relationships involving $\ptc$ in
\ifArxiv
Fig.~\ref{fig:plan-families} (proofs are in the Appendix):
\fi
\ptc improves upon \tdgj, and augmenting TDs with \pt or with \ptc give equivalent classes.  
The separation shown in Theorem~\ref{thm:tdgj-pt} implies that \ptc strictly improves upon all classes above it.
However, we state strict dominations collectively later in Corollary \ref{cor:ptcr-tdpt-ptc}

\begin{restatable}{theorem}{thmtdgjptc}   
\label{thm:tdgj-ptc}

The class $\ptc$ dominates $\tdgj$, i.e., $\ptc \preceq \tdgj$. 
\end{restatable}

\begin{proof}[Proof Sketch.] Given a plan $\Pi := ((T,\chi),\pi) \in \td^{\gj}$, we construct a plan $\Pi_0 := (P,\bm C) \in \ptc$ such that $\Pi_0 \preceq \Pi$.  The construction is based on the variable elimination procedure for a tree decomposition~\cite{DBLP:journals/corr/KhamisNRR15}, and proceeds by induction on the number of bags in $T$.  If $T$ has a single bag, then $\Pi$ is essentially a $\gj$ plan, and the claim follows from $\pt \preceq \gj$.  Otherwise, let $v$ be a leaf of $T$, and $p:= \parent(v)$.  We eliminate all variables $\set{A_1, \ldots, A_k} := \chi(v)\setminus \chi(p)$.  Let $\bm Z$ be their neighbors, $\bm Z := \setof{B}{\exists \text{ atom } R_i(\bm X_i), \exists j, \text{ s.t. } A_j, B \in \bm X_i}$.  Let $Q'$ to be the query obtained from $Q$ by removing all variables $A_1, \ldots, A_k$, and adding a new atom $R(\bm Z)$.  Let $\Pi' = ((T',\chi),\pi)$ be the plan obtained from $\Pi$ by removing the leaf $v$.  By induction hypothesis, $\Pi'$ can be converted to a $\ptc$ plan $\Pi_0' = (P',\bm C')$ such that $\Pi_0' \leq \Pi'$.  All variables $\bm Z$ belong to a branch of $\Pi_0'$ (because of the atom $R(\bm Z)$).  Construct the pseudo-tree $P$ from $P'$ by adding a branch $A_1\text{-}A_2\text{-}\cdots\text{-}A_k$ as a child of the last variable in $\bm Z$.  Finally, define $\Pi_0 := (P,\bm C \cup \set{A_1})$ (only $A_1$ receives a cache).  It can be checked that $\Pi_0 \preceq \Pi$, which proves $\ptc \preceq \td^{\gj}$.  
\end{proof}

\begin{restatable}{theorem}{thmtdpttdptc}
\label{thm:tdpt-tdptc}
$\tdpt \preceq \tdptc$. Therefore, $\tdpt \equiv \tdptc$ and $\tdpt \preceq \ptc$. 
\end{restatable}

\section{Resetting the Cache}
\label{sec:ptcr}

We saw how the addition of caches reduces the time exponent while increasing the space exponent of a pseudo-tree.  We describe here \emph{pseudo-tree with caching and resets}, \ptcr, which allows for a finer tradeoff between space and time.  The basic principle was introduced in~\cite{DBLP:series/synthesis/2019Dechter}, but with no analysis of its complexity.  We provide an analysis, and describe a non-trivial improvement, using a notion called \emph{relevant ancestors}, which reduces the asymptotic time complexity of the algorithm.

\begin{example} \label{ex:ptcr:basic} To motivate resetting caches, consider the scalar query $\QShort()$ in Fig.~\ref{fig:ptcrQuery} and the \ptree consisting of a single path $(A\text{-}B\text{-}C\text{-}D\text{-}E\text{-}F)$, and with root $A$.  Adding caches to $E, F$ leads to space-time exponents of $(1.5,2)$.  For example, the cache $M_E$ for variable $E$ has key $\con(E)=ABD$, and its space usage is given by $\rho^*(ABD)=1.5$.  When $\textsc{solve}(E,abcd)$ is called, Algorithm \refalgptc stores the result in $M_E(abd)$; in later calls, if $c$ changed while $abd$ are the same, the algorithm immediately returns the cached value. Note that once the values of $a$ or $b$ change we can safely discard (or reset) all entries $M_E(ab-)$, because the values $abc$ arrive at $\textsc{solve}(E,-)$ in lexicographic order.  We can reduce the space of the cache by only keeping cached entries whose keys agree on $AB$ -- essentially only storing the results for different values of $D$. %
This decreases the space usage to $\rho^*(D)=1$, eq. for the cache $M_F$.  With this ``reset'' improvement, $\QShort$ can be computed with space-time complexity~$(1,2)$.
\end{example}

We describe now this technique in general.  Fix an SPQ $Q(\boldsymbol{X}) \leftarrow \bigotimes_i R_i(\boldsymbol{X}_i)$.

\begin{figure}
\begin{minipage}[b]{0.23\textwidth}
\begin{figure}[H]
\centering

\begin{center}
    \begin{tikzpicture}[scale = 0.45,font=\footnotesize]
        \node (A) at (4,2) {$A$};
        \node (B) at (6,2) {$B$};
        \node (F) at (2,0) {$F$};
        \node (E) at (4,0) {$E$};
        \node (D) at (6,0) {$D$};
        \node (C) at (8,0) {$C$};

        \draw (A) -- (B);
        \draw (A) -- (C);
        \draw (A) -- (D);
        \draw (A) -- (E);
        \draw (A) -- (F);
        \draw (B) -- (C);
        \draw (B) -- (D);
        \draw (B) -- (E);
        \draw (B) -- (F);
        \draw (F) -- (E);
        \draw (E) -- (D);
        \draw (D) -- (C);
    \end{tikzpicture}
\end{center}
\caption[]{Query $\QShort$}
    \label{fig:ptcrQuery}
\end{figure}
\end{minipage}
\hfill
\begin{minipage}[b]{0.18\textwidth}
\begin{figure}[H]
\centering

    \begin{tikzpicture}[scale = 0.7,font=\footnotesize]

                \foreach \angle/\label in {18/A, 90/B, 162/C, 234/D, 306/E}
                {
                    \node (\label) at (\angle:1) {$\label$};
                }
                \draw (A) -- (B);
                \draw (A) -- (C);
                \draw (D) -- (B);
                \draw (D) -- (C);
                \draw (E) -- (B);
                \draw (E) -- (D);
                \draw (C) -- (B);
            \end{tikzpicture}
    \caption[]{Query $Q_{\pentagon}$}
    \label{fig:badptcrquery}
\end{figure}
\end{minipage}
\hfill
\begin{minipage}[b]{0.2\textwidth}
\begin{figure}[H]
\centering
\begin{tikzpicture}[scale=0.4,font=\footnotesize]

                \node[inner sep=1.5pt] (A) at (0,0) {$A$};
                \node[inner sep=1.5pt] (B) at (0,-1) {$B$};
                \node[inner sep=1.5pt] (C) at (0,-2) {$C$};
                \node[inner sep=1.5pt] (D) at (0,-3) {$D$};
                \node[inner sep=1.5pt] (E) at (0,-4) {$E$};

                \node[anchor=west] at ($(A) + (.5,0)$) {$|$};
                \node[anchor=west] at ($(B) + (.5,0)$) {$A|$};
                \node[anchor=west] at ($(C) + (.5,0)$) {$AB|$};
                \node[anchor=west] at ($(D) + (.5,0)$) {$|BC$};
                \node[anchor=west] at ($(E) + (.5,0)$) {$B|D$};

                \draw (A) -- (B);
                \draw (B) -- (C);
                \draw (C) -- (D);
                \draw (D) -- (E);
            \end{tikzpicture}
    \caption[]{PTCR for $Q_{\pentagon}$}
    \label{fig:badptcr}
    
\end{figure}
\end{minipage}
\hfill
\begin{minipage}[b]{0.36\textwidth}
\begin{figure}[H]
\begin{center}
    \begin{tikzpicture}[scale = 0.45,font=\footnotesize]
        \node (A) at (-2,0) {$A$};
        \node (B) at (-1,2) {$B$};
        \node (C) at (0,0) {$C$};
        \node (D) at (1,2) {$D$};
        \node (E) at (2,0) {$E$};
        \node (F) at (3,2) {$F$};
        \node (G) at (4,0) {$G$};
        \node (H) at (5,2) {$H$};
        \node (I) at (6,0) {$I$};
        \node (J) at (7,2) {$J$};
        \node (K) at (8,0) {$K$};

        \draw (A) -- (B);
        \draw (A) -- (C);
        \draw (A) -- (D);
        \draw (B) -- (C);
        \draw (B) -- (D);
        \draw (B) -- (E);
        \draw (C) -- (D);
        \draw (C) -- (E);
        \draw (D) -- (E);
        \draw (D) -- (F);
        \draw (D) -- (G);
        \draw (E) -- (F);
        \draw (E) -- (G);
        \draw (E) -- (H);
        \draw (F) -- (G);
        \draw (F) -- (H);
        \draw (G) -- (H);
        \draw (I) -- (G);
        \draw (I) -- (H);
        \draw (I) -- (J);
        \draw (I) -- (K);
        \draw (J) -- (G);
        \draw (J) -- (H);
        \draw (J) -- (K);
        \draw (H) -- (K);
    \end{tikzpicture}
\end{center}
\caption[]{Query $\QLong$}
    \label{fig:tdptQuery}
\end{figure}
\end{minipage}
\end{figure} 
\begin{definition}
\label{def:ptcr}
A pseudo-tree with caching and resets (PTCR) of $Q$ is a pair $(P,C)$, where $P=(\boldsymbol{V},E)$ is a pseudo-tree and $C$ is a function $C\colon \boldsymbol{V} \rightarrow \mathbb{N}$.
\end{definition}

Fix a variable $A \in \bm V$, and let its context be $\con(A) = \set{A_1, \ldots, A_n}$; recall that $\con(A) \subseteq \anc(A)$.  We order $\con(A)$ such that $A_1$ is closest to the root and $A_n$ is closest to $A$.  The function $C$ in Def.~\ref{def:ptcr} indicates how many variables from $\con(A)$ will be stored simultaneously.  If $k := \min(C(A),|\con(A)|)$, we partition $\con(A)$ into $\icon(A) \cup \scon(A)$, where $\icon(A) := \{A_1, \dots, A_{n-k}\}$ is the \emph{instantiated context} and $\scon(A) := \{A_{n-k+1}, \dots, A_n\}$ is the \emph{stored context}.  The keys of the cache $M_A$ will always agree on $\icon(A)$ but may differ on $\scon(A)$. %
Any change of a variable in $\icon(A)$ invalidates (resets) $M_A$.  We use a bar to indicate the partition of $\con(A)$: in Example~\ref{ex:ptcr:basic}, if $C(E)=1$ then we write $\con(E)=AB|D$.  If $C(A)=0$ for a variable $A$ then it is equivalent to having no cache\footnote{\label{footnote:c:eq:0} Strictly speaking, if $C(A)=0$ then we cache at $A$ a single returned value $M_A(\boldsymbol{x}[\con(A)]) := \textsc{solve}(A,\bm x)$, and can reuse it as long as $\bm x[\con(A)]$ doesn't change.  However, if $\con(A)$ contains $\parent(A)$, then we cannot expect to ever reuse the value $M_A(\boldsymbol{x}[\con(A)])$, which is equivalent to not having a cache.  If $\parent(A) \not\in \con(A)$, then the PT is suboptimal:
we can simply connect $A$ to $parent(parent(A))$, that is, $A$ and $parent(A)$ become siblings.}.%

A basic algorithm for handling resets described in~\cite{DBLP:series/synthesis/2019Dechter} is as follows.  Each variable $A$ has a cache $M_A$ with key $\con(A)$, and an \emph{instantiated}  tuple $\icon_A \in \dom^{\icon(A)}$ storing the last value of $\icon(A)$. When  $\textsc{solve}(A,\bm y)$ is called, it first checks whether $\icon_A = \bm y[\icon(A)]$.  If yes, then it uses the cache like Algorithm~\refalgptc.  If not, then it resets the cache $M_A$.  However, this algorithm is not optimal, as we explain in the next example.

\begin{example} \label{ex:ptcr:fill:cache} Consider the scalar query $Q_{\pentagon}()$ in Fig.~\ref{fig:badptcrquery} and the PTCR $(P,C)$ in Fig.~\ref{fig:badptcr}.  The figure also shows the partitions of each context.  For example, $\con(E)=B|D$, means that $\icon(E)=B$ and $\scon(E)=D$, and therefore the keys of its cache $M_E(BD)$ always agree on $B$; similarly, $D$ has a full cache $M_D(BC)$.  We want to compute the time complexity of $E$, and for that we need to reason about how often the cache $M_E$ is reset due to $B$ changing its value.  We can do this in two ways. Either we notice that the only ancestor of $B$ is $A$ and, thus, the number of calls to $\textsc{solve}(E,-)$ with a new $B$ is bounded by $\rho^*(AB)$.  Or, we notice that $D$ has a fully stored cache $M_D(BC)$, and calls $\textsc{solve}(E,-)$ only with unique $BC$ pairs, thus the complexity is also bounded by $\rho^*(BC)$.

Our algorithm reduces this to $\rho^*(B)$, by computing the cache $M_D(BC)$ eagerly.  When $\textsc{solve}(D,abc)$ is called for the first time, $D$ will ignore the values $bc$, and instead it fills its cache $M_D(BC)$ entirely with \emph{all} values of $BC$: it iterates over the distinct values\footnote{$\supp(\sem{Q[BC]})$ means the projection of the query on $BC$.  It is a natural extension of Eq.~\eqref{eq:qay}.
} $bc \in \supp(\sem{Q[BC]})$, calls $\textsc{solve}(E,abcd)$ recursively, and stores the result in $M_D(bc)$.  $D$ traverse the $\supp(\sem{Q[BC]})$ by grouping on $B$ (e.g., by sorting it lexicographically), so that the same $B$-values occur consecutively, e.g. $b_1c_1, b_1c_2, b_1c_3, b_2c_1,\ldots$ When $B$ changes from $b_1$ to $b_2$, $E$ resets its cache, but there is no loss of work, because $b_1$ will never be seen again; this is similar to the argument in Example~\ref{ex:ptcr:basic}.  The number of cache resets is reduced\footnote{The careful reader may have noticed that $\rho^*(AB)=\rho^*(BC)=\rho^*(B)=1$, so for this simple example our improved algorithm does not reduce the runtime but the dependency.  The runtime reduction does happen for more complex examples.} to $\rho^*(B)$.  This example was simple, because $D$ had no instantiated variables, $\icon(D)=\emptyset$; the general case requires the technical Definition~\ref{def:ria:ra} below.

Filling the cache eagerly is a significant extension of GJ and all its implementations in practice~\cite{DBLP:conf/icdt/Veldhuizen14,DBLP:journals/pvldb/FreitagBSKN20}, where the values of $AB\ldots$ are examined in strict lexicographic order, e.g. $a_1b_1,a_1b_2,a_2b_1,a_2b_2\ldots$ Instead, the values of $B$ arrive at the function $\textsc{solve}(E,-)$ in sorted order $b_1,b_1,b_2,b_2,b_3,\ldots$
\end{example}

Algorithm~\ref{alg:ptcr} extends Algorithm~\refalgptc from \ptc to \ptcr, and uses the following:

\begin{definition}  \label{def:ria:ra}
  For a PTCR $(P,C)$ and variable $A \in \boldsymbol{V}(P)$, we define the {\em (closed) relevant (instantiated) ancestors} $\ra(A),\raexc(A),\ria(A)$ --- where $\conn(A)\subseteq\ra(A)\subseteq \ancc(A), \con(A)\subseteq\raexc(A)\subseteq \anc(A),$ and (for $\icon(A)\neq \emptyset$) $\icon(A)\subseteq\ria(A)\subseteq \ancc(\min(\icon(A)))$ --- recursively as follows:
\begin{align*}
    \ria(A) &= \begin{cases}
        \ra(parent(A)) \cap \ancc(\min(\icon(A))) & \icon(A)\neq \emptyset, \\
        \emptyset & \icon(A)= \emptyset,
    \end{cases}\\ %
    \raexc(A) & = \ria(A)\cup \scon(A), \quad \ra(A) = \raexc(A) \cup \{A\}. \nonumber
\end{align*}
\end{definition}

The key difference between Algorithms~\ref{alg:ptcr} and~\refalgptc is that the function $\textsc{solve}(A,\bm x)$ fills its cache $M_A$ eagerly.  This is done by \textsc{fillCache}, which iterates over all these values in line~\ref{line:line11}.  Notice that $\supp(\sem{Q[\scon(A)|\ria_A]})$ is a \emph{set}: the algorithm processes these values lexicographically. %
However, the function \textsc{fillCache} is called by $\textsc{solve}(A,-)$ only once for each value of the variables $\ria(A)$.  Referring to Example~\ref{ex:ptcr:fill:cache}, when the function $\textsc{solve}(D,-)$ is called, we have $\ria(D)=\emptyset$, hence $\textsc{fillCache}$ is called only once, and in line~\ref{line:line11} it iterates over $\supp(\sem{Q[BC]})$.

\begin{algorithm}[t]
\caption{PTCR Algorithm}%
\label{alg:ptcr}
\begin{minipage}[t]{0.48\textwidth}
\begin{algorithmic}[1]
    \Statex \textbf{Input:} $\text{Query } Q(\boldsymbol{X}), \text{PTCR } (P,C)$
    \Statex \textbf{Output:} $\sem{Q}$
    \For{$A\in V(P)$}
        \State $M_A \gets \emptyset, \ria_A \gets \bot$ 
    \EndFor  
    \State \Return \Call{solve}{$root(P),()$}

    \Function{solve}{$A,{\boldsymbol{x}}$} 
    \If{${\boldsymbol{x}}[\ria(A)] = \ria_A $} \Comment{Cache hit}
        \State \Return $M_A({\boldsymbol{x}}[\con(A)])$ 
    \EndIf
    \State $\ria_A \gets \boldsymbol{x}[\ria(A)]$ \label{line:ptcr} \Comment{Cache miss}
    \State $\Call{fillCache}{A,\boldsymbol{x}[\anc(A)\setminus \scon(A)]}$
    \State \Return $M_A({\boldsymbol{x}}[\con(A)])$
    \EndFunction
\end{algorithmic}
\end{minipage}
\hfill
\begin{minipage}[t]{0.51\textwidth}
\begin{algorithmic}[1]
\setcounter{ALG@line}{9}
    \Function{fillCache}{$A,\yanc$} 
        \For{$\ysto\in \supp(\sem{Q[\scon(A)|\ria_A]})$} \label{line:line11}
            \State $\boldsymbol{y} \gets (\yanc,\ysto)$
            \State OUT $\gets \{\boldsymbol{z} \mapsto \0 \mid \boldsymbol{z}\in\dom^{\outt(A)}\}$  
            \For{$a\mapsto s \in \sem{Q[A|\ria_A,\ysto]}$}
                    \State $\boldsymbol{y}'\gets (\boldsymbol{y}, a)$ \label{line:ptcr:fromhere} \Comment{L. \ref{line:ptcr:fromhere}--\ref{line:ptcr:tohere} as in Alg. \ref{alg:pt}}
                    \State TMP $\gets \{a[\boldsymbol{X}\cap \{A\}] \mapsto s\}$ 
                    \For{$B\in child(A)$}
                        \State TMP $\gets$ TMP $\otimes$ \Call{solve}{$B,\boldsymbol{y}'$} \label{line:ptcr:recurse}
                    \EndFor 
                    \State OUT $\gets$ OUT $\oplus$ TMP \label{line:ptcr:tohere}
            \EndFor
            \State $M_A \gets M_A \cup \{\boldsymbol{y}[\con(A)] \mapsto $ OUT$\}$
        \EndFor
        
    \EndFunction
\end{algorithmic}
\end{minipage}
\end{algorithm} 
\begin{definition}  
The class of query plans $\ptcr(Q)$ consists of PTCR $(P,C)$ of $Q$ and their space and time exponents are defined as:
\begin{align*}
  s( P,C) =  \max_{A \in {V}(P)} \rho^*(\scon(A) \cup \outt(A)), \quad t( P,C) =  \max_{A \in {V}(P)} \rho^*(\ra(A)\cup \out(A)).
\end{align*}
\end{definition}

\begin{restatable}{theorem}{thmalgoptcr}    
\label{thm:algo:ptcr}
If $(P,C) \in \ptcr(Q)$, then Algorithm \ref{alg:ptcr} computes $\sem{Q}$ in ${O}(|D|^{t( P,C)})$ time and uses $O(|D|^{s( P,C)})$ space.
\end{restatable}

Lastly, we prove the remaining relationships in Fig.~\ref{fig:plan-families} up to $\tdptcr$, in particular demonstrating the difference in strength of \tdpt and \ptcr.  This separation together with the separation given in Theorem~\ref{thm:tdgj-pt} imply that all dominations corresponding to downward arrows up to $\tdptcr$ are strict.

\begin{restatable}{theorem}{thmptcrtdpt}   
\label{thm:ptcr-tdpt}
$\ptcr \not\preceq \tdpt$ and $\tdpt \not\preceq \ptcr$. 
\end{restatable}

\begin{proof}[Proof Sketch]
  For $\ptcr \not\preceq \tdpt$, we use the query $\QShort()$ depicted in Fig.~\ref{fig:ptcrQuery} and the $\ptcr$ in given Example~\ref{ex:ptcr:basic}, with space-time exponents $(1,2)$.  Because no edge separates the query, any TD with $2+$ bags uses superlinear space, and a single bag is equivalent to a \ptree, which takes more than quadratic time.  For $\tdpt \not\preceq \ptcr$, we use the query $\QLong()$ depicted in Fig.~\ref{fig:tdptQuery} and a TD with three bags, separated by $DE$ and $GH$; its space-time exponents are $(1,2)$. The best $\ptcr$ has space-time exponents of $(1.5,2)$ or $(1,2.5)$.
\end{proof}

\begin{restatable}{corollary}{corralgoptcr}   
\label{cor:ptcr-tdpt-ptc}
    The dominations relations represented by downward arrows up to the ones ending at $\tdptcr$ in Fig.~\ref{fig:plan-families} are all strict.
\end{restatable}

\section{Using Recursion to Reorient Sub-Trees }
\label{sec:rpt}

We now present a final class of plans termed recursive pseudo-trees, $\rptcr$, that unify the strengths of PTs and TDs into a single approach. While $\ptcr$ (and also $\ptc$) captured some of the advantages of  TDs, as seen by $\ptcr \prec\td^{\gj}$, TDs can bring further benefits when computed with \ptrees, as seen by $\tdptcr \prec \ptcr$.  One benefit of TD-based plans is that different bags can use different variable orders, for example in  Fig.~\ref{fig:tdptex} we could use $A\text{-}B$ in one bag, and $B\text{-}A$ in the other, which is not possible in a \ptree.  $\rptcr$ loosens this restriction, and fully captures the benefit of TDs: we will show the equivalence of $\rptcr$ and $\tdrptcr$ in Thm.~\ref{thm:algo:rpt}. 

To begin, we revisit the function $\textsc{fillCache}(A,\boldsymbol{y}_{\text{anc}})$ of Algorithm~\ref{alg:ptcr}.
This function evaluates for every tuple $\ysto \in \supp(\sem{Q[\scon(A)|\ria_A]})$ (see line \ref{line:line11}) the subquery of $Q$ restricted to variables $\descc(A)$ and caches the result in $M_A$. %
To do this, the algorithm based on a PTCR plan first iterates through the variables in $\scon(A)$, then proceeds to the variables in $\descc(A)$. 
While intuitive, this is an arbitrary decision. One could instead simply use a different PTCR $(P',C')$ to compute this subquery where $\scon(A)$ is included as output variables. Effectively, we fill the cache using a different PTCR and see this as replacing the subtree of $(P,C)$ rooted at $A$ with the new sub-plan $(P',C')$. 
We call the resulting structure a \textit{recursive pseudo-tree}.

\begin{figure}
\begin{minipage}[b]{0.25\textwidth}
\begin{figure}[H]
\centering
\begin{tikzpicture}[scale=0.75, font=\footnotesize]

                \foreach \angle/\label in {0/K, -36/I, -72/G, -108/E, -144/C, -180/A}
                {
                    \node[inner sep=1pt] (\label) at (\angle:2) {$\label$};
                }
        
                \foreach \angle/\label in {-18/J, -54/H, -90/F, -126/D, -162/B}
                {
                    \node[inner sep=1pt] (\label) at (\angle:1) {$\label$};
                }

                \node[inner sep=1pt] (L) at (0,0) {$L$};

                \draw (A) -- (B);
                \draw (A) -- (C);
                \draw (A) -- (D);
                \draw (B) -- (C);
                \draw (B) -- (D);
                \draw (B) -- (E);
                \draw (C) -- (D);
                \draw (C) -- (E);
                \draw (D) -- (E);
                \draw (D) -- (F);
                \draw (D) -- (G);
                \draw (E) -- (F);
                \draw (E) -- (G);
                \draw (E) -- (H);
                \draw (F) -- (G);
                \draw (F) -- (H);
                \draw (G) -- (H);
                \draw (I) -- (G);
                \draw (I) -- (H);
                \draw (I) -- (J);
                \draw (I) -- (K);
                \draw (J) -- (G);
                \draw (J) -- (H);
                \draw (J) -- (K);
                \draw (H) -- (K);
                \draw (L) -- (A);
                \draw (L) -- (B);
                \draw (L) -- (C);
                \draw (L) -- (D);
                \draw (L) -- (E);
                \draw (L) -- (F);
                \draw (L) -- (G);
                \draw (L) -- (H);
                \draw (L) -- (I);
                \draw (L) -- (J);
                \draw (L) -- (K);
            \end{tikzpicture}
    \caption[]{Query $\QCirc$}
    \label{fig:rptQuery}
\end{figure}
\end{minipage}
\hfill
\begin{minipage}[b]{0.74\textwidth}
\begin{figure}[H]
  \begin{center}
    \begin{tikzpicture}[scale = 0.5, font=\footnotesize, rotate=90]
        
            \node {$L$}
                child { node [inner sep=1pt]  {$C$}
                    child { node [inner sep=1pt]  {$B$}
                        child { node [inner sep=1pt]  {$D$}
                            child [xshift=7.5mm] {node  [inner sep=1pt] {$A$}}
                            child [xshift=7.5mm] {node [inner sep=1pt] (E) {$E$}
                                child [xshift=7.5mm, edge from parent/.style={draw,darkgray}] {node  [inner sep=1pt] (Fprime) {$\textcolor{gray}{F'}$}
                                    child [edge from parent/.style={draw,gray}] {node  [inner sep=1pt] {\textcolor{lightgray}{$G'$}}
                                        child [edge from parent/.style={draw,lightgray}] { node [inner sep=1pt]  {$\phantom{c}$}}
                                        }
                                    }
                                 child [xshift=-22.5mm, edge from parent/.style={draw,dashed}]  {node  [inner sep=1pt] (F) {$F$}
                                     child [edge from parent/.style={draw, solid}] { node  [inner sep=1pt] (E2) {$E$}
                                        child { node  [inner sep=1pt] (G) {$G$}
                                            child [xshift=7.5mm] {node [inner sep=1pt]  (D) {$D$}}
                                            child [xshift=7.5mm] {node [inner sep=1pt]  (H) {$H$}
                                                child [xshift=7.5mm, edge from parent/.style={draw,darkgray}] {node [inner sep=1pt]  (Fprime) {$\textcolor{gray}{I'}$}
                                                    child [edge from parent/.style={draw,gray}] {node [inner sep=1pt]  {\textcolor{lightgray}{$J'$}}
                                                        child [edge from parent/.style={draw,lightgray}] { node [inner sep=1pt]  {$\phantom{c}$}}
                                                        }
                                                    }
                                                child [xshift=-22.5mm, edge from parent/.style={draw,dashed}] {node [inner sep=1pt]  (I) {$I$}
                                                    child [edge from parent/.style={draw, solid}] { node  [inner sep=1pt] (H2) {$H$}
                                                        child { node [inner sep=1pt]  (J) {$J$}
                                                            child [xshift=7.5mm] {node [inner sep=1pt]  (G2) {$G$}}
                                                            child [xshift=7.5mm] {node  [inner sep=1pt] (K) {$K$}}
                                                        }
                                                    }
                                                }
                                            }
                                        }
                                    }
                                }
                            }
                        }
                    }
                };

            \node at ($(F) + (0.65,0)$) {\textcolor{red}{$L$}};
            \node at ($(F) + (-0.65,0)$) {\textcolor{blue}{$ED$}};
            \node at ($(E2) + (-0.65,0)$) {\textcolor{blue}{$ED$}};
            \node at ($(G) + (-0.65,0)$) {\textcolor{blue}{$D$}};
            \node at ($(D) + (-0.65,0)$) {\textcolor{blue}{$D$}};

            \node at ($(I) + (0.65,0)$) {\textcolor{red}{$L$}};
            \node at ($(I) + (-0.65,0)$) {\textcolor{blue}{$HG$}};
            \node at ($(H2) + (-0.65,0)$) {\textcolor{blue}{$HG$}};
            \node at ($(J) + (-0.65,0)$) {\textcolor{blue}{$G$}};
            \node at ($(G2) + (-0.65,0)$) {\textcolor{blue}{$G$}};

    \end{tikzpicture}
\end{center} 
  \caption[An RPT for $Q$]{An RPT for $\QCirc$}
    \label{fig:rpt}
\end{figure}
\end{minipage}
\end{figure}

\begin{example} \label{ex:rpt:basic} 

To motivate using different PTCR to fill the caches, let us consider the scalar query $\QCirc()$ given in Fig.~\ref{fig:rptQuery}. 
As there are many 5-cliques, our aim will be to arrive at a time complexity of~$2.5$ and minimize the space complexity.
Note the structure of the query: $L$ is connected to everything, $\mathcal{S}_1 =LDE$ and $\mathcal{S}_2 =LHG$ are the minimal non-trivial separators, and $\mathcal{K}_1= LBDC, \mathcal{K}_2= LFEG, \mathcal{K}_3= LHJI$ are three 4-cliques which each extend in two ways to a 5-clique -- $\mathcal{K}_1$ together with $A$ and $E$ is a 5-clique, respectively, as well as $\mathcal{K}_2 $ together with $D$ and $H$, and $\mathcal{K}_3$ together with $G$ and $K$.
Thus, a natural way of evaluating $\QCirc$ would be to start like a $\ptcr$ plan.
That is, we start with a path $\mathcal{K}_2\colon L-F-E-G$ of nested loops and then branch to the separators $\mathcal{S}_1 \colon D$ and $\mathcal{S}_2\colon H$ separately.
Let us focus on the left branch where we continue with, say, $C$.
At this point, we need a cache to not increase the time complexity beyond $2.5$.
We set the size of the cache to $2$ (the space complexity will be $1$ but the cache contains 2 variables), thus get the context partition $con(C)=L|ED$, and now would have to execute $\textsc{fillCache}(C,l)$.
That is, we have to find and cache the possible values of $\scon(C)=ED$ that fit to $l$ and solve the remaining query on $A,B,C$.
This can be done by a loop structure that first extends $l$ to values $bdc$ and then loops through $A$ and $E$ independently.
We have seen in previous sections that such a loop structure will take time $O(N^{2.5})$ and space $O(N)$ including the loop over the values $l$.
Naturally, the right branch is symmetric.
Thus, we have arrived at an algorithm with the desired space and time consumption.
However, note that this algorithm is \textit{not} the result of a $\ptcr$ plan as this is not the way $\textsc{fillCache}(C,l)$ would have filled the cache of $C$ (no matter how we complete the PTCR).
Crucially, we inverted the order of $E$ and $D$ to fill the cache; we will return to this example a couple more times in this section.

\end{example}

To define recursive pseudo-trees formally, we introduce SPQs with input variables $\boldsymbol{I}$ (these will be $\ria(A)$).
That is, let $Q(\boldsymbol{X})$ be an SPQ and let $\iv \subseteq \var(Q)\setminus \boldsymbol{X}$ denote a subset of the variables in $Q$, which 
we refer to  as the {\em input variables} of $Q(\boldsymbol{X})$. 
Then a PTCR of $Q(\boldsymbol{X})$ is also a PTCR of $(Q(\boldsymbol{X}),\iv)$ when $\iv$ are variables that occur above all other variables. 
That is, these variables form a chain and we call the first variable $\ptroot\in V(P)\setminus \iv$ with $\anc(\ptroot)=\iv, \descc(\ptroot)=\var(Q)\setminus \boldsymbol{I}$ the (real) root of $(P,C,\iv)$.
Furthermore, we require $\ria(\ptroot)=\iv$.

\begin{definition}%
\label{def:rpt}
A {\em recursive pseudo-tree}
(RPT) of an SPQ with input variables $(Q(\boldsymbol{X}),\boldsymbol{I})$
is recursively defined as follows: 

\begin{itemize}
    \item  {\em Base case.} A PTCR $(P,C,\iv)$ of $(Q(\boldsymbol{X}),\iv)$  is an 
    RPT of $(Q(\boldsymbol{X}), \iv)$.
\item {\em Recursion.} Let $(P,C,\iv)$ be a PTCR of $(Q(\boldsymbol{X}),\iv)$, and let $A_1, \dots, A_l$ be vertices in $\boldsymbol{V}(P)\setminus \iv$,
such that $A_i\notin \ancc(A_j)\forall\, i\neq j$. 
Let $Q_i$ be the subquery of $Q(\boldsymbol{X})$ restricted to the variables $\desc(A_i)\cup \ra(A_i)$ and with head variables $\scon(A_i)\cup \outt(A_i)$. Let $(\calP_i,\calC_i, \caliv_i)$ be RPTs of $(Q_i,\ria(A_i))$. 
Then, $(\calP,\calC, \caliv):=((P,(\calP_i)_i), (C,(\calC_i)_i), (\iv, (\caliv_i)_i))$ is an RPT of $(Q(\boldsymbol{X}),\iv)$.
\end{itemize}
The class of query plan $\rptcr(Q)$ consist of RPTs $(\calP,\calC, \caliv)$ of $(Q,\emptyset)$ and their space and time exponents are defined recursively via:
\begin{align*}
    s(\calP,\calC, \caliv) &= \max(\max_{A\in \bigcup_i \anc(A_i)} \rho^*(\scon(A)\cup \outt(A));\max_i s(\calP_i,\calC_i, \caliv_i)),\\
    t(\calP,\calC, \caliv) &= \max(\max_{A\in \bigcup_i \anc(A_i)} \rho^*(\ra(A)\cup \outt(A));\max_i t(\calP_i,\calC_i, \caliv_i)).
\end{align*}
\end{definition}

Using the RPTs $(\calP_i,\calC_i, \caliv_i)$ to perform the same task as $\textsc{fillCache}(A_i,-)$ we get the following:

\begin{restatable}{theorem}{thmalgorpt}    
\label{thm:algo:rpt}
    If $(\calP,\calC, \caliv)\in \rptcr(Q)$, then $\sem{Q}$ can be computed in space $O(|D|^{s(\calP,\calC, \caliv)})$ and time ${O}(|D|^{t(\calP,\calC, \caliv)})$.
\end{restatable}

\begin{example}
    Fig.~\ref{fig:rpt} depicts an RPT $\Pi=(\calP,\calC, \caliv)$ of the query $\QCirc()$ depicted in
    Fig.~\ref{fig:rptQuery}.
    Solid edges represent ``normal'' edges of a PTCR while the dashed edges represent a recursive replacement. 
    The original PTCR had the path PT $(F'\text{-}G'\text{-}H'\text{-}I'\text{-}J'\text{-}K')$ (depicted gray) as a child of the first $E$ while $F'$ was the only variable with a cache.
    Then, $A_1=F'$ got replaced by the sub-RPT $(\calP_1,\calC_1, \caliv_1)$ rooted in $F$ as depicted in Fig.~$\ref{fig:rpt}$ (this is the first dotted edge).
    Note that $L|DE=\con(F')$, $L=\ria(F')=\ria(F)$ and $DE=\scon(F')=\outt(F)$ as required (output and input variables are drawn in blue and red, respectively, in Fig.~\ref{fig:rpt}).
    $(\calP_1,\calC_1, \caliv_1)$ was constructed similarly.
    The original PTCR had the path PT $(I'\text{-}J'\text{-}K')$ as a child of the first $H$ while $I'$ was the only variable with a cache.
    Then, $A_2=I'$ (note that $A_1$ and $A_2$ come from different PTs) got replaced by the sub-PTCR $(P_2,C_2,\{L\})$ rooted in $I$ as depicted in Fig.~$\ref{fig:rpt}$.
    Note that $L|GH=\con(I')$, $L=\ria(I')=\ria(I)$ and $GH=\scon(I')=\outt(I)$ as required.
    The space-time exponents are $(1,\nicefrac{5}{2})$.

    The evaluation algorithm proceeds analogously to the algorithm described in Example~\ref{ex:rpt:basic} with the slight change that the cliques $\mathcal{K}_1,\mathcal{K}_2,\mathcal{K}_3$ are processed one after the other in exactly that order.
    Thus, the RPT depicted in Fig.~\ref{fig:rpt} has two recursion on one single branch while the algorithm described in Example~\ref{ex:rpt:basic} essentially had two branches with one recursion on each branch.
\end{example}

Finally, we relate $\rptcr$ to the other plan classes.
First, we note that $\tdrptcr$ provides no benefit over $\rptcr$.
The proof is inductive, similar to that of Theorem \ref{thm:tdgj-ptc}.

\begin{restatable}{theorem}{thmrptcrtdrptcr}    
\label{thm:rptcr-tdrptcr}
The class $\rptcr$ dominates $\tdrptcr$. Thus, in particular $\rptcr \equiv \tdrptcr$. 
\end{restatable}

Considering the query $\QCirc()$ in Fig.~\ref{fig:rptQuery} and using computer-assisted exhaustive search, we verified that there is no plan $(P, C)\in\ptcr$ with $s(P,C)\leq1$ and $t(P,C)\leq \nicefrac{5}{2}$. There is no linear separator in this query, so the same holds for $\tdptcr$. Combined with Theorem \ref{thm:rptcr-tdrptcr}, we can prove:

\begin{restatable}{theorem}{thmrptcrtdptcr}    
\label{thm:rptcr-tdptcr}
The class $\rptcr$ strictly dominates $\tdptcr$, i.e., $\rptcr\prec \tdptcr$.
\end{restatable}

Naturally, the time exponents of RPTs are always at least as large as the fractional hypertree width (fhw).
As $\rptcr$ dominates all other classes of query plans discussed in this paper, the same is the case for them.

\begin{restatable}{theorem}{thmrptcrfhw}    
\label{thm:rptcr-fhw}
If $(\calP,\calC, \caliv)\in \rptcr(Q)$, then ${t(\calP,\calC, \caliv)} \geq fhw(Q)$, where $fhw(Q)$ is the fractional hypertree width of $Q$.
\end{restatable}

\begin{proof}[Proof Sketch]
    The RPT $(\calP,\calC, \caliv)$ describes a tree structure (e.g., see black parts of Fig.~\ref{fig:rpt}) and we can use this tree structure together with bags $\ra(A)\cup \outt(A)$ at nodes $A$ to construct a tree decomposition of $Q$.
    The width of this tree decomposition is the same as the time exponent of the RPT.
\end{proof}

\section{Conclusion and a Glimpse Beyond}
\label{sec:beyond}

We have presented several novel algorithms for CQ and, more generally, SPQ evaluation by combining and significantly extending existing approaches based on pseudo-trees and tree decompositions. In most cases, we have matched the optimal time complexity of previous algorithms with asymptotically lower space complexity.  We end here by discussing two important lines of future work.

\paragraph{Conjectures on Lower Bounds}
Although some conditional lower bounds have been established for the time complexity of query evaluation~\cite{DBLP:conf/icalp/FanKZ23,DBLP:conf/stoc/BringmannG25}, much less is known about space-time lower bounds. In the complexity community, the space-time tradeoff is studied by proving lower bounds on the product $S \times T$ for specific problems.  For example, the set difference and distinct element problems have been proven to have a lower bound $ST = \Omega(n^2)$~\cite{DBLP:journals/siamcomp/Beame91,DBLP:conf/innovations/McKayW19}. These results are too weak to constrain query evaluation, since even acyclic CQs require at least a quadratic space-time product.

Since there are no widely accepted assumptions to build on, we take a first step toward understanding lower bounds for space-constrained query answering.  We propose a problem called the \emph{Triple $k$-Clique Problem}, and conjecture a lower bound on its space-time complexity.  To motivate our conjecture, we briefly recall the \emph{min-weight $k$-clique hypothesis}, which is a standard complexity assumption for combinatorial problems~\cite{DBLP:conf/soda/LincolnWW18,DBLP:conf/icalp/AbboudWW14, DBLP:conf/icalp/FanKZ23}.  The problem concerns the $k$-clique scalar query:
\begin{align}
  Q() \leftarrow \bigotimes_{i<j \in [1,k]} E(A_i,A_j) \label{eq:q:k:clique}
\end{align}
and the task is to evaluate $Q$ in the tropical semi-ring, where $x \oplus y := \min(x,y)$, $x \otimes y := x+y$; in other words, we are asked to identify a clique with the smallest total edge weight in a weighted graph.  The most commonly used variant of the min-weight $k$-clique hypothesis is stated relative to the number of vertices in the graph. However,~\cite{DBLP:conf/icalp/FanKZ23} showed that this is equivalent to
saying that, for any $\varepsilon > 0$, no algorithm can solve this problem in time $O(|E|^{\frac{k}{2}-\varepsilon})$.  This problem is not a good candidate for a hard space-time problem, because GJ already computes $Q$ in optimal time $O(|E|^{\frac{k}{2}})$ and optimal space $O(1)$.
Instead, we propose the following extension:

\begin{conjecture}[The Triple $k$-Clique Conjecture]  %
  For $k:= 2\ell, \ell\geq 2$, consider the query:
    \begin{align}
        Q() \leftarrow \bigotimes_{i < j\in [1,k]} E(A_i,A_j) \otimes\bigotimes_{i < j\in [k-\ell+1,k+\ell]} E(A_i,A_j) \otimes\bigotimes_{i < j\in [k+1,2k]} E(A_i,A_j) \label{eq:q:triple:k:clique}
    \end{align}
    Then, for any $\varepsilon > 0$, no algorithm can solve this problem over the tropical semi-ring in space $S=O(|E|^{\frac{k}{4}-\epsilon})$ and time $T=O(|E|^{\frac{k}{2}})$.
\end{conjecture}
The tree decomposition with three bags corresponding to the three $\bigotimes$-expressions above uses space\footnote{For a quick computation of $\rho^*$ observe that, for any graph $G$ with vertices $V$ and no isolated vertices, $\rho^*(V) \geq |V|/2$, because this is the value of the fractional vertex packing where each vertex has weight $1/2$.  On the other hand, if we can partition $V = V_1 \cup V_2 \cup \cdots$ such that each graph induced by $V_i$ is a clique, then $\rho^*(V)\leq \sum_i \rho^*(V_i)=|V|/2$.  For example, for a single $k$-clique $V$ we can conclude $\rho^*(V)=k/2$ and for $Q$ in Eq.~\eqref{eq:q:triple:k:clique} we can take $V_1=[1,k]$ and $V_2=[k+1,2k]$ and conclude that $\rho^*(A_1\cdots A_{2k})=k$.}  $S=O(|E|^{k/4})$ and time $T=O(|E|^{k/2})$; the conjecture claims that it is optimal.  On the other hand, GJ computes $Q$ in $S=O(1)$ and $T=O(|E|^k)$.  This justifies a stronger version of the conjecture: for any $\varepsilon > 0$, no algorithm can compute $Q$ such that $S^2T= O(|E|^{k-\epsilon})$.

\paragraph{Aiming Toward Submodular Width}
The best known time bound for Boolean conjunctive queries is based on the \emph{submodular width}~\cite{DBLP:conf/pods/Khamis0S17, DBLP:journals/jacm/Marx13, DBLP:journals/pacmmod/KhamisHS25}.  
This measure partitions the input data based on degrees, and uses different query plans for each partition.  
We have not considered the submodular width in this paper, instead extended the (weaker) fractional hypertree width (see Thm.~\ref{thm:rptcr-fhw}).
However, extending submodular width is much more intricate as it does not generalize to arbitrary semi-rings, e.g., not to the natural numbers~$\mathbb{N}$%
~\cite{DBLP:conf/pods/KhamisCM0NOS19}.

Nevertheless, aiming at adapting our approach to achieve submodular width is promising future work. However, it at least requires the introduction of degree constraints in the analysis, which complicates the picture significantly.  Generic Join is no longer optimal in the presence of degree constraints, and this affects pseudo-trees too.  Intuitively, a pseudo-tree with structure $A_1-A_2-A_3$, rooted at $A_1$, can benefit from constraints on the degree from $A_1$ to $A_3$ but not constraints on the degree from $A_3$ to $A_1$ when analyzing its time and space. However, in some cases, it is possible to meet the submodular-width's time complexity while  minimizing the space complexity, as shown here:

\begin{restatable}{theorem}{fourcycle}    
\label{thm:4cycle}
 The  query $Q_\square$ below  can be computed in space $O(|D|^{\frac{1}{2}})$ and time $O(|D|^{\frac{3}{2}})$:
    \begin{align*}
        Q_\square() \leftarrow  E_1(A_1,A_2) \otimes E_2(A_2,A_3) \otimes E_3(A_3,A_4) \otimes E_4(A_1,A_4)
    \end{align*}
\end{restatable}

\begin{proof}[Proof Sketch]
    
    To achieve this, we first perform a heavy-light partitioning of the input relations where a \textit{heavy} join value $a_j$ is one which appears in at least $\sqrt{|E_{i}|}$ tuples of $E_{i}$.
    As there cannot be more than $\sqrt{|E_{i}|}$ heavy values, we can iterate over them, instantiate them, and solve the reaming linear query in linear time and constant space.
    Thus, in space $O(1)$ and time $O(|D|^{\frac{3}{2}})$ we can handle all heavy values.

    For the case where all values are light, a different technique is needed.
    There, the aim is to perform a merge join on the fly, essentially using the decomposition
    $$\bigoplus_{a_1,a_3}\left(\bigoplus_{a_2} E_{1}(a_1,a_2) \otimes E_{2}(a_2,a_3)\right)\otimes \left( \bigoplus_{a_4} E_{4}(a_1,a_4) \otimes E_{3}(a_4,a_3)\right).$$
    We explain how to iterate through $E_{1}(A_1,A_2) \otimes E_{2}(A_2,A_3)$ projected to $A_1,A_3$ in lexicographic order (the other side is symmetric).
    To that end, we iterate through $a_1$ at the top level and compute the $\leq \sqrt{|E_{1}|}$ values $a_2\in \supp(E_1[A_2|a_1])$ that extend $a_1$.
    For every $a_2\in \supp(E_1[A_2|a_1])$ we spawns a separate process that iterates (in an ordered manner) through $a_3\in\supp(E_2[A_3|a_2])$  -- hence we use $O(|D|^{\frac{1}{2}})$ processes that all require $O(1)$ space.
    Merging the loops of the different $a_2$ results in an ordered stream of $a_3$ values.
    Thus, we go over pairs $a_1,a_3$ in lexicographic order.
    Doing the same for $ E_{4}(a_1,a_4) \otimes E_{3}(a_4,a_3)$ allows us to perform a merge join.
    In total, this is done in space $O(|D|^{\frac{1}{2}})$ and time $O(|D|^{\frac{3}{2}})$.
\end{proof}

\section*{Acknowledgment}
The work of Merkl and Pichler was supported  by the Vienna Science and Technology Fund (WWTF) [10.47379/ICT2201, 10.47379/VRG18013, 10.47379/NXT22018]. 
Deeds and Suciu were partially supported by NSF IIS 2314527, NSF SHF
2312195, and NSF IIS 2507117.

	\bibliographystyle{abbrv}
	\bibliography{biblio}

	\appendix
	\onecolumn

\ifArxiv        %

\section{Some General Notes on the Appendix and Additional Definitions}
\label{app:general}

This appendix contains both additional examples to better understand the strength of the different formalisms and query plans as well as further proof details.
Admittedly, the appendix is at times relatively technical but we believe this to be necessary to achieve the necessary rigor.
To be able to do that, we extend our notation.
That is, for $\K$-relations $R(\boldsymbol{X})$ we write directly 
\begin{align}
    \boldsymbol{x}\in R \text{ for } \boldsymbol{x}\in \supp(R).
\end{align}
Then, further, we interpret $\supp(R)$ as a $\K$-relation and write 
\begin{align}
    \supp(R):=\{\boldsymbol{x}\mapsto \1 \mid \boldsymbol{x}\in R\}.
\end{align}
We further extend the definition $\supp(R[A|\boldsymbol{y}])$.
That is, for disjoint sets of variables $\boldsymbol{Y},\boldsymbol{Z}$ and tuple $\boldsymbol{y}\in \dom^{\boldsymbol{Y}}$, we write
\begin{align}
    \supp(R[\boldsymbol{Z}|\boldsymbol{y}]) := \{\boldsymbol{z}\mapsto \1 \mid \exists\boldsymbol{x}\in \supp(R), \boldsymbol{x}[\boldsymbol{Z}]=\boldsymbol{z}, \boldsymbol{x}[\boldsymbol{Y}]=\boldsymbol{y}\}. \label{eq:app:def:supp:pro}
\end{align}
Furthermore, we extend the definition to relations $R$ when $\boldsymbol{Y}\cup \boldsymbol{Z} \supseteq \boldsymbol{X}$ and $\boldsymbol{Y}\subsetneq \boldsymbol{X}$.
That is, we define
\begin{align}
    R[\boldsymbol{Z}|\boldsymbol{y}] = \{\boldsymbol{x}[\boldsymbol{Z}]\mapsto R(\boldsymbol{x})\mid \boldsymbol{x}\in \dom^{\boldsymbol{X}}, \boldsymbol{x}[\boldsymbol{Y}] = \boldsymbol{y}\}.
\end{align}
Note that $R[\boldsymbol{Z}| \boldsymbol{y}]$ is well defined and never a scalar.
Then, we also extend the definition of $Q[A| \boldsymbol{y}]$.
Again, $\boldsymbol{Y},\boldsymbol{Z}$ need only be disjoint sets of variables and $\boldsymbol{y}\in \dom^{\boldsymbol{X}}$.
Then, we write 
\begin{align}
    Q[{\boldsymbol{Z}}|{\boldsymbol{y}}]\leftarrow \bigotimes_{i\colon \emptyset \neq \boldsymbol{X}_i\setminus\boldsymbol{Y} \subseteq \boldsymbol{Z}}R_i[\boldsymbol{Z}| \boldsymbol{y}] \otimes \bigotimes_{i\colon \emptyset = \boldsymbol{X}_i\setminus \boldsymbol{Y}} \supp(R_i[\emptyset | \boldsymbol{y}]) \otimes \bigotimes_{i\colon \boldsymbol{X}_i\setminus\boldsymbol{Y} \not\subseteq \boldsymbol{Z}}\supp(R_i[\boldsymbol{Z}| \boldsymbol{y}]). \label{eq:app:def:q:pro}
\end{align}
We see $Q[{\boldsymbol{Z}}|{\boldsymbol{y}}]$ as a SPQ with head variables $\boldsymbol{Z}$.
Further, we also see ${Q[{\boldsymbol{Z}}|{\boldsymbol{y}}](\boldsymbol{Z'})}$ as an SPQ for subsets $\boldsymbol{Z'}\subseteq \boldsymbol{Z}$ where the head variables are simply different ones.
Thus, both $\sem{Q[{\boldsymbol{Z}}|{\boldsymbol{y}}]}$ and $\sem{Q[{\boldsymbol{Z}}|{\boldsymbol{y}}](\boldsymbol{Z'})}$ are well-defined.
Notice that for all $\boldsymbol{y}\in \dom^{\boldsymbol{Y}}$ such that for all $R_i$ there is a support $\boldsymbol{x}\in\supp(R_i[\boldsymbol{X}_i\setminus\boldsymbol{Y}| \boldsymbol{y}])$
the definition of $\sem{Q[A|\boldsymbol{y}]}$ coincides with that given in Eq. \eqref{eq:qay}.
When $\boldsymbol{Z}$ (resp. $\boldsymbol{Y}$) is the emptyset we omit $\boldsymbol{Z}$ (resp. $\boldsymbol{y}$) and the bar.

Next, we prove some generally helpful lemmas.

\subsection{Some General Lemmas}

Next we recall a lemma that is at the heart of GJ \cite{DBLP:conf/pods/NgoPRR12}.
As we will use this heavily, we provide a proof of the statement.

\begin{lemma}
\label{lem:app:gj}
Let $Q(\boldsymbol{X})\leftarrow \bigotimes_i R_i(\boldsymbol{X}_i)$ be a SPQ, $\boldsymbol{Y}\subseteq \var(Q)$ a set of variables and $A\in \var(Q)\setminus \boldsymbol{Y}$ a variable.
Then, we can iterate through $a\mapsto s \in \sem{Q[A|{\boldsymbol{y}}]}$ over all $\boldsymbol{y}\in \sem{Q[{\boldsymbol{Y}}]}$ in time $O(|D|^{\rho^*(\boldsymbol{Y}\cup \{A\})})$.
\end{lemma}

\begin{proof}

First note that 
\begin{align}
    Q[A|{\boldsymbol{y}}] \leftarrow \bigotimes_{i\colon A = \boldsymbol{X}_i\setminus \boldsymbol{Y}} R_i[A| \boldsymbol{y}]\otimes \bigotimes_{i\colon A\subsetneq \boldsymbol{X}_i \setminus \boldsymbol{Y}} \supp(R_i[A| \boldsymbol{y}])\otimes \bigotimes_{i\colon A\not\in \boldsymbol{X}_i} \supp(R_i[\boldsymbol{y}])\label{eq:app:pt2}
\end{align}
Thus, to iterate through $a\mapsto s \sem{Q[A|{\boldsymbol{y}}]}$ we can do the following:
\begin{enumerate}
    \item We can ignore the last part of the product of Eq. \eqref{eq:app:pt2} by assuming $\boldsymbol{y}\in \sem{Q[{\boldsymbol{Y}}]}$.
    \item We can iterate through the possible $a\in \sem{Q[A|{\boldsymbol{y}}]}$ by taking the smallest (cardinality wise) $R$ of the relations  $R_i[A|\boldsymbol{y}],\supp(R_i[A|\boldsymbol{y}])$ in the first and second product of Eq. \eqref{eq:app:pt2}, iterate through the domain $a \mapsto \1 \in R$ and check that $a$ is present in every other relation $R_i[A|\boldsymbol{y}],\supp(R_i[A|\boldsymbol{y}])\neq R$ in the first and second product of Eq. \eqref{eq:app:pt2}.
    \item The annotation $s$ of $a$ in $\sem{Q[A|{\boldsymbol{y}}]}$ is then the product of the annotations of $a$ in the relations $R_i[A|\boldsymbol{y}]$ in the first product of Eq. \eqref{eq:app:pt2}
\end{enumerate}
Iterating through $\sem{Q[A|{\boldsymbol{y}}]}$ in this way allows us to do it in time (look-ups take logarithmic time which we ignore)
\begin{align}O(\min_{i\colon A\subseteq \boldsymbol{X}_i \setminus \boldsymbol{Y}} |\supp(R_i[A|\boldsymbol{y}])). \label{eq:app:gjupto}
\end{align}
Note the index $i$ goes over the relations of the first and second product in Eq. \eqref{eq:app:pt2}.
Thus, to compute the overall time, we simply have to sum this up over all $\boldsymbol{y}\in \sem{Q[{\boldsymbol{Y}}]}$.
To that end, we proof that the claim of the lemma holds for an arbitrary fractional edge cover and, thus, it also holds for the minimal.
Therefore, let $(w_i)_i$ be a fractional edge cover of $\boldsymbol{Y}\cup \{A\}$.
Then, clearly, $\forall B\in \boldsymbol{Y}\colon \sum_{i\colon B\in \boldsymbol{X}_i} w_i\geq 1$ and $\sum_{i\colon A\in \boldsymbol{X}_i} w_i\geq 1$.
Furthermore, let us fix an arbitrary order $\boldsymbol{Y} = B_1,\dots,B_l$ and the shorthand $\boldsymbol{y}_j^{j'}:=(b_j,\dots, b_{j'}))$ for $j\leq j'$.
{\small
\begin{align}
   &\sum_{\boldsymbol{y}\in \sem{Q[{\boldsymbol{Y}}]}}\min_{i\colon A\subseteq \boldsymbol{X}_i \setminus \boldsymbol{Y}} |\supp(R_i[A|\boldsymbol{y}])| \\
   &\leq \sum_{\boldsymbol{y}\in \sem{Q[{\boldsymbol{Y}}]}}\prod_{i\colon A\subseteq \boldsymbol{X}_i \setminus \boldsymbol{Y}}|\supp(R_i[A|\boldsymbol{y}])|^{w_i} \label{eq:gm}\\
   &\leq \sum_{\boldsymbol{y}\in \sem{Q[{\boldsymbol{Y}}]}}\prod_{i}|\supp(R_i[A|\boldsymbol{y}])|^{w_i} \label{eq:add1s}\\
   & \leq \sum_{b_1\in \sem{Q[{B_1}]}} \sum_{b_2\in \sem{Q[B_2|{\boldsymbol{y}^1_1}]}}\dots \sum_{b_l\in \sem{Q[B_l|\boldsymbol{y}_1^{l-1}]}}\prod_{i}|\supp(R_i[A|\boldsymbol{y}_1^{l}])|^{w_i} \\
   & = \sum_{b_1\in \sem{Q[{B_1}]}} \dots \sum_{b_{l-1}\in \sem{Q[B_{l-1}|\boldsymbol{y}_1^{l-2}]}} \prod_{i\colon B_l\not\in \boldsymbol{X}_i}|\supp(R_i[A|\boldsymbol{y}_1^{l-1}])|^{w_i} \sum_{b_l\in \sem{Q[B_l|\boldsymbol{y}_1^{l-1}]}}\prod_{i\colon B_l\in \boldsymbol{X}_i}|\supp(R_i[A|\boldsymbol{y}_1^{l}])|^{w_i} \\
   & \leq \sum_{b_1\in \sem{Q[{B_1}]}} \dots \sum_{b_{l-1}\in \sem{Q[B_{l-1}|\boldsymbol{y}_1^{l-2}]}} \prod_{i\colon B_l\not\in \boldsymbol{X}_i}|\supp(R_i[A|\boldsymbol{y}_1^{l-1}])|^{w_i} \prod_{i\colon B_l\in \boldsymbol{X}_i}(\sum_{b_l\in \sem{Q[B_l|\boldsymbol{y}_1^{l-1}]}}|\supp(R_i[A|\boldsymbol{y}_1^{l}])|)^{w_i} \label{eq:hölder}\\
   & \leq \sum_{b_1\in \sem{Q[{B_1}]}} \dots \sum_{b_{l-1}\in \sem{Q[B_{l-1}|\boldsymbol{y}_1^{l-2}]}} \prod_{i\colon B_l\not\in \boldsymbol{X}_i}|\supp(R_i[A|\boldsymbol{y}_1^{l-1}])|^{w_i} \prod_{i\colon B_l\in \boldsymbol{X}_i}|\supp(R_i[AB_l|\boldsymbol{y}_1^{l-1}])|^{w_i} \\
   & = \sum_{b_1\in \sem{Q[{B_1}]}} \dots \sum_{b_{l-1}\in \sem{Q[B_{l-1}|\boldsymbol{y}_1^{l-2}]}} \prod_{i}|\supp(R_i[AB_l|\boldsymbol{y}_1^{l-1}])|^{w_i} \label{eq:from}\\
   & \vdots \nonumber\\
   & \leq \sum_{b_1\in \sem{Q[{B_1}]}} \prod_{i}|\supp(R_i[AB_2\cdots B_l|b_1])|^{w_i} \label{eq:to}\\
   & = \prod_{i\colon B_1\not \in X_i}|\supp(R_i[AB_2\cdots B_l)|^{w_i} \sum_{b_1\in \sem{Q[{B_1}]}} \prod_{i\colon B_1 \in X_i}|\supp(R_i[AB_2\cdots B_l|b_1])|^{w_i}\\
   & \leq \prod_{i\colon B_1\not \in X_i}|\supp(R_i[AB_2\cdots B_l)|^{w_i} \prod_{i\colon B_1 \in X_i}(\sum_{b_1\in \sem{Q[{B_1}]}} |\supp(R_i[AB_2\cdots B_l|b_1])|)^{w_i} \label{eq:hölder2}\\
   & \leq \prod_{i\colon B_1\not \in X_i}|\supp(R_i[AB_2\cdots B_l)|^{w_i} \prod_{i\colon B_1 \in X_i} |\supp(R_i[AB_1\cdots B_l)|^{w_i}\\
   & = \prod_{i}|\supp(R_i[AB_1\cdots B_l)|^{w_i}\\
   & \leq \prod_{i} |D|^{w_i} \label{eq:fullTable} \\   
   & \leq |D|^{\sum_iw_i} 
\end{align}
}
Note that the Inequality $\eqref{eq:gm}$ is due to the geometric mean being a bound on a minimum.
In, Equality $\eqref{eq:add1s}$, we simply add the remaining relations that do not contain $A$.
Inequalities $\eqref{eq:hölder}$ and \eqref{eq:hölder2} are due to Hölder.
To get form line \eqref{eq:from} to \eqref{eq:to}, we simply proceed as was done for $B_l$ until we arrive at $B_1$.
The argument for $B_1$ is the same as for all the other $B_j$'s but we repeat it for clarity.
In Inequality \eqref{eq:fullTable}, we simply bound the sizes of the subrelations with the whole database.
(Note that these are precisely the arguments used in generic join.)    
\end{proof}

Next we prove a lemma that is at the heart of why a PT can help us evaluate SPQs.

\begin{lemma}
\label{lem:app:ptdecompose} 
    Let $Q(\boldsymbol{X})\leftarrow \bigotimes_i R_i(\boldsymbol{X}_i)$ be a SPQ, $P$ a PT of $Q(\boldsymbol{X})$ and $A\in \var(Q)$ arbitrary.
    Then, for an $\boldsymbol{x}\in \dom^{\anc(A)}\colon$ 
    {\small
    \begin{align}
        \sem{Q[{\descc(A)}|{\boldsymbol{x}}](\outt(A))}=\bigoplus_{a\mapsto s\in \sem{Q[A|{\boldsymbol{x}}]}} \{a[\boldsymbol{X}]\mapsto s\}\otimes\bigotimes_{B\in child(A)} \sem{Q[{\descc(B)}|{(\boldsymbol{x},a)}](\outt(B))}. \label{eq:app:decom}
    \end{align}
    }
\end{lemma}

\begin{proof}
Let us denote the children as $\{B_1,\dots, B_l\}=child(A)$.
Further, let us consider an $R_i(\boldsymbol{X}_i)$ with $\descc(A)\cap \boldsymbol{X}_i\neq \emptyset$.
Then, note that due to the property of a PT, either $\boldsymbol{X}_i\subseteq \ancc(A)$ or for exactly one $j=1,\dots,l$ we have $\descc(B_j)\cap \boldsymbol{X}_i \neq \emptyset$.
Thus, we have
{\small
\begin{align}
    Q&[{\descc(A)}|{\boldsymbol{x}}](\outt(A))\nonumber\\
    &\leftarrow {\bigotimes_{i \colon \substack{\descc(A) \cap \boldsymbol{X}_i\neq \emptyset \\ \land \boldsymbol{X}_i \subseteq \ancc(A)}}R_i[A|\boldsymbol{x}]}\otimes \bigotimes_{j}{\bigotimes_{i \colon {\descc(B_j)\cap \boldsymbol{X}_i \neq \emptyset}}R_i[\descc(B_j)\cup \{A\}|\boldsymbol{x}]}  \otimes{\bigotimes_{i\colon \descc(A)\cap \boldsymbol{X}_i =\emptyset}\supp(R_i[\boldsymbol{x}]}).\label{eq:app:ptdecompose}
\end{align}
}

Note that $\sem{Q[{\descc(A)}|{\boldsymbol{x}}](\outt(A))}=\{\boldsymbol{y}\mapsto \0\mid \boldsymbol{y}\}$ if $\boldsymbol{x}\not\in \sem{Q[\boldsymbol{X}]}$ and the same for the right hand side of Eq. \eqref{eq:app:decom} due to convention.
Thus, let us consider the case $\boldsymbol{x}\in \sem{Q[\boldsymbol{X}]}$.
Then, we can ignore the relations in the last product of Eq. \eqref{eq:app:ptdecompose}.

Let us consider the cases $A\in \boldsymbol{X}$ and $A\not\in \boldsymbol{X}$ individually.

\textbf{Case} $A\in \boldsymbol{X}$:
Then, for $\boldsymbol{y}\in \dom^{\outt(A)}$, we can compute $\sem{Q[{\descc(A)}|{\boldsymbol{x}}](\outt(A))}(\boldsymbol{y})=$
\begin{align*}
    {\bigotimes_{i \colon \substack{\descc(A) \cap \boldsymbol{X}_i\neq \emptyset \\ \land \boldsymbol{X}_i \subseteq \ancc(A)}}R_i(\boldsymbol{x}[\boldsymbol{X}_i],\boldsymbol{y}[A])}\otimes \bigotimes_{j}\bigoplus_{\boldsymbol{y}_j\in \dom^{{\descc(B_j)}}}{\bigotimes_{i \colon {\descc(B_j)\cap \boldsymbol{X}_i \neq \emptyset}}R_i((\boldsymbol{x},\boldsymbol{y},\boldsymbol{y}_j)[\boldsymbol{X}_i])}.
\end{align*}
However, we can rewrite this to

\begin{align*}
    \sem{Q[A|\boldsymbol{x}]}(\boldsymbol{y}[A])\otimes \bigotimes_{j} \sem{Q[\descc(B_j)|(\boldsymbol{x}, \boldsymbol{y}[A])]}(\boldsymbol{y}[\outt(B)]).
\end{align*}

Note that this just introduces additions and redundant checks if $(\boldsymbol{x},\boldsymbol{y}[A])$ is in the support of all relations.
Thus, this completes this case.

\textbf{Case} $A\not\in \boldsymbol{X}$:
Then, for $\boldsymbol{y}\in \dom^{\outt(A)}$ ($A\not\in \outt(A)$), we can compute $\sem{Q[{\descc(A)}|{\boldsymbol{x}}](\outt(A))}(\boldsymbol{y})=$
\begin{align*}
    \bigoplus_{a\in \dom^{A}}{\bigotimes_{i \colon \substack{\descc(A) \cap \boldsymbol{X}_i\neq \emptyset \\ \land \boldsymbol{X}_i \subseteq \ancc(A)}}R_i(\boldsymbol{x}[\boldsymbol{X}_i],a)}\otimes \bigotimes_{j}\bigoplus_{\boldsymbol{y}_j\in \dom^{{\descc(B_j)}}}{\bigotimes_{i \colon {\descc(B_j)\cap \boldsymbol{X}_i \neq \emptyset}}R_i((\boldsymbol{x},a,\boldsymbol{y},\boldsymbol{y}_j)[\boldsymbol{X}_i])}.
\end{align*}
However, we can rewrite this to

\begin{align*}
    \bigoplus_{a\in \sem{Q[A|\boldsymbol{x}]}}\sem{Q[A|\boldsymbol{x}]}(a)\otimes \bigotimes_{j} \sem{Q[\descc(B_j)|(\boldsymbol{x},a)]}(\boldsymbol{y}[\outt(B)]).
\end{align*}

Note that this just introduces additions and redundant checks if $(\boldsymbol{x},a)$ is in the support of all relations.
Thus, this completes this case and with it the whole proof.
\end{proof}

Next we prove a lemma that explains how considering the context variables helps us avoid unnecessary work.

\begin{lemma}
\label{lem:app:con}
    Let $Q$ be a SPQ, $(P,\boldsymbol{C})$ a PTC of $Q$, and $A\in \var(Q)$ an arbitrary variable.
    Then, for any set of variables $\boldsymbol{Y}$ such that $\con(A)\subseteq \boldsymbol{Y}\subseteq \var(Q)\setminus \descc(A)$ and tuple $\boldsymbol{y}\in \sem{Q[{\boldsymbol{Y}}]}$
    we have $\sem{Q[{\descc(A)}|{\boldsymbol{y}}]}=\sem{Q[{\descc(A)}|{\boldsymbol{y}[\con(A)]}]}$.
\end{lemma}

\begin{proof}

First, let us consider the structure of $Q[{\descc(A)}|{\boldsymbol{y}}]\leftarrow$
{\small
\begin{align}
    \bigotimes_{i \colon \substack{\emptyset \neq \boldsymbol{X}_i\setminus \boldsymbol{Y},\\
     \boldsymbol{X}_i\setminus \boldsymbol{Y} \subseteq \descc(A)}}     R_i[\descc(A)|\boldsymbol{y}]
    \otimes \bigotimes_{i \colon \emptyset = \boldsymbol{X}_i\setminus \boldsymbol{Y}}\supp(R_i[\boldsymbol{y}])
    \otimes\bigotimes_{i\colon \boldsymbol{X}_i\setminus \boldsymbol{Y} \not\subseteq \descc(A)}\supp(R_i[\descc(A)|\boldsymbol{y}]).\label{eq:app:ptc}
\end{align}    
}

Then, let us consider a relation $R_i$ considered in the last product of Eq. \eqref{eq:app:ptc} and assume there exists a $C\in \boldsymbol{X}_i\cap \descc(A)$.
By definition, there is a $B\in \boldsymbol{X}_i\setminus (\boldsymbol{Y}\cup \descc(A))$.
Then, by the definition of a PT, $B\in \anc(A)$ and by the definition of context $B\in \con(A)$.
Thus, actually, such a $B$ cannot exist and as a consequence, $\boldsymbol{X}_i\cap \descc(A)=\emptyset$ for every relation $R_i$ in the last product of Eq. \eqref{eq:app:ptc}.
Thus, we can rewrite it to

\begin{align}
    Q[{\descc(A)}|{\boldsymbol{y}}] & \leftarrow 
    \bigotimes_{i \colon \substack{\emptyset \neq \boldsymbol{X}_i\setminus \boldsymbol{Y},\\
     \boldsymbol{X}_i\setminus \boldsymbol{Y} \subseteq \descc(A)}}     R_i[\descc(A)|\boldsymbol{y}]
    \otimes \bigotimes_{i \colon  \substack{\emptyset = \boldsymbol{X}_i\setminus \boldsymbol{Y} \\ \lor \boldsymbol{X}_i\setminus \boldsymbol{Y} \not\subseteq \descc(A)}}\supp(R_i[\boldsymbol{y}]). \label{eq:app:ptcy1}
\end{align}  

Further, let us consider the structure of $Q[{\descc(A)}|{\boldsymbol{y}[\con(A)]}]$ (note that $\con(A)$ could have also been such a set of variables $\boldsymbol{Y}$):

\begin{align}
    Q&[{\descc(A)}|{\boldsymbol{y}[\con(A)]}]  \leftarrow \nonumber \\
    &\bigotimes_{i \colon \substack{\emptyset \neq \boldsymbol{X}_i\setminus \con(A), \\
    \boldsymbol{X}_i\setminus \con(A) \subseteq \descc(A)}}     R_i[\descc(A)|\boldsymbol{y}[\con(A)]]
    \otimes \bigotimes_{i\colon \substack{\emptyset = \boldsymbol{X}_i\setminus \con(A) \\ \lor \boldsymbol{X}_i\setminus \con(A) \not\subseteq \descc(A)}}\supp(R_i[\boldsymbol{y}[\con(A)]]). \label{eq:app:ptccon1}
\end{align}  

First, note that for any $i$ the relation $\supp(R_i[\boldsymbol{y}[\con(A)]])$ is the same as $\supp(R_i[\boldsymbol{y}])$ since both are $\{()\mapsto \1\}$ as $\boldsymbol{y}\in \sem{Q[{\boldsymbol{Y}}]}$.
Moreover, if $\boldsymbol{X}_i\cap \boldsymbol{Y}\subseteq \con(A)$, clearly $R_i[\descc(A)|\boldsymbol{y}[\con(A)]]$ is the same as $R_i[\descc(A)|\boldsymbol{y}]$.
Hence, if we can show that Eq. \eqref{eq:app:ptcy1} and \eqref{eq:app:ptccon1} actually partition the $i$'s in the same way and that $\boldsymbol{X}_i\cap \boldsymbol{Y}\subseteq \con(A)$ always holds for relations in the first part, we can be sure that the expressions are actually the same.

To that end, consider a $i$ such that $\emptyset \neq \boldsymbol{X}_i\setminus \boldsymbol{Y} \subseteq \descc(A)$.
Then, similarly to above, using the definition of PTs and contexts, we get $\boldsymbol{X}_i \subseteq \con(A)\cup \descc(A)$.
Hence, $\boldsymbol{X}_i\cap \boldsymbol{Y}\subseteq \con(A)$ and $\emptyset \neq \boldsymbol{X}_i \setminus \con(A) \subseteq \descc(A)$.

Proving the converse is even simpler.
To that end, consider a $i$ such that $\emptyset \neq \boldsymbol{X}_i \setminus \con(A) \subseteq \descc(A)$.
Hence,  $\boldsymbol{X}_i \subseteq \con(A)\cup \descc(A)$ and $\emptyset \neq \boldsymbol{X}_i \setminus \boldsymbol{Y} \subseteq \descc(A)$.

As the expressions are the same, naturally, $\sem{Q[{\descc(A)}|{\boldsymbol{y}}]}=\sem{Q[{\descc(A)}|{\boldsymbol{y}[\con(A)]}]}$.
\end{proof}

\fi 
\ifArxiv

\section{Appendix for Section \ref{sec:pt}}

Before proving Theorem \ref{thm:algo:pt}, we briefly showcase the computation of the space-time exponent of a PT.

\subsection{Additional example of a plan in $\pt$}

Let us reconsider the example used in Section \ref{sec:pt} which can also be seen in Figures~\ref{fig:app:ptQuery} and \ref{fig:app:pt}.
However, this time we added some output variables marked in blue.
In Figure \ref{fig:app:pt}, we added variables $\outt$ to each node, as well as the quantities $\rho^*(\outt), \rho^*(\ancc \cup \out)$ as these are important to compute the space-time exponent.
That is, the space-time exponent of this PT is simply the max over these quantities (component wise), i.e., $(2,2) = (\rho^*(BCF), \rho^*(BCF))$ achieved at the root $B$. 
Note that this is ``unavoidable'' as the fractional edge cover number of the output variables is already 2 and hence, the output can be up to quadratic in size.
However, for the subset $Q(B,F)$, the same PT would have the space-time exponent $(1,1.5)$. 
The only changes would be that for $B$ we get the annotations $(BF,1,1)$ and for $C$ we get the annotations $(\emptyset, 0, 1)$.

\begin{figure}

\begin{minipage}[b]{0.42\textwidth}
\begin{figure}[H]
    \centering
    \begin{tikzpicture}[font=\footnotesize]
            \node (A) at (-2,2) {$A$};
            \node (B) at (0,2) {$B$};
            \node (C) at (2,2) {$C$};
            \node (D) at (-2,0) {$D$};
            \node (E) at (0,0) {$E$};
            \node (F) at (2,0) {$F$};

            \draw[fill=blue,opacity=0.5, fill opacity=0.2] \convexpath{B,C,F}{.38cm};

            \draw (A) -- (B) node[midway, above] {$R_1$};
            \draw (B) -- (C) node[midway, above] {$R_2$};
            \draw (B) -- (D) node[midway, left] {$R_3$};
            \draw (B) -- (E) node[midway, right] {$R_4$};
            \draw (B) -- (F) node[midway, right] {$R_5$};
            \draw (D) -- (E) node[midway, below] {$R_6$};
            \draw (E) -- (F) node[midway, below] {$R_7$};

    \end{tikzpicture}
    \caption{Query $Q(B,C,F)$}
    \label{fig:app:ptQuery}
\end{figure}
\end{minipage}
\hfill
\begin{minipage}[b]{0.57\textwidth}
\begin{figure}[H]
    \centering
    \begin{tikzpicture}[font=\footnotesize]
    \node (B) {$B$}
        child {node (A) {$A$}}
        child { node (E) {$E$}
            child {node (D) {$D$}}
            child {node (F) {$F$}}
        }
        child {node (C) {$C$}};         

        \node[anchor=east] at ($(A) + (-.1,0)$) (EXTDPT) { $\emptyset, 0, 1$  };
        \node[anchor=west] at ($(B) + (.1,0)$) (EXTDPT) { $BCF, 2, 2$  };
        \node at ($(B) + (0,.5)$) (EXTDPT) { $\outt, \rho^*(\outt), \rho^*(\ancc \cup \out)$  };
        \node[anchor=west] at ($(C) + (.1,0)$) (EXTDPT) { $C, 1, 1$  };
        \node[anchor=east] at ($(D) + (-.1,0)$) (EXTDPT) { $\emptyset, 0, 1.5$  };
        \node[anchor=west] at ($(E) + (.1,0)$) (EXTDPT) { $F, 1, 1.5$  };
        \node[anchor=west] at ($(F) + (.1,0)$) (EXTDPT) { $F, 1, 1.5$  };

    \end{tikzpicture}
    \caption{A PT of $Q(B,C,F)$}
    \label{fig:app:pt}
\end{figure}
\end{minipage}

\end{figure}

Now, let us proceed towards a proof of Theorem \ref{thm:algo:pt}.
To that end, the following lemmas will be helpful.

\subsection{Lemmas for Theorem \ref{thm:algo:pt}}

\begin{lemma}
\label{lem:app:pt1}
    Let $P$ be a PT of the SPQ $Q$ and $A\in \var(Q)$ arbitrary.
    Then, for an $\boldsymbol{x}\in \dom^{\anc(A)}$, the function $\textsc{solve}(A,\boldsymbol{x})$ of Algorithm \ref{alg:pt} returns $\sem{Q[{\descc(A)}|{\boldsymbol{x}}](\outt(A))}$.
\end{lemma}

\begin{proof}

We prove this by induction on the size of $\descc(A)$.
The base case is $\descc(A) = \{A\}$ which we will consider now.
What Algorithm \ref{alg:pt} does to compute $\sem{Q[{A}|{\boldsymbol{x}}](\outt(A))}$ is to allocate a $\K$-relation $\text{OUT}$ over the same schema and iteratively considers $a\mapsto s \in \sem{Q[A|{\boldsymbol{x}}]}$.
As $\desc(A)=\emptyset$, the $\K$-relation $\text{TMP}$ allocated next is then also of the same schema.
Thus, Algorithm \ref{alg:pt}
simply computes 
\begin{align*}
    \text{OUT} = \bigoplus_{a\mapsto s \in \sem{Q[A|\boldsymbol{x}]}}\{a[\outt(A)]\mapsto s\}
\end{align*}
Due to Lemma \ref{lem:app:ptdecompose}, this equals $\sem{Q[A|\boldsymbol{x}](\outt(A))}$ as required.

Now for the induction step, let us consider an $A$ such that $\desc(A)\neq \emptyset$.
What Algorithm \ref{alg:pt} does to compute $\sem{Q[{A}|{\boldsymbol{x}}](\outt(A))}$ is to allocate a $\K$-relation $\text{OUT}$ over the same schema and iteratively considers $a\mapsto s \in \sem{Q[A|{\boldsymbol{x}}]}$.
Then, $\text{TMP}$ is instantiated as the relation $\{a[\outt(A)]\mapsto s\}$.
Depending on whether $A\in \outt(A)$, the schema is either $A$ or $\emptyset$.
In either case, we iterative multiply to it the calls to the children $B\in child(A)$
Thus, in the end
\begin{align*}
    \text{OUT} = \bigoplus_{a\mapsto s \in \sem{Q[A|\boldsymbol{x}]}}\{a[\outt(A)]\mapsto s\} \otimes \bigotimes_{B\in child(A)} \textsc{solve}(B,(\boldsymbol{x},a))
\end{align*}
Due to the induction hypothesis and Lemma \ref{lem:app:ptdecompose}, this equals $\sem{Q[\descc(A)|\boldsymbol{x}](\outt(A))}$ as required.
\end{proof}

\begin{lemma}
\label{lem:app:pt2}
    Let $P$ be a PT of the SPQ $Q$ and $A\in \var(Q)$ arbitrary.
    Then, for a run of Algorithm~\ref{alg:pt}, the function $\textsc{solve}(A,\boldsymbol{x})$ is called at most once per argument pair $(A,\boldsymbol{x})$ and $\boldsymbol{x}\in \sem{Q[{\anc(A)}]}$.
\end{lemma}

Note that to be completely rigorous, we would have to take spacial care of the pathological case where a relation $R_i(\boldsymbol{X}_i)$ is empty ($R_i = \{\boldsymbol{x}\mapsto \0 \mid \boldsymbol{x}\in \dom^{\boldsymbol{X}_i}\}$) as we call $\textsc{solve}(\ptroot(P), ())$ even though $\sem{Q[{\emptyset}]}=\{()\mapsto \0\}$.
Of course this would be possible but at the same time it would just make the discussion unnecessarily more cumbersome. 
Thus, we simply disregard this pathological case.

\begin{proof}[Proof of Lemma \ref{lem:app:pt2}]
    First notice that the first call is made to $\ptroot(P)$ together with the only possible tuple $()\in \dom^{\anc(\ptroot(P))}=\dom^\emptyset.$
    Then, the algorithm iterates through $a\in \sem{Q[{\ptroot}]}$ and for every $A\in child(\ptroot(P))$ calls $\textsc{solve}(A,a)$.

    Then, we can argue inductively.
    To that end, let $A$ be such that the induction hypothesis holds.
    I.e., all calls gets a tuple $\boldsymbol{x}$ together with a variable $A$ as input such that $\boldsymbol{x}\in \sem{Q[{\anc(A)}]}$ and such that $(A,\boldsymbol{x})$ is unique.
    Then the call simply iterates through domain elements $a\in \sem{Q[A|{\boldsymbol{x}}]}$ and extends $\boldsymbol{x}$ by that element to get $\boldsymbol{x}'=(\boldsymbol{x},a)$.
    Thus, by construction $\boldsymbol{x}'\in \sem{Q[{\ancc(A)}]}$.
    Then, we call $\textsc{solve}(B,\boldsymbol{x}')$ for every $B\in child(A)$.
    Note that $\anc(B)=\ancc(A)$ and, thus, $\boldsymbol{x}'\in \sem{Q[{\anc(B)}]}$.
    Importantly, domain elements are never removed from the input tuple and, thus, we can trace back the call of $\textsc{solve}(B,\boldsymbol{x}')$ to $\textsc{solve}(parent(B),\boldsymbol{x}'[\anc(parent(B))])$.
    Hence, as $(A,\boldsymbol{x})$ was unique, also $(B,\boldsymbol{x}')$ is unique.
\end{proof}

\begin{lemma}
\label{lem:app:pt3}
    Let $P$ be a PT of the SPQ $Q$ and $A\in \var(Q)$ arbitrary.
    Then, for a run of the Algorithm \ref{alg:pt}, the collective time spent in function calls $\textsc{solve}(A,\boldsymbol{x})$ (excluding the time spent in recursive calls made in line \ref{line:pt:recurse} that exit the function call) over all values for $\boldsymbol{x}$ is bounded by $O(|D|^{\rho^*(\ancc(A)\cup \out(A))})$.
\end{lemma}

\begin{proof}
Due to Lemma \ref{lem:app:pt2}, we only have to sum up the time spent in call $\textsc{solve}(A,\boldsymbol{x})$ over all $\boldsymbol{x}\in \sem{Q[{\anc(A)}]}$.
For each call, we simply have to go through all $a\mapsto s\in \sem{Q[A|{\boldsymbol{x}}]}$ and then, as we have seen in the proof of Lemma \ref{lem:app:pt1}, compute $\sem{Q[{\descc(A)}|{\boldsymbol{x}}](\outt(A))}$ by computing the huge sum
$$\sem{Q[{\descc(A)}|{\boldsymbol{x}}](\out(A) \cup \{A\})} = \sem{Q[A|{\boldsymbol{x}}]} \otimes \bigotimes_{B\in child(A)} \sem{Q[{\descc(B)}|{(\boldsymbol{x},a)}] (\outt(B))},$$
Possibly also marginalizing away the values for $A$ when $A\not\in \outt(A)$.
However, each annotated tuple in the sum is unique and as we materialize $\text{TMP}$ and $\text{OUT}$, the time for the computation can be bound by\footnote{Formally, for this to also work in the case that a set is empty, we simpy have to replace the $|\cdot|$ with $\max(|\cdot|, 1)$}
$$O(|\sem{Q[A|{\boldsymbol{x}}](A)}|   \prod_{B\in child(A)} |\sem{Q[{\descc(B)}|{(\boldsymbol{x},a)}] (\outt(B))}|),$$
i.e., in \textit{output optimal time} where the output is $\{A\}\cup\out(A)$.
Thus, over all $\boldsymbol{x}\in \sem{Q[{\anc(A)}]}$, the time required is $O(|D|^{\rho^*(\anc(A)\cup \out(A))})$ (due to $O(|D|^{\rho^*(\anc(A)\cup \out(A))})$ being a bound on the size of $\sem{Q[{\ancc(A)\cup\desc(A)}](\ancc(A)\cup \out(A))}$; see \cite{DBLP:conf/focs/AtseriasGM08}) while the time to iterate through $a\mapsto s\in \sem{Q[A]}$ is bound by $O(|D|^{\rho^*(\ancc(A))})$ due to Lemma \ref{lem:app:gj}.
\end{proof}

Now, let us proof Theorem \ref{thm:algo:pt}

\subsection{Proof of Theorem \ref{thm:algo:pt}}

\thmalgopt*
\begin{proof}
The correctness follows directly by applying Lemma \ref{lem:app:pt1} to the call $\textsc{solve}(\ptroot,())$ which consequently returns $\sem{Q(\boldsymbol{X})}$ as $\descc(\ptroot)=\var(Q)$ and $\outt(\ptroot)=\boldsymbol{X}$.

To bound the runtime of Algorithm \ref{alg:pt}, we can attribute each timestep to the current call of the function $\textsc{solve}(A,\boldsymbol{x})$.
To that end, for each variables $A$ we sum up the time spent in function calls $\textsc{solve}(A,\boldsymbol{x})$ (excluding the time spent in recursive calls made in line 12 that exit the function call) over all values for $\boldsymbol{x}$.
Due to Lemma \ref{lem:app:pt3}, this can be bound by $O(|D|^{\rho^*(\ancc(A)\cup \out(A))})$.
Thus, by taking the maximum over all $A$, we arrive at a overall bound of $O(|D|^{t(P)})$ for Algorithm \ref{alg:pt}.

Lastly, the bound on the space consumption of Algorithm \ref{alg:pt} is rather immediate.
First, note that the recursion depth is constant (bounded by $|\var(Q)|$).
Second, note that for every variable $A$, the $\K$-relations stored in calls to $\textsc{solve}(A,\boldsymbol{x})$ are essentially subsets of the results of $\sem{Q[{\outt(A)}]}$ and, thus, their sizes never exceed $O(|D|^{\rho^*(\outt(A))})$.
Moreover, $\outt(A)\subseteq \outt(\ptroot(P)) = \boldsymbol{X}$.
\end{proof}

\fi 
\ifArxiv

\section{Appendix for Section \ref{sec:td}}
\label{app:td}

Before proving Theorem~\ref{thm:tdgj-pt},  we briefly showcase how TDs interact with classes of query plans $\calC$ to produce $\tdp$.

\subsection{Additional Example of a Plan in $\tdpt$}
\label{sec:app:tdptexample}

We consider the query $Q(E,F,H)$ and plan in $\tdpt$ depicted in Fig.~\ref{fig:app:tdptQuery} and \ref{fig:app:tdptexample}.
Note that in comparison with the query in Fig.~\ref{fig:tdptQuery}, we added some output variables $EFH$ and removed the $GH$ edge.
In Fig. \ref{fig:app:tdptexample}, we added variables $\outt$ to each node in each PT, as well as the quantities $\rho^*(\outt), \rho^*(\ancc \cup \out)$ as these are important to compute the space-time exponent - both of the individual PTs and the overall TD with the PTs.
That is, the space-time exponent of this structure is simply the max over these quantities (component wise), i.e., $(2,2) = (\rho^*(GH), \rho^*(GHI))$ achieved at the root $I$ of the right PT.

\begin{figure}
\begin{figure}[H]
\begin{center}
    \begin{tikzpicture}[scale = 0.45,font=\footnotesize]
        \node (A) at (-2,0) {$A$};
        \node (B) at (-1,2) {$B$};
        \node (C) at (0,0) {$C$};
        \node (D) at (1,2) {$D$};
        \node (E) at (2,0) {$E$};
        \node (F) at (3,2) {$F$};
        \node (G) at (4,0) {$G$};
        \node (H) at (5,2) {$H$};
        \node (I) at (6,0) {$I$};
        \node (J) at (7,2) {$J$};
        \node (K) at (8,0) {$K$};

        \draw[fill=blue,opacity=0.5, fill opacity=0.2] \convexpath{E,F,H}{.5cm};

        \draw (A) -- (B);
        \draw (A) -- (C);
        \draw (A) -- (D);
        \draw (B) -- (C);
        \draw (B) -- (D);
        \draw (B) -- (E);
        \draw (C) -- (D);
        \draw (C) -- (E);
        \draw (D) -- (E);
        \draw (D) -- (F);
        \draw (D) -- (G);
        \draw (E) -- (F);
        \draw (E) -- (G);
        \draw (E) -- (H);
        \draw (F) -- (G);
        \draw (F) -- (H);
        \draw (I) -- (G);
        \draw (I) -- (H);
        \draw (I) -- (J);
        \draw (I) -- (K);
        \draw (J) -- (G);
        \draw (J) -- (H);
        \draw (J) -- (K);
        \draw (H) -- (K);
    \end{tikzpicture}
\end{center}
\caption{Query $Q(E,F,H)$}
    \label{fig:app:tdptQuery}
\end{figure}
\begin{figure}[H]

\begin{center}
    \begin{tikzpicture}[scale = 0.9, font=\footnotesize]
        
        \node[rectangle, draw] (ABCDE) at (0,0) {
        \begin{tikzpicture}
            \node (C) {$C$}
                child { node (B) {$B$}
                    child { node (D) {$D$}
                        child {node (A) {$A$}}
                        child {node (E) {$E$}}
                    }
                };
                
            \node[anchor=west] at ($(C) + (.1,0)$) (EXTDPT) { $DE, 1, 1.5$  };
            \node[anchor=west] at ($(B) + (.1,0)$) (EXTDPT) { $DE, 1, 2$  };
            \node[anchor=west] at ($(D) + (.1,0)$) (EXTDPT) { $DE, 1, 2$  };
            \node[anchor=west] at ($(A) + (.1,0)$) (EXTDPT) { $\emptyset, 0, 2$  };
            \node[anchor=west] at ($(E) + (.1,0)$) (EXTDPT) { $E, 1, 2$  };
        \end{tikzpicture}
        };
        \node[rectangle, draw] (DEFGH) at (5,2) {
        \begin{tikzpicture}
            \node (F2) {$F$}
                child { node (E2) {$E$}
                    child { node (G2) {$G$}
                        child {node (D2) {$D$}}
                        child {node (H2) {$H$}}
                    }
                };

            \node[anchor=west] at ($(F2) + (.1,0)$) (EXTDPT) { $EFH, 1.5, 1.5$  };
            \node[anchor=west] at ($(E2) + (.1,0)$) (EXTDPT) { $EH, 1, 1.5$  };
            \node[anchor=west] at ($(G2) + (.1,0)$) (EXTDPT) { $H, 1, 2$  };
            \node[anchor=west] at ($(D2) + (.1,0)$) (EXTDPT) { $\emptyset, 0, 2$  };
            \node[anchor=west] at ($(H2) + (.1,0)$) (EXTDPT) { $H, 1, 2$  };
        \end{tikzpicture}
        };
        \node[rectangle, draw] (GHIJK) at (10,0) {
        \begin{tikzpicture}
            \node (I) {$I$}
                child { node (H) {$H$}
                    child { node (J) {$J$}
                        child {node (G) {$G$}}
                        child {node (K) {$K$}}
                    }
                };
                    
            \node[anchor=west] at ($(I) + (.1,0)$) (EXTDPT) { $GH, 2, 2$  };
            \node[anchor=west] at ($(H) + (.1,0)$) (EXTDPT) { $GH, 2, 2$  };
            \node[anchor=west] at ($(J) + (.1,0)$) (EXTDPT) { $G, 1, 2$  };
            \node[anchor=west] at ($(G) + (.1,0)$) (EXTDPT) { $G, 1, 2$  };
            \node[anchor=west] at ($(K) + (.1,0)$) (EXTDPT) { $\emptyset, 0, 2$  };
        \end{tikzpicture}
        };

        \draw (ABCDE) -- (DEFGH) node[midway, above] {$DE$};
        \draw (DEFGH) -- (GHIJK) node[midway, above] {$GH$};

        \node at ($(DEFGH) + (0,3.2)$) (EXTDPT) { $\outt, \rho^*(\outt), \rho^*(\ancc \cup \out)$  };

        \node at ($(ABCDE) + (0,-4)$) (EXTDPT) { 
            \begin{tikzpicture}[scale = 0.45,font=\footnotesize]
                \node (A) at (-2,0) {$A$};
                \node (B) at (-1,2) {$B$};
                \node (C) at (0,0) {$C$};
                \node (D) at (1,2) {$D$};
                \node (E) at (2,0) {$E$};

                \draw[fill=blue,opacity=0.5, fill opacity=0.2] \convexpath{D,E}{.5cm};

                \draw (A) -- (B);
                \draw (A) -- (C);
                \draw (A) -- (D);
                \draw (B) -- (C);
                \draw (B) -- (D);
                \draw (B) -- (E);
                \draw (C) -- (D);
                \draw (C) -- (E);
                \draw (D) -- (E);
            \end{tikzpicture}
            };

        \node at ($(DEFGH) + (0,-4)$) (EXTDPT) { 
            \begin{tikzpicture}[scale = 0.45,font=\footnotesize]
                \node (D) at (1,2) {$D$};
                \node (E) at (2,0) {$E$};
                \node (F) at (3,2) {$F$};
                \node (G) at (4,0) {$G$};
                \node (H) at (5,2) {$H$};

                \draw[fill=red,opacity=0.5, fill opacity=0.2] \convexpath{D,E}{.4cm};
                \draw[fill=red,opacity=0.5, fill opacity=0.2] \convexpath{G,H}{.4cm};
                \draw[fill=blue,opacity=0.5, fill opacity=0.2] \convexpath{E,F,H}{.5cm};

                \draw[dashed] (D) -- (E);
                \draw (D) -- (F);
                \draw (D) -- (G);
                \draw (E) -- (F);
                \draw (E) -- (G);
                \draw (E) -- (H);
                \draw (F) -- (G);
                \draw (F) -- (H);
            \end{tikzpicture}
            };

        \node at ($(GHIJK) + (0,-4)$) (EXTDPT) { 
            \begin{tikzpicture}[scale = 0.45,font=\footnotesize]
                \node (G) at (4,0) {$G$};
                \node (H) at (5,2) {$H$};
                \node (I) at (6,0) {$I$};
                \node (J) at (7,2) {$J$};
                \node (K) at (8,0) {$K$};
                
                \draw[fill=blue,opacity=0.5, fill opacity=0.2] \convexpath{G,H}{.5cm};

                \draw (I) -- (G);
                \draw (I) -- (H);
                \draw (I) -- (J);
                \draw (I) -- (K);
                \draw (J) -- (G);
                \draw (J) -- (H);
                \draw (J) -- (K);
                \draw (H) -- (K);
            \end{tikzpicture}
            };

    \end{tikzpicture}
\end{center} 
\caption{A TD with PTs in the bags for $Q(E,F,H)$}
    \label{fig:app:tdptexample}
\end{figure}
\end{figure}

Furthermore, we also depicted in Fig.~\ref{fig:app:tdptexample} the subqueries the individual PTs solve. 
To that end, we colored the output variables blue and the additional \textit{input relations} that come from the children red.
Note that the input relations restrict the structure of the PT for each bag while they do not change any fractional edge cover numbers $\rho^*$ which are defined relative to $Q$.
To that end, consider the middle PT.
Even though there is no (non-input) relation between $H$ and $G$, we could not move the variable $H$ to be a child of $E$ and a sibling of $G$ as we could then never check the input relation steaming from the right child PT.
Nevertheless, $\rho^*(HG)=2$ even in the middle PT.
Moreover, we have to make sure that the semi-ring values of each relation is only used once.
Thus, either in the left or in the middle PT we have to use is the induced version $\supp(R[D,E])$ (compare this with the relations $R_i^v$ defined in Section \ref{sec:td}) of the $DE$ relation $R(D,E)$ during the evaluation process (the edge $DE$ is dashed in the middle PT to indicate this).

Re-adding the edge $GH$, would only change the following annotations:
$I$ and $H$ would get the annotations $GH,1,1.5$ in the right PT.
Thus, the right PT only used linear space in that case and this query plan would have an overall space-time exponent of $(1.5,2)$, which is optimal.

Moreover, if we further remove the output variables and consider the scalar version (i.e., the one depicted in Fig.~\ref{fig:ptcrQuery}), the middle PT looses its output variables and, thus, the space-time exponent is $(1,2)$ in total.

\subsection{Defining the Space-Time Exponent for Arbitrary Classes of Plans}

To rigorously define the space-time exponents of $\td^{\calC}$ for an arbitrary class of query plans $\calC$, we augment queries with log-scaled cardinality constraints.
That is, to a query $Q(\boldsymbol{X})\leftarrow \bigotimes_i R_i(\boldsymbol{X}_i)$ we add normalized cardinality constraints $cc_i\in \mathbb{R}_{>0}$, i.e., with $\max\{cc_i \mid i\} = 1$.
Then, for a query plan $\Pi\in \calC(Q,(cc_i)_i)$, we require that for databases $D$ with
$$\log_{|D|}|R^D_i|\leq cc_i$$
the plan $\Pi$ computes $\sem{Q}$ in time $O(|D|^{t(\Pi)})$ and uses $O(|D|^{s(\Pi)})$ space.

First, we extend fraction edge cover numbers to account for cardinality constraint.
To that end, let $\bm Y \subseteq \var(Q)$ be a set of variables. A \emph{fractional edge cover} of $\bm Y$ (with respect to $(Q,(cc_i)_i)$) is a set of non-negative weights $w_i$, one for each relation $R_i$, such that, for every variable $A \in \bm Y$, $\sum_{i: A \in \bm X_i}w_i \geq 1$.  The \emph{fractional edge cover number of $\bm Y$}, $\rho^*(\bm Y, Q, (cc_i)_i)$, is the minimum value of $\sum_i w_i\cdot cc_i$, when the weights $w_i$ range over fractional edge covers of $\bm Y$.
Note that this is also known as the AGM bound \cite{DBLP:conf/focs/AtseriasGM08} and all arguments also work when extending fraction edge cover numbers to account for cardinality constraint.

Recall the definition of the sub-queries $Q^v$ defined in Section \ref{sec:td}.
Then, we redefine Definition \ref{def:time:space:td} as:

\begin{definition}
  Fix a query $Q$ and cardinality constraints $(cc)_i$, and let $\calC$ be a class of query plans for evaluating SPQs. 
  The class of plans $\tdc$ consists all pairs $(TD,\pi)$, where $TD$ is a tree decomposition of $Q$, and $\pi$ is a function that maps each vertex $v$ in $TD$ to a plan 
  $$\pi^v \in \calC(Q^v, ((cc_i / m)_i, (\rho^*(\boldsymbol{Z}^w, Q, (cc_i)_i)/m)_w)$$
  where 
  $$m=\max(\{cc_i \mid i\}\cup\{\rho^*(\boldsymbol{Z}^w, Q, (cc_i)_i) \mid w\}).$$
  I.e., $\rho^*(\boldsymbol{Z}^w, Q, (cc_i)_i)$ is the maximal possible size of the relation $Q^w(\boldsymbol{Z}^w)$ and $\pi^v$ respects these in normalized cardinality constraints.
  The space  and time exponents are then 
  \begin{align}
      s(TD, \pi) = \max_v s(\pi^v)\cdot m,\quad t(TD, \pi) = \max_v t(\pi^v)\cdot m.
  \end{align}
\end{definition}

Next, we let us proof Theorem \ref{thm:tdgj-pt}.

\subsection{Proof for Theorem \ref{thm:tdgj-pt}}

\thmtdgjpt*
\begin{proof}%
To prove $\pt \not\preceq \tdgj$,  take the 4-path query
$Q() = R_1(A,B) \otimes R_2(B,C) \otimes R_3(C,D)$.
A  $\tdgj$ query plan $P$ can evaluate this query 
with $e(P) = (1,1)$. Just take a TD consisting of  3 nodes whose bags are 
$\{A,B\}$, $\{B,C\}$, and $\{C,D\}$.
On the other hand, a $\pt$ query plan
can achieve $e(P') = (0,2)$ via a PT that takes one of $B$ or $C$ as root. 
But linear time complexity 
is out of reach
for $\pt$, since we cannot 
avoid a branch consisting of 2 edges.  
Also additional space does not help.

To prove $\tdgj \not\preceq \pt$, we consider the 3-path query
$Q() = R_1(A,B) \otimes R_2(B,C)$.
It can be evaluated by a plan $P$ in $\pt$ with $e(P) = (0,1)$. Just take the \ptree 
with root $B$ and two child nodes $A,B$. On the other hand, a TD either has two nodes with 
bags $\{A,B\}$ and $\{B,C\}$ or with a single bag that contains all 3 variables. This leads to 
query plans $P'$ with either $e(P') = (1,1)$ or $e(P') = (0,2)$. That is, in $\tdgj$, 
we can reach space exponent = 0 and time exponent = 1 separately, but not at the same time.  

For arbitrary classes of query plans $\calC$ and queries $Q$, we can simply us the TD $(T,\chi)$ consisting of a single node $\{v\}=V(T)$ and bag $\chi(v)=\var(Q)$.
Then, $Q^v=Q$ and $e(\pi^v)=e(T,\chi,\pi)$ for all $\pi^v \in\calC$.
\end{proof}

\fi 
\ifArxiv

\section{Appendix for Section \ref{sec:ptc}}

We start this section by briefly showcasing the computation of the space-time exponent of a PTC.

\subsection{Additional Example of a Plan in $\ptc$}

Consider the query $Q(D,E)$ and PTC seen in Figures~\ref{fig:app:ptcQuery} and \ref{fig:app:ptcExample}.
The nodes/variables with a cache are marked yellow.
In Figure \ref{fig:app:ptcExample}, we added the information necessary to compute the space-time exponent.
That is, for each variable $V=A,\dots,G$ we added $\con(V),\outt(V), \rho^*(\con(V)\cup \outt(V))$ as well as $\rho^*(\con(B_V)\cup [V, B_V]\cup \out(V))$ where $B_V$ is the closed ancestor with a cache (i.e., wither $D$ or $F$).
Thus, the space-time exponent of this structure is simply the max over the last two quantities (component wise), i.e., $(1,2) = (\rho^*(DE), \rho^*(DEFG))$ achieved for example at the root $D$ and at $G$, respectively.

\begin{figure}
\begin{minipage}[b]{0.5\textwidth}

\begin{figure}[H]
    \begin{center}
        \begin{tikzpicture}[font=\footnotesize]
            \node (A) at (-1,3) {A};
            \node (B) at (0,2) {B};
            \node (C) at (0,0) {C};
            \node (D) at (2,2) {D};
            \node (E) at (2,0) {E};
            \node (F) at (4,2) {F};
            \node (G) at (4,0) {G};

        \draw[fill=blue,opacity=0.5, fill opacity=0.2] \convexpath{E,D}{.3cm};

            \draw (A) -- (B);
            \draw (A) -- (C);
            \draw (A) -- (D);
            \draw (B) -- (C);
            \draw (B) -- (D);
            \draw (B) -- (E);
            \draw (C) -- (D);
            \draw (C) -- (E);
            \draw (E) -- (D);
            \draw (D) -- (G);
            \draw (D) -- (F);
            
            \draw (E) -- (F);
            \draw (E) -- (G);
            \draw (F) -- (G);
        \end{tikzpicture}
    \end{center} 
    \caption{Query $Q(D,E)$}
    \label{fig:app:ptcQuery}
\end{figure}

\end{minipage}
\hfill
\begin{minipage}[b]{0.49\textwidth}

\begin{figure}[H]
    \begin{center}
        \begin{tikzpicture}[scale = 0.65, font=\footnotesize]
            \node[circle, draw=black!50, fill=yellow!20] (D) {$D$}
                child {node (B) {$B$}
                    child {node (C) {$C$}
                        child {node (A) {$A$}}
                        child {node (E) {$E$}
                            child {node[circle, draw=black!50, fill=yellow!20] (F) {$F$}
                                child {node (G) {$G$}}
                                }
                            }
                        }
                    };

        \node[anchor=west] at ($(D) + (.25,0)$) (EXTDPT) { $\emptyset, DE, 1, 1$  };
        \node[anchor=west] at ($(B) + (.25,0)$) (EXTDPT) { $D, E, 1, 1$  };
        \node[anchor=west] at ($(C) + (.25,0)$) (EXTDPT) { $DB, E, 1, 1.5$  };
        \node[anchor=east] at ($(A) + (-.25,0)$) (EXTDPT) { $DBC, \emptyset, 0, 2$  };
        \node[anchor=west] at ($(E) + (.25,0)$) (EXTDPT) { $DBC, E, 1, 2$  };
        \node[anchor=west] at ($(F) + (.25,0)$) (EXTDPT) { $DE, \emptyset, 1, 1.5$  };
        \node[anchor=west] at ($(G) + (.25,0)$) (EXTDPT) { $DEF, \emptyset,  0, 2$  };

        \node at ($(D) + (.15,1)$) (EXTDPT) { $\con, \outt, \rho^*(\con \cup \outt), \rho^*(\con(B_{\cdot})\cup [\cdot, B_{\cdot}] \cup \out)$  };
        \end{tikzpicture}
    \end{center} 
    \caption{A PTC of $Q(D,E)$}
    \label{fig:app:ptcExample}
\end{figure}

\end{minipage}
\end{figure}

Now, let us proceed towards a proof of Theorem \ref{thm:algo:ptc}.
To do so, the following lemmata will be helpful.

\subsection{Lemmas for Theorem \ref{thm:algo:ptc}}

\begin{lemma}
\label{lem:app:ptc3}
    Let $(P,\boldsymbol{C})$ be a PTC of the SPQ $Q$ and $A\in \var(Q)$ arbitrary. 
    Then, for an $\boldsymbol{x}\in \sem{Q[{\anc(A)}]}$, the function $\textsc{solve}(A,\boldsymbol{x})$ of Algorithm \refalgptc returns $\sem{Q[{\descc(A)}|{\boldsymbol{x}}](\outt(A))}$.
\end{lemma}

\begin{proof}
    Formally, we can prove this by induction on the size of $\descc(A)$.
    However, we have already seen this in the proof of Lemma \ref{lem:app:pt1} for variables without a cache.
    Thus, we only provide the arguments for the induction step for $A\in \boldsymbol{C}$ here.
    To that end, let $\boldsymbol{y}\in \sem{Q[{\anc(A)}]}$ with $\boldsymbol{y}[\con(A)]=\boldsymbol{x}[\con(A)]$ be the first tuple of this form such that a function call $\textsc{solve}(A,\boldsymbol{y})$ is made.
    Then, $\sem{Q[{\descc(A)}|{\boldsymbol{y}}]}$ is stored in the cache $M_A$ and returned for the function call $\textsc{solve}(A,\boldsymbol{x})$.
    This is correct as due to Lemma \ref{lem:app:con}
    $$\sem{Q[{\descc(A)}|{\boldsymbol{y}}]} = \sem{Q[{\descc(A)}|{\boldsymbol{y}[\con(A)]}]} = \sem{Q[{\descc(A)}|{\boldsymbol{x}[\con(A)]}]} = \sem{Q[{\descc(A)}|{\boldsymbol{x}}]}.$$
    This completes the proof.
\end{proof}

\begin{lemma}
\label{lem:app:ptc4}
    Let $(P,\boldsymbol{C})$ be a PTC of the SPQ $Q$ and $A\in \boldsymbol{C}$ be an arbitrary variable with cache.
    Then, for a run of Algorithm \refalgptc, the function call $\textsc{solve}(A,\boldsymbol{x})$ has a cache miss (reaches line \ref{line:ptc:miss}) at most once per argument pair $(A,\boldsymbol{x}[\con(A)])$.%
\end{lemma}

\begin{proof}
    Let $\boldsymbol{y}\in \dom^{\anc(A)}$ with $\boldsymbol{y}[\con(A)]=\boldsymbol{x}[\con(A)]$ be the first tuple of this form such that a function call $\textsc{solve}(A,\boldsymbol{y})$ is made.
    Then, this call produces a cache miss but $\boldsymbol{y}[\con(A)]\mapsto \sem{Q[{\descc(A)}|{\boldsymbol{y}}]}$ is then added to the cache (using Lemma~\ref{lem:app:ptc3} and Lemma~\ref{lem:app:pt2}) for subsequent calls that agree on $\con(A)$, i.e., also for the call $\textsc{solve}(A,\boldsymbol{x})$.
\end{proof}

\begin{lemma}
\label{lem:app:ptc5}
    Let $(P,\boldsymbol{C})$ be a PTC of the SPQ $Q$ and $A\in \var(Q)$ arbitrary. 
    Further, let $B_A=\min(\boldsymbol{C}\cap \ancc(A))$ be the closest ancestor of $A$ with a cache.
    Then, for a run of Algorithm~\refalgptc, the function call $\textsc{solve}(A,\boldsymbol{x})$ has a cache miss (reaches line \ref{line:ptc:miss}) at most once per argument pair $(A,\boldsymbol{x}[\con(A)\cup (A, B_A]])$. %
\end{lemma}

\begin{proof}
    We proceed to prove this by induction on the size of $(A,B_A]$.
    The base case $B_A=A$ is taken care of by Lemma \ref{lem:app:ptc4}. %
    To that end, note that every call $\textsc{solve}(A,\boldsymbol{x})$ of a run of Algorithm \refalgptc has to also be a call $\textsc{solve}(A,\boldsymbol{x})$ of a run of Algorithm \ref{alg:pt}.

    Then, we can argue inductively.
    To that end, let $A$ be such that the induction hypothesis holds.
    I.e., all calls that incur a cache miss get a tuple $\boldsymbol{x}$ together with a variable $A$ as input such that $\boldsymbol{x}\in \sem{Q[{\anc(A)}]}$ and such that $(A,\boldsymbol{x}[\con(B_A)\cup (A,B_A]])$ is unique (recall also Lemma~\ref{lem:app:pt2}).
    Then, the call simply iterates through domain elements $a\in \sem{Q[A|{\boldsymbol{x}}]}$ and extends $\boldsymbol{x}$ by that element to get $\boldsymbol{x}'=(\boldsymbol{x},a)$.
    Thus, by construction $\boldsymbol{x}'\in \sem{Q[{\ancc(A)}]}$.
    Then, we call $\textsc{solve}(B,\boldsymbol{x}')$ for every $B\in child(A)$.
    Note that $\anc(B)=\ancc(A)$ and, thus, $\boldsymbol{x}'\in \sem{Q[{\anc(B)}]}$.
    Importantly, domain elements are never removed from the input tuple and, thus, we can trace back the call of $\textsc{solve}(B,\boldsymbol{x}')$ to $\textsc{solve}(parent(B),\boldsymbol{x}'[\anc(parent(B))])$.
    Hence, as $(A,\boldsymbol{x}[\con(B_A)\cup (A, B_A]])$ was unique, also $(B,\boldsymbol{x}'[\con(B_A)\cup [A, B_A]])$ is unique.    
\end{proof}

The next lemma is a slight strengthening of Lemma \ref{lem:app:gj}.

\begin{lemma}
\label{lem:app:gj2}
Let $Q(\boldsymbol{X})\leftarrow \bigotimes_i R_i(\boldsymbol{X}_i)$ be a SPQ, $\boldsymbol{Y}\subseteq \boldsymbol{Z}\subseteq \var(Q)$ be sets of variables and $A\in \var(Q)\setminus \boldsymbol{Z}$ a variable.
Further, for every $\boldsymbol{y}\in \sem{Q[{\boldsymbol{Y}}]}$, let $\boldsymbol{z}(\boldsymbol{y})\in \dom^{\boldsymbol{Z}}$ be a arbitrary extended tuple, i.e., such that $\boldsymbol{y}=\boldsymbol{z}(\boldsymbol{y})[\boldsymbol{Y}]$.
Then, we can iterate through $a\mapsto s \in \sem{Q[A|{\boldsymbol{z}(\boldsymbol{y})}]}$ over all $\boldsymbol{y}\in \sem{Q[{\boldsymbol{Y}}]}$ in time $O(|D|^{\rho^*(\boldsymbol{Y}\cup \{A\})})$.
\end{lemma}

\begin{proof}

An almost verbatim copy of the proof of Lemma \ref{lem:app:gj} also works for this lemma.
However, up to Equation \ref{eq:app:gjupto}, $\boldsymbol{Z}$ plays the roll of $\boldsymbol{Y}$ while $\boldsymbol{z}(\boldsymbol{y})$ plays the role of $\boldsymbol{y}$.

Then, for the long equation, we start with
\begin{align*}
    &\sum_{y\in \sem{Q[\boldsymbol{Y}]}} \min_{i\colon A\subseteq \boldsymbol{X}_i \setminus \boldsymbol{Y}} |\supp(R_i[A|\boldsymbol{z}(\boldsymbol{y})])| \\
     & \leq \sum_{y\in \sem{Q[\boldsymbol{Y}]}} \min_{i\colon A\subseteq \boldsymbol{X}_i \setminus \boldsymbol{Y}} |\supp(R_i[A|\boldsymbol{y}])|
\end{align*}
I.e., we simply fix less variables.
From that point onwards, we simply continue the proof verbatim. 
\end{proof}

\begin{lemma}
\label{lem:app:ptc6}
    Let $(P,\boldsymbol{C})$ be a PTC of the SPQ $Q$ and $A\in \var(Q)$ arbitrary. 
    Further, let $B_A=\min(\boldsymbol{C}\cap \ancc(A))$ be the closest ancestor of $A$ with a cache.  
    Then, for a run of the Algorithm \refalgptc, the collective time spent in function calls $\textsc{solve}(A,\boldsymbol{x})$ (excluding cache hits - counting starts from \ref{line:ptc:miss} - and excluding recursive calls - line \ref{line:pt:recurse}) over all values for $\boldsymbol{x}$ is bounded by $O(|D|^{\rho^*(\con(B_A)\cup [A, B_A]\cup \out(A))})$.
\end{lemma}

\begin{proof}
Due to Lemma \ref{lem:app:ptc5}, we only have to sum up the time spent in calls $\textsc{solve}(A,\boldsymbol{x})$ over all $\boldsymbol{x}[\con(B_A)\cup (A,B_A]]\in \sem{Q[{\con(B_A)\cup (A,B_A]}]}$.
For each call, we simply have to go through all $a\mapsto s\in \sem{Q[A|{\boldsymbol{x}}]}$.
We can use Lemma \ref{lem:app:gj2}, to calculate how much time this takes (excluding the time spent in the loop body).
That is, we use $\boldsymbol{Y}:= \con(B_A)\cup (A,B_A]$ and $\boldsymbol{Z}=\anc(A)$ with $\boldsymbol{y}\in \sem{Q[{\con(B_A)\cup (A,B_A]}]}$ mapping to the unique $\boldsymbol{z}(\boldsymbol{y}):=\boldsymbol{x}\in \dom^{\anc(\boldsymbol{A})}$ that caused the cache miss.
I.e., in total, this takes at most $O(|D|^{\rho^*(\con(B_A)\cup [A, B_A])})$ time.

Moreover, in the loop body $(\boldsymbol{x}[\con(B_A)\cup (A,B_A]],a)\in \sem{Q[{\con(B_A)\cup [A,B_A]}]}$.
Then, as we have seen in the proof of Lemma $\ref{lem:app:pt3}$, we compute $\sem{Q[{\descc(A)}|{\boldsymbol{x}}](\outt(A))}$ in \textit{output optimal time} where the output is $\{A\}\cup\out(A)$.
Thus, over all $\boldsymbol{x}\in \sem{Q[{\anc(A)}]}$ (more importantly actually, over all $\boldsymbol{x}[\con(B_A)\cup (A,B_A]]\in \sem{Q[{\con(B_A)\cup (A,B_A]}]}$), the time required is $O(|D|^{\rho^*(\con(B_A)\cup [A,B_A]\cup \out(A))})$.
\end{proof}

Now, let us proof Theorem \ref{thm:algo:ptc}

\subsection{Proof of Theorem \ref{thm:algo:ptc}}

\thmalgoptc*

\begin{proof}

The correctness follows directly by applying Lemma \ref{lem:app:ptc3} to the call $\textsc{solve}(\ptroot(P),())$ which consequently returns $\sem{Q}$ as $\descc(\ptroot(P))=\var(Q)$ and $\outt(\ptroot(P))=\boldsymbol{X}$.

To bound the runtime of Algorithm \refalgptc, we can attribute each step past line \ref{line:ptc:miss} (cache miss) to the current call of the function $\textsc{solve}(A,\boldsymbol{x})$ and the steps before line \ref{line:ptc:miss} to the calling function.
As the number of steps before line \ref{line:ptc:miss} is constant, we can actually simply ignore them.
To that end, for each variables $A$ we sum up the time over all values for $\boldsymbol{x}$.
Due to Lemma \ref{lem:app:ptc6}, this can be bound by $O(|D|^{\rho^*(\con(B_A)\cup [A,B_A]\cup \out(A))})$ where $B_A$ is the closest ancestor of $A$ with a cache.
Thus, by taking the maximum over all $A$, we arrive at a overall bound of $O(|D|^{t(P,\boldsymbol{C})})$ for Algorithm \refalgptc.

Lastly, the bound on the space consumption of Algorithm \refalgptc is rather immediate.
For every variable $A$, the $\K$-relations $\text{TMP}, \text{OUT}$ stored in calls to $\textsc{solve}(A,\boldsymbol{x})$ never exceed $O(|D|^{\rho^*(\outt(A))})$.
For variables $A$ with caches, the caches $M_A$ are essentially a collection of subsets of $\sem{Q[{\outt(A)}|{\boldsymbol{y}}]}$ for $\boldsymbol{y}\in \sem{Q[{\con(A)}]}$.
Thus, they are collectively essentially a subset of $\sem{Q[{\con(A)\cup\outt(A)}]}$ and its size can, therefore, never exceed $O(|D|^{\rho^*(\con(A)\cup\outt(A))})$.    
\end{proof}

Next, we proceed to prove Theorem \ref{thm:tdgj-ptc}.

\subsection{Proof of Theorem \ref{thm:tdgj-ptc}}

\thmtdgjptc*

\begin{proof}
Given a plan $\Pi := ((T,\chi),\pi) \in \td^{\gj}$ of the query $Q(\boldsymbol{X})$, we construct a plan $\Pi_0 := (P,\bm C) \in \ptc$ such that $\Pi_0 \preceq \Pi$.  The construction is based on the variable elimination procedure for a tree decomposition~\cite{DBLP:journals/corr/KhamisNRR15}, and proceeds by induction on the number of bags in $T$.  If $T$ has a single bag, then $\Pi$ is essentially a $\gj$ plan, and the claim follows from $\pt \preceq \gj$.  Otherwise, let $v$ be a leaf of $T$, and $p:= \parent(v)$.  We eliminate all variables $\set{A_1, \ldots, A_k} := \chi(v)\setminus \chi(p)$.  Let $\bm Z$ be their neighbors, $\bm Z := \setof{B\not\in \set{A_1, \ldots, A_k}}{\exists \text{ atom } R_i(\bm X_i), \exists j, \text{ s.t. } A_j, B \in \bm X_i}$.  
Then, $\boldsymbol{Z}$ must be in $\chi(v)\cap \chi(p)$.
Let $Q'$ be the query obtained from $Q$ by removing all variables $A_1, \ldots, A_k$, and adding a new atom $R(\bm Z)$.
That is, let
$$Q'\leftarrow Q[\var(Q)\setminus \set{A_1, \ldots, A_k}](\boldsymbol{X})\otimes R(\bm Z).$$
Note that $\set{A_1, \ldots, A_k}\cap \boldsymbol{X}=\emptyset $ as $\boldsymbol{X}\subseteq \chi(\ptroot(T))$.
Let $\Pi' = ((T',\chi),\pi)$ be the plan obtained from $\Pi$ by removing the leaf $v$. 
Notice that $\Pi'$ is also a plan of $Q'$.
By induction hypothesis, $\Pi'$ can be converted to a $\ptc$ plan $\Pi_0' = (P',\bm C')$ such that $\Pi_0' \preceq \Pi'$.  All variables $\bm Z$ belong to a branch of $\Pi_0'$ (because of the atom $R(\bm Z)$).  Construct the pseudo-tree $P$ from $P'$ by adding a branch $A_1\text{-}A_2\text{-}\cdots\text{-}A_k$ as a child of the last variable in $\bm Z$.  Finally, define $\Pi_0 := (P,\bm C \cup \set{A_1})$ (only $A_1$ receives a cache).  It can be checked that $\Pi_0 \preceq \Pi$, which proves $\ptc \preceq \td^{\gj}$.  

Strictness, $\ptc \prec \td^{\gj}$, follows from Theorem~\ref{thm:tdgj-pt} and the trivial relationship $\ptc\preceq \pt$.
That is, notice that
$$\tdgj \succeq \ptc, \quad \ptc\preceq \pt,$$
while also 
$$\tdgj \not\succeq \pt,\quad \tdgj \not\preceq \pt.$$
Thus, neither $\pt\equiv \ptc$ nor $\tdgj\equiv \ptc$ can hold.
\end{proof}

Lastly, we prove Theorem \ref{thm:tdpt-tdptc}

\subsection{Proof of Theorem \ref{thm:tdpt-tdptc}}

\thmtdpttdptc*

\begin{proof}
    Let $\Pi=((T,\chi), \pi)$ be a query plan in $\tdptc$ of the query $Q$.
    We prove that there is a query plan $\Pi'=((T',\chi'),\pi')$ in $\tdptc$ of the query $Q$ with less caches than $\Pi$ (formally excluding the caches at the roots of $\pi^v, \pi'^{v'}$) such that $\Pi'\preceq \Pi$.
    Thus, by induction, there must also be a $\Pi'$ in $\tdpt$.

    To that end, we explain how to transform $\Pi=((T,\chi), \pi)$ into $\Pi'=((T',\chi'),\pi')$, 
    where the second contains 1 less cache.
    To do so, let $v\in V(T)$ be a node where $\pi^v=(P^v,\boldsymbol{C}^v)$ contains a cache at a non-root variable $A\in V(P^v)=\chi(v)$.
    We then simply split the bag $\chi(v)$ into $\descc(A)\cup \con(A)$ and $(\chi(v)\setminus\descc(A))\cup \outt(A)$ (note that the intersection is $\con(A)\cup \outt(A)$).
    To that end, we introduce new nodes $c,p$ to replace $v$ with exactly those bags and where $c$ is the child of $p$.
    Then, for every child $c_{v}$ of $v$, it can either be made a child of $c$ or of $p$, depending of whether $\boldsymbol{Z}^{c_v}:=\chi(c_v)\cap \chi(v)$ is a subset of $\chi(c)$ or $\chi(p)$.
    To see that this is well-defined, note that $Q^v$ contains the relation $Q^{c_v}(\boldsymbol{Z}^{c_v})$.
    Thus, $B=\min(\boldsymbol{Z}^{c_v})$ is well-defined.
    If, $B\in \descc(A)$, then, 
    $$\boldsymbol{Z}^{c_v}\subseteq [B,A]\cup \con(A)\subseteq \chi(c).$$
    Otherwise, when $B\not \in \descc(A)$, then 
    $\boldsymbol{Z}^{c_v}\cap \descc(A)=\emptyset$ and 
    $$\boldsymbol{Z}^{c_v}\in (\chi(v)\setminus \descc(A))\subseteq \chi(p).$$
    The parent $p_v$ of $v$ becomes the parent of $p$.
    Now for the PTC of $c$ and $p$.
    The PTC $\pi^p:=(P^p,\boldsymbol{C}^p)$ will be the same as $\pi^v$ but where the subtree rooted in $A$ is simply replaced by a path containing $\outt(A)$ and no caches. 
    The PTC $\pi^v=(P^c,\boldsymbol{C}^c)$ will be the subtree of $P^v$ rooted in $A$ where we add $\con(A)$ as a path above $A$.
    Furthermore, we keep all caches except the ones at variables of $\conn(A)$.
    Thus, we decreasing the overall amount of caches by 1 as we in particular removed the cache at $A$.
    Furthermore, the new structure still describes a query plan of $\tdptc$ for $Q$ and the space-time exponent did not change.
    
    Note that when all $\boldsymbol{C}_{v}$ only contain the root, then we can simply drop the $\boldsymbol{C}_v$ altogether and view $\Pi$ as a plan in $\tdptc$. 

    As $\tdptc\preceq \tdpt$ trivially holds, we have shown in particular $\tdpt\equiv \tdptc$.
    Further, $\tdptc\preceq \ptc$ due to Theorem \ref{thm:tdgj-ptc} and as such also $\tdpt \preceq \ptc$.
\end{proof}

\fi 

\ifArxiv

\section{Appendix for Section \ref{sec:ptcr}}
\label{sec:app:ptcr}

We start this section by briefly showcasing the computation of the space-time exponent of a PTCR as well as the interplay of $\ria, \raexc, \ra, \icon,$ and $\scon$.

\subsection{Additional Example of a Plan in $\ptcr$}

Consider the scalar query $Q()$ and PTCR seen in Figures~\ref{fig:app:ptcrQuery} and \ref{fig:app:ptcrExample} (the graph and PT are an adaptation of an example given in \cite{DBLP:series/synthesis/2019Dechter}).
In Figure \ref{fig:app:ptcrExample}, for each variable $V=A,\dots,K$ we added $\raexc(V)$ as annotations.
Every set $\raexc(V)$ can be partitioned into the stored context $\scon(V)$ and the relevant instantiated ancestor $\ria(V)$, and $\ria(V)$ can be partitioned into those explicitly given in $\icon(V)$ and those that are only implicitly give, i.e., $\ria(V)\setminus \icon(V)$.
Thus, we added $\raexc(V)$ where we used bar indicated the partition, i.e.,

\[ 
\underbrace{\overunderbraces{&&&\br{3}{\con(V)}}%
{&\ria(V)\setminus \icon(V) & | & \icon(V) &| &\scon(V).&}%
{&\br{3}{\ria(V)} &}}_{\raexc(V)}
\]

Additionally, we added the quantities $\rho^*(\scon(V))$ and $\rho^*(\ra(V))$ to compute the space-time exponent.
Note that we did not include any output variables ($\outt(V)=\emptyset$) as their effect is very similar to that in the PTC case.
Thus, the space-time exponent of this structure is simply the max over the last two quantities (component wise), i.e., $(1,2) = (\rho^*(AI), \rho^*(CHABI))$ achieved for example at $F$ and at $I$, respectively.

\begin{figure}
\begin{minipage}[b]{0.44\textwidth}

\begin{figure}[H]
    \begin{center}
    \begin{tikzpicture}[scale = 0.65, font=\footnotesize]

        \foreach \angle/\label in {90/A, 162/B, 234/H, 306/C, 18/E}
                {
                    \node (\label) at (\angle:1.5) {$\label$};
                }
        \node (I) at (162:4) {$I$};
        \node (F) at (126:4) {$F$};
        \node (G) at (90:4) {$G$};

        \node (J) at (54:4) {$J$};
        \node (D) at (18:4) {$D$};
        \node (K) at (-18:4) {$K$};

        \draw (A) -- (G);
        \draw (A) -- (F);
        \draw (A) -- (B);
        \draw (A) -- (C);
        \draw (A) -- (E);
        \draw (A) -- (H);
        \draw (A) -- (I);
        \draw (A) -- (J);
        
        \draw (B) -- (I);
        \draw (B) -- (H);
        \draw (B) -- (C);
        \draw (B) -- (E);

        \draw (C) -- (H);
        \draw (C) -- (E);
        \draw (C) -- (D);
        \draw (C) -- (K);

        \draw (D) -- (J);
        \draw (D) -- (E);
        \draw (D) -- (K);
        
        \draw (E) -- (J);

        \draw (F) -- (G);
        \draw (F) -- (I);
        
        \draw (H) -- (I);
        \draw (H) -- (J);

    \end{tikzpicture}
    \end{center} 
    \caption{Query $Q()$}
    \label{fig:app:ptcrQuery}
\end{figure}

\end{minipage}
\hfill
\begin{minipage}[b]{0.55\textwidth}

\begin{figure}[H]
    \begin{center}
        \begin{tikzpicture}[scale = 0.65, font=\footnotesize]
            \node (C) {$C$}
                child {node (H) {$H$}
                    child {node (A) {$A$}
                        child {node (B) {$B$}
                            child {node (I) {$I$}
                                child {node (F) {$F$}
                                    child {node (G) {$G$}
                                        }
                                    }
                                }
                            child {node (E) {$E$}
                                child {node (J) {$J$}
                                    child {node (D) {$D$}
                                        child {node (K) {$K$}
                                            }
                                        }
                                    }
                                }
                            }
                        }
                    };

        \node[anchor=west] at ($(C) + (.25,0)$) (EXTDPT) { $||, 0, 1$  };
        \node[anchor=west] at ($(H) + (.25,0)$) (EXTDPT) { $|C|, 0, 1$  };
        \node[anchor=west] at ($(A) + (.25,0)$) (EXTDPT) { $|CH|, 0, 1.5$  };
        \node[anchor=west] at ($(B) + (.25,0)$) (EXTDPT) { $|CHA|, 0, 2$  };

        \node[anchor=west] at ($(E) + (.25,0)$) (EXTDPT) { $|CHA|, 0, 2$  };
        \node[anchor=west] at ($(J) + (.25,0)$) (EXTDPT) { $|CHAE|, 0, 2.5$  };
        \node[anchor=west] at ($(D) + (.25,0)$) (EXTDPT) { $|C|EJ, 1, 2$  };
        \node[anchor=west] at ($(K) + (.25,0)$) (EXTDPT) { $|C|D, 1, 1.5$  };

        \node[anchor=east] at ($(I) + (-.25,0)$) (EXTDPT) { $C|HAB|, 0, 2.5$  };
        \node[anchor=east] at ($(F) + (-.25,0)$) (EXTDPT) { $||AI, 1, 1.5$  };
        \node[anchor=east] at ($(G) + (-.25,0)$) (EXTDPT) { $|A|F, 1, 1.5$  };

        \node at ($(C) + (0,1.5)$) (EXTDPT) { $(\ria\setminus \icon) | \icon| \scon, \rho^*(\scon), \rho^*(\ra)$  };
        \end{tikzpicture}
    \end{center} 
    \caption{A PTCR of $Q(D,E)$}
    \label{fig:app:ptcrExample}
\end{figure}

\end{minipage}
\end{figure}

To compute $\ria, \ra, \con, \icon,$ and $\scon$, it is best to go top-down through the PT.
That is, for each variable $V$, one first computes the context $\con(V)$ (i.e., the ancestor that are still relevant further down) and performs the split according to $C(V)$.
Then, to compute the $\ria(V)$ one simply has to take $U=\min(\icon(V))$ and take the variables $\ra(parent(V))\cap \ancc(U)$.
If $\icon(V)=\emptyset$, nothing is taken.
$\ra(V)$ is then simply $\ria(V) \cup \scon(V)\cup \{V\}$.
Cases that showcase this computation are the variables $I,F,G,D,K$.
For $I$, $\ra(B)$ get added to $\ria(I)$ and thus also $C$.
For $F$, nothing gets added to $\ria(F)$ as $\icon(F)=\emptyset$.
For $G$, nothing gets added to $\ria(G)$ as there is no variable in $\ra(F)$ higher up than $A$.
For $D,K$, nothing gets added to $\ria(D),\ria(K)$ as $C$ is already the root.

Further, we note that always $\icon(A)\subseteq \ria(A)\subseteq \min(\icon(A))$ (with the understanding that $\min(\emptyset) = \emptyset$), $\con(A) \subseteq \raexc(A)\subseteq \ancc(A)$, $\con(A) \subseteq \ra(A)\subseteq \ancc(A)$, and $\raexc(A)\subseteq \ra(parent(A))$ always hold.

Now, let us proceed towards a proof of Theorem \ref{thm:algo:ptcr}.
To do so, the following lemmas will be helpful.

\subsection{Lemmas for Theorem \ref{thm:algo:ptcr}}

\begin{lemma}
\label{lem:app:ptcr1}    
    Let $(P,C)$ be a PTCR of the SPQ $Q$ and $A\in \var(Q)$ arbitrary. 
    Then, for a run of Algorithm \ref{alg:ptcr}, the function $\textsc{solve}(A,\boldsymbol{x})$ is only called for arguments $\boldsymbol{x}\in \dom^{\anc(A)}$ where $\boldsymbol{x}[\raexc(A)]\in \sem{Q[{\raexc(A)}]}$.
\end{lemma}

\begin{proof}    
We prove this by induction on the size of $\anc(A)$.
To that end, the base case holds trivially due to $\raexc(\ptroot(P)) = \emptyset$.

For the induction step, let $A$ be such that the induction hypothesis holds.
I.e., all calls are made for $\boldsymbol{x}\in \dom^{\anc(A)}$ where $\boldsymbol{x}[\raexc(A)]\in \sem{Q[{\raexc(A)}]}$.
Now, let $\boldsymbol{x}$ be such that is incurs a cache miss.
Then, the call takes $\yanc := \boldsymbol{x}[\anc(A)\setminus \scon(A)]$ and iterates through 
$$(\ysto,a)\in \sem{Q[{\scon(A)\cup \{A\}}|{\boldsymbol{x}[\ria(A)]}]}.$$
Consequently, 
$$(\boldsymbol{x}[\ria(A)], \ysto,a)\in \sem{Q[{\ria(A)\cup \scon(A)\cup \{A\}}]}.$$
Put differently, for $\boldsymbol{y}':=(\yanc, \ysto, a)$, we get 
$$\boldsymbol{y}'[\ra(A)] = (\boldsymbol{x}[\ria(A)], \ysto,a)\in \sem{Q[{\ria(A)\cup \scon(A)\cup \{A\}}]} = \sem{Q[{\ra(A)}]}.$$
Then, let $B$ be a child of $A$.
We know, $\raexc(B)\subseteq \ra(A)$ and hence, $\boldsymbol{y}'[\raexc(B)]  \in \sem{Q[{\raexc(B)}]}$.
Thus, the recursive call $\textsc{solve}(B,\boldsymbol{y}')$ is as required.
Note that all calls made to the function $\textsc{solve}$ (except the very first) arise in this form.
Hence, this complete the induction step and with it the proof.
\end{proof}

\begin{lemma}
\label{lem:app:ptcr2}
    Let $(P,C)$ be a PTCR of the SPQ $Q$ and $A\in \var(Q)$ arbitrary.  
    Then, during the execution of Algorithm \ref{alg:ptcr}, the function call $\textsc{solve}(A,\boldsymbol{x})$ returns $\sem{Q[{\descc(A)}|{\boldsymbol{x}[\con(A)]}](\outt(A))}$.
    Furthermore, after a call to $\textsc{fillCache}(A,\boldsymbol{x}[\anc(A)\setminus \scon(A)])$ the cache $M_A$ is of the form
    {\small
    \begin{align}\{(\boldsymbol{x}[\icon(A)],\ysto)\mapsto \sem{Q[{\descc(A)}|{\boldsymbol{x}[\icon(A)],\ysto}](\outt(A))}\mid \ysto\in \sem{Q[{\scon(A)}|{\boldsymbol{x}[\ria(A)]}]}\}.\label{eq:app:ptcrcache}        
    \end{align}
    }
\end{lemma}

\begin{proof}
We prove both claims in a simultaneous induction on the size of $\descc(A)$.
To that end, the base case is $\descc(A) = \{A\}$.
We start with the base claim for $\textsc{fillCache}(A,\yanc)$ with $\yanc:= \boldsymbol{x}[\anc(A)\setminus \scon(A)]$.
We then iterate through $\ysto\in \sem{Q[{\scon(A)}|{\boldsymbol{x}[\ria(A)]}]}$ and compute 
$$\text{OUT} = \bigoplus_{a\mapsto s \in \sem{Q[A|{\boldsymbol{x}[\ria(A)],\ysto]}}}\{a[\outt(A)]\mapsto s\}$$
As there are no children of $A$, this is clearly equal to $\sem{Q[A|{\boldsymbol{x}[\ria(A)],\ysto]}(\outt(A))}.$
Further, due to Lemma \ref{lem:app:con}, this is equal to $\sem{Q[{\descc(A)}|{\boldsymbol{x}[\icon(A)],\ysto}](\outt(A))}$.
Thus, the base claim for $\textsc{fillCache}(A,\yanc)$ holds.

Now, for the base claim for $\textsc{solve}(A,\boldsymbol{x})$.
There are two cases to consider:
\begin{itemize}
    \item $ \boldsymbol{x} [ \ria(A) ] = \ria_A \colon $  Then, there was a previous call with $\textsc{solve}(A,\boldsymbol{z})$ such that the cache was filled and $\boldsymbol{z}[\ria(A)]=\boldsymbol{x}[\ria(A)]$.
    Thus, $M_A$ is of the form given in Eq. \eqref{eq:app:ptcrcache}.
    Furthermore, due to Lemma \ref{lem:app:ptcr1}, $\boldsymbol{x}[\raexc(A)]\in \sem{Q[{\raexc(A)}]}$.
    Moreover, $\boldsymbol{x}[\scon(A)]\in \sem{Q[{\scon(A)}|{\boldsymbol{x}[\ria(A)]}]}$.
    Hence, 
    \begin{align*}
        M_A(\boldsymbol{x}[\con(A)]) & = M_A(\boldsymbol{x}[\icon(A)],\boldsymbol{x}[\scon(A)]) \\
        & = \sem{Q[{\descc(A)}|{\boldsymbol{x}[\icon(A)],\boldsymbol{x}[\scon(A)]}](\outt(A))} \\
        &= \sem{Q[{\descc(A)}|{\boldsymbol{x}[\con(A)]}](\outt(A))}.
    \end{align*}
    \item $ \boldsymbol{x} [ \ria(A) ] \neq \ria_A \colon $ Then, $\textsc{fillCache}(A,\boldsymbol{x}[\anc(A)\setminus \scon(A)])$ will be called and after the call, $M_A$ will be of the form given in Eq. \eqref{eq:app:ptcrcache}.
    Thus, we can proceed as in the previous case.
\end{itemize}

Now for the induction step.
To that end, let $\descc(A) \neq \{A\}$.
We start with the inductive claim for $\textsc{fillCache}(A,\yanc)$ with $\yanc:= \boldsymbol{x}[\anc(A)\setminus \scon(A)]$.
We then iterate through $\ysto\in \sem{Q[{\scon(A)}|{\boldsymbol{x}[\ria(A)]}]}$ and compute for $\boldsymbol{y}:=(\yanc,\ysto)$
$$\text{OUT} = \bigoplus_{a\mapsto s \in \sem{Q[A|{\boldsymbol{y}[\raexc(A)]}]}}\{a[\outt(A)]\mapsto s\}\otimes\bigotimes_{B\in child(A)}\textsc{solve}(B,(\boldsymbol{y},a))$$
which is equal to 
$$\text{OUT} =  \bigoplus_{a\mapsto s \in \sem{Q[A|{\boldsymbol{y}[\raexc(A)]}]}}\{a[\outt(A)]\mapsto s\}\otimes\bigotimes_{B\in child(A)}\sem{Q[{\descc(B)}|{\boldsymbol{y},a[\con(B)]}](\outt(B))} $$
due to the induction hypothesis applied to the children $B$.
Then, by construction, $(\boldsymbol{y}[\raexc(A)],a)\in \sem{Q[{\ra(A)}]}$ and by Lemma \ref{lem:app:con} (applied twice)
\begin{align*}
    \text{OUT} &=  \bigoplus_{a\mapsto s \in \sem{Q[A|{\boldsymbol{y}[\raexc(A)]}]}}\{a[\outt(A)]\mapsto s\}\otimes\bigotimes_{B\in child(A)}\sem{Q[{\descc(B)}|{\boldsymbol{y}[\raexc(A)],a}](\outt(B))} \\
    &= \sem{Q[{\descc(A)}|{\boldsymbol{y}[\raexc(A)]}](\outt(A))} \\
    &= \sem{Q[{\descc(A)}|{\boldsymbol{y}[\con(A)]}](\outt(A))} \\
    &= \sem{Q[{\descc(A)}|{\boldsymbol{x}[\icon(A)],\ysto}](\outt(A))}.
\end{align*}
Thus, after the call $\textsc{fillCache}(A,\yanc)$, the cache $M_A$ is as required in the form of Eq. \eqref{eq:app:ptcrcache}.

Now, for the inductive claim for $\textsc{solve}(A,\boldsymbol{x})$ we can simply reuse the argument for the base claim.
Note that we never used the fact that $A$ has no children there.
\end{proof}

For the following claim, we assume that loops over multiple variables are processed essentially in lexicographic order.
E.g., for $ab\in\sem{Q[AB]}$, first all $b$ values matching to the first $a$ value are processed before going to the next value for $a$.

\begin{lemma}
\label{lem:app:ptcr3}
    Let $(P,C)$ be a PTCR of the SPQ $Q$ and $A\in \var(Q)$ arbitrary.  
    Then, for a run of Algorithm \ref{alg:ptcr}, let $\textsc{solve}(A,\boldsymbol{x}_1), \textsc{solve}(A,\boldsymbol{x}_2)$ be two calls where the first happens before the latter.
    Then, $\boldsymbol{x}_1[\ria(A)]$ comes lexicographically before $\boldsymbol{x}_2[\ria(A)]$ (or they are equal).
\end{lemma}

\begin{proof}    
We prove this by induction on the size of $\anc(A)$.
To that end, the base case holds trivially due to $\ria(\ptroot(P)) = \emptyset$.

For the induction step, let $A$ be such that the induction hypothesis holds.
I.e., all calls are made in lexicographically order.
Then, for $B\in child(A)$, consider two $\textsc{solve}(B,\boldsymbol{y}'_1), \textsc{solve}(B,\boldsymbol{y}'_2)$ calls where the first happens before the latter.
Then, there must have been calls $\textsc{solve}(A,\boldsymbol{x}_1), \textsc{solve}(A,\boldsymbol{x}_2)$ made in that order (or they are equal) such that $$\boldsymbol{x}_i[\anc(A)\setminus \scon(A)] = \boldsymbol{y}'_i[\anc(A)\setminus \scon(A)].$$
Furthermore, $\boldsymbol{x}_1[\ria(A)]$ comes lexicographically before $\boldsymbol{x}_2[\ria(A)]$ due to the induction hypothesis.
We distinguish two cases:
\begin{itemize}
    \item $\boldsymbol{x}_1[\ria(A)] = \boldsymbol{x}_2[\ria(A)]\colon$ Then, all calls $\textsc{solve}(A,\boldsymbol{x}_3)$ made in between the two also agree on $\ria(A)$.
    Thus, the cache was only filled in the first calls, i.e., as both spawned recursive calls, they have to be the same call and $\boldsymbol{x}_1=\boldsymbol{x}_2$.
    Then, note that 
    $$(\ysto,a)\in \sem{Q[\scon(A)\cup \{A\}|{\boldsymbol{x}[\ria(A)]}]}$$
    are processed in lexicographic order and produce the recursive calls $$\textsc{solve}(B,(\boldsymbol{x}_1[\anc(A)\setminus \scon(A)],\ysto,a))$$
    Importantly, $\boldsymbol{y}'_1[\ria(A)] = \boldsymbol{y}'_2[\ria(A)]$ and $\boldsymbol{y}'_1[\scon(A)\cup \{A\}]$ comes lexicographically before $\boldsymbol{y}'_2[\scon(A)\cup \{A\}]$.
    Further, $\ria(B)$ is a prefix of $\ra(A) = \ria(A)\cup \scon(A) \cup \{A\}$.
    Consequently, $\boldsymbol{y}'_1[\ria(B)]$ also comes lexicographically before $\boldsymbol{y}'_2[\ria(B)]$.

    \item $\boldsymbol{x}_1[\ria(A)] \neq \boldsymbol{x}_2[\ria(A)]\colon$ 
    First note that $\ria(A)\subseteq \anc(A)\setminus \scon(A)$.
    Then, we have to consider different sub-cases:
    \begin{itemize}
        \item $\icon(B)=\emptyset\colon$ the claim holds trivially as $\ria(B)=\emptyset$ in this case.
        \item $\min(\icon(B))\in \icon(A)\colon$ In this case, $\ria(B)$ is a prefix of $\ria(A)$.
        Thus, $\boldsymbol{y}'_1[\ria(B)]= \boldsymbol{x}_1[\ria(B)]$ comes lexicographically before  $\boldsymbol{y}'_1[\ria(B)]= \boldsymbol{x}_1[\ria(B)]$.
        \item $\min(\icon(B))\in \scon(A)\cup \{A\}\colon$ Then, $\ria(A)$ is a prefix of $\ria(B)$ and $\boldsymbol{y}'_1[\ria(A)]= \boldsymbol{x}_1[\ria(A)]$ comes lexicographically strictly before  $\boldsymbol{y}'_2[\ria(A)]= \boldsymbol{x}_2[\ria(A)]$.
        Thus, also $\boldsymbol{y}'_1[\ria(B)]$ comes lexicographically strictly before  $\boldsymbol{y}'_2[\ria(B)]$.
    \end{itemize}
\end{itemize}
This completes the induction step and end the proof.
\end{proof}

\begin{lemma}
\label{lem:app:ptcr4}
    Let $(P,C)$ be a PTCR of the SPQ $Q$ and $A\in \var(Q)$ arbitrary.  
    Then for a run of the Algorithm \ref{alg:ptcr}, the collective time spent in function calls $\textsc{solve}(A,\boldsymbol{x})$ (excluding cache hits - counting starts from line \ref{line:ptcr} - and excluding recursive calls - line \ref{line:ptcr:recurse}) over all values for $\boldsymbol{x}$ is bounded by $O(|D|^{\rho^*(\ra(A)\cup \out(A))})$.
\end{lemma}

\begin{proof}
Due to Lemma \ref{lem:app:ptcr3}, we only have to sum up the time spent in calls $\textsc{solve}(A,\boldsymbol{x})$ over all $\boldsymbol{x}[\ria(A)]\in \sem{Q[{\ria(A)}]}$ as only the first $\boldsymbol{x}$ that matches $\boldsymbol{x}[\ria(A)]$ triggers a cache miss.
For each call, we simply have to go through all $\ysto\in \sem{Q[{\scon(A)}|{\boldsymbol{x}[\ria(A)]}]}$.
We can use Lemma \ref{lem:app:gj}, to calculate how much time this takes.
To that end, let $A_1\dots A_l = \scon(A)$ and let us use nested loops.
Then, by applying Lemma \ref{lem:app:gj} $l$-times, we arrive at a bound 
$$O(\sum _{i=1,\dots,l}|D|^{\rho^*(\ria(A)\cup \{A_1,\dots,A_i\})}) = O(|D|^{\rho^*(\ria(A)\cup \scon(A))}) = O(|D|^{\rho^*(\raexc(A))}).$$
Then, iterating further through $a\mapsto s\in  \sem{Q[{A}|{\boldsymbol{x}[\ria(A)],\ysto}]}$ and applying Lemma \ref{lem:app:gj} again, we arrive at
$$O(|D|^{\rho^*(\ria(A)\cup \scon(A)\cup \{A\})}) = O(|D|^{\rho^*(\ra(A))}).$$

Moreover, in the loop body $(\boldsymbol{x}[\ria(A)],\ysto ,a)\in \sem{Q[{\ra(A)}]}$.
Then, as we have seen in the proof of Lemma $\ref{lem:app:ptcr2}$, we compute $\sem{Q[{\descc(A)}|{\boldsymbol{x}[\ria(A)],\ysto}](\outt(A))}$ in \textit{output optimal time} where the output is $\{A\}\cup\out(A)$.
Thus, over all $\boldsymbol{x}[\raexc(A)]\in \sem{Q[{\raexc(A)}]}$ (more importantly actually, over all $\boldsymbol{x}[\ria(A)]\in \sem{Q[{\ria(A)}]}$), the time required is $O(|D|^{\rho^*(\ra(A)\cup \out(A))})$.
\end{proof}

Now, let us proof Theorem \ref{thm:algo:ptcr}

\subsection{Proof of Theorem \ref{thm:algo:ptcr}}

\thmalgoptcr*

\begin{proof}

The correctness follows directly by applying Lemma \ref{lem:app:ptcr2} to the call $\textsc{solve}(root(P),())$ which consequently returns $\sem{Q(\boldsymbol{X})}$ as $\descc(root(P))=\var(Q)$ and $\outt(root(P))=\boldsymbol{X}$.

To bound the runtime of Algorithm \ref{alg:ptcr}, we can attribute each step past line \ref{line:ptcr} (cache miss) to the current call of the function $\textsc{solve}(A,\boldsymbol{x})$ and the steps before line \ref{line:ptc:miss} to the calling function.
As the number of steps before line \ref{line:ptc:miss} is constant, we can actually simply ignore them.
To that end, for each variables $A$ we sum up the time over all values for $\boldsymbol{x}$.
Due to Lemma \ref{lem:app:ptcr4}, this can be bound by $O(|D|^{\rho^*(\ra(A)\cup \out(A))})$.
Thus, by taking the maximum over all $A$, we arrive at a overall bound of $O(|D|^{t(\Pi)})$ for Algorithm \ref{alg:ptcr}.

Lastly, to bound the space consumption of Algorithm \ref{alg:ptcr} is suffice to look at Lemma \ref{lem:app:ptcr2} that describes the caches.
That is, for variables $A$, the cache $M_A$ is a collection of $\K$-relations 
$$\sem{Q[{\descc(A)}|{\boldsymbol{x}[\ria(A)],\ysto}](\outt(A))}$$
for $\ysto\in \sem{Q[{\scon(A)}|{\boldsymbol{x}[\ria(A)]}]}$.
Put differently, they are essentially a partitioning of 
$$\sem{Q[{\descc(A)\cup \scon(A)}|{\boldsymbol{x}[\ria(A)]}](\scon(A)\cup \outt(A))}.$$
The size of the cache $M_A$ can, therefore, never exceed $O(|D|^{\rho^*(\scon(A)\cup\outt(A))})$.   
Hence, the size of all caches $M_A$ together never exceed $O(|D|^{s(\Pi)})$.
\end{proof}

Next, we proceed to Theorem \ref{thm:ptcr-tdpt}.

\subsection{Proof of Theorem \ref{thm:ptcr-tdpt}}

\thmptcrtdpt*

\begin{proof}
To prove $\tdpt \not\preceq \ptcr$,  take the scalar query
$Q()$ depicted in Fig,~\ref{fig:app:ptcrQuery2} together with the PTCR depicted in Fig.~\ref{fig:app:ptcrExample2}.
The space-time exponent of this PTCR is $(1,2)$.
However, a TD of $Q()$ that only uses linear space has to put everything in one bag as there is no edge that separates the query.
To show that any such TD-based plan would require more than quadratic time, we used our computer implementation and exhaustively searched the space of $\pt$, and we found that none solve $Q()$ in quadratic time.

To prove $\tdpt \not\preceq \ptcr$, take the query $Q()$ depicted in Fig.~\ref{fig:tdptQuery}.
It can be evaluated by a plan in $\tdpt$ with space-time exponent $(1,2)$ by using the TD and PTs depicted in Fig.~\ref{fig:app:tdptexample}.
Note that the space computation given in the figure is for when the edge $GH$ is removed and $EFH$ are output variable.
However, we explained in Section \ref{sec:app:tdptexample} that re-adding the edge $GH$ and removing the output would result in the space-time exponent $(1,2)$.
However, we again leveraged our computer implementation for an exhaustive search and found that there is no plan in $\ptcr$ that is at least as good. 
\end{proof}

\begin{figure}
\begin{minipage}[b]{0.44\textwidth}

\begin{figure}[H]
    \begin{center}
    \begin{tikzpicture}[scale = 0.45,font=\footnotesize]
        \node (A) at (4,2) {$A$};
        \node (B) at (6,2) {$B$};
        \node (F) at (2,0) {$F$};
        \node (E) at (4,0) {$E$};
        \node (D) at (6,0) {$D$};
        \node (C) at (8,0) {$C$};

        \draw (A) -- (B);
        \draw (A) -- (C);
        \draw (A) -- (D);
        \draw (A) -- (E);
        \draw (A) -- (F);
        \draw (B) -- (C);
        \draw (B) -- (D);
        \draw (B) -- (E);
        \draw (B) -- (F);
        \draw (F) -- (E);
        \draw (E) -- (D);
        \draw (D) -- (C);
    \end{tikzpicture}
    \end{center} 
    \caption{Query $Q()$}
    \label{fig:app:ptcrQuery2}
\end{figure}

\end{minipage}
\hfill
\begin{minipage}[b]{0.55\textwidth}

\begin{figure}[H]
    \begin{center}
        \begin{tikzpicture}[scale = 0.65, font=\footnotesize]
            \node (A) {$A$}
                child {node (B) {$B$}
                    child {node (C) {$C$}
                        child {node (D) {$D$}
                            child {node (E) {$E$}
                                child {node (F) {$F$}
                                }
                            }
                        }
                    }
                };

        \node[anchor=west] at ($(A) + (.25,0)$) (EXTDPT) { $||, 0, 1$  };
        \node[anchor=west] at ($(B) + (.25,0)$) (EXTDPT) { $|A|, 0, 1$  };
        \node[anchor=west] at ($(C) + (.25,0)$) (EXTDPT) { $|AB|, 0, 1.5$  };
        \node[anchor=west] at ($(D) + (.25,0)$) (EXTDPT) { $|ABC|, 0, 2$  };

        \node[anchor=west] at ($(E) + (.25,0)$) (EXTDPT) { $|AB|D, 1, 2$  };
        \node[anchor=west] at ($(F) + (.25,0)$) (EXTDPT) { $|AB|E, 1, 2$  };

        \node at ($(A) + (0,1.5)$) (EXTDPT) { $(\ria\setminus \icon) | \icon| \scon, \rho^*(\scon), \rho^*(\ra)$  };
        \end{tikzpicture}
    \end{center} 
    \caption{A PTCR of $Q(D,E)$}
    \label{fig:app:ptcrExample2}
\end{figure}
\end{minipage}
\end{figure}

\corralgoptcr*

\begin{proof}
    This follows from preorders: Consider 4 elements $A,B,C,D$ ordered in the following way: \begin{align*}
        A\preceq B,C, \quad B,C\preceq D, \quad B\not\preceq C, \quad C\not\preceq B.
    \end{align*}
    Then, necessarily all are different, i.e.,
    \begin{align*}
        A\not\equiv B, \quad A\not\equiv C, \quad A\not\equiv D, \quad B\not\equiv C, \quad B\not\equiv D, \quad C\not\equiv D.
    \end{align*}
    Thus, consequently,
    \begin{align*}
        A\prec B,C, \quad B,C\prec D
    \end{align*}

    To apply this to that classes of query plans considered in Fig-\ref{fig:plan-families} above $\tdptcr$, let us note that the following hold trivially:
    \begin{align*}
        \ptcr \preceq \ptc \preceq \pt \preceq \gj, \quad \tdptcr \preceq \tdpt.
    \end{align*}
    While Theorem \ref{thm:tdgj-pt} showed
    \begin{align*}
        \tdgj \preceq \gj, \quad \tdptcr \preceq\ptcr
    \end{align*}
    as well as
    \begin{align*}
        \pt \not\preceq \tdgj, \quad \tdgj \not\preceq \pt.
    \end{align*}
    Furthermore, Theorem \ref{thm:tdgj-ptc} showed 
    \begin{align*}
        \ptc \preceq \tdgj.
    \end{align*}
    while Theorem \ref{thm:tdpt-tdptc} showed
    \begin{align*}
        \tdpt \preceq \ptc.
    \end{align*}
    Then, lastly, Theorem \ref{thm:ptcr-tdpt} showed 
    \begin{align*}
        \ptcr \not\preceq \tdpt, \quad \tdpt \not\preceq \ptcr.
    \end{align*}
    Consequently, all relationships $\preceq$ can actually be replaced by the strict version $\prec$.
\end{proof}

\fi 

\ifArxiv

\section{Appendix for Section \ref{sec:rpt}}

We start this section by briefly showcasing the computation of the space-time exponent of a RPT by revising the example used in Section \ref{sec:rpt}.

\subsection{Example of a RPT and RPT Algorithm}
\label{sec:app:rptexample}

Let us reconsider the query $Q()$ depicted in Fig.~\ref{fig:app:rptQuery} and the RPT depicted in Fig.~\ref{fig:app:rptExample}.
In Fig.~\ref{fig:app:rptExample}, we drew all nested PTs.
That is, we drew both the replaced sub-pseudo-tree and the removed inputs.
I.e., 
\begin{itemize}
    \item The top PT starts with $L$, has a path to $D$, branches into $A$ and $E$, and has a path from $F$ to $K$ attached to $E$ (draw in gray).
    \item the middle PT with (real) root $F$ replaces the gray path $F-G-H-I-J-K$ of the top PT and has the input $L$ (drawn in gray),
    \item the final PT with (real) root in $I$ replaces the path $I-J-K$ of the middle PT and also has the input $L$ (drawn in gray).
\end{itemize}
For each variable $V$ (excluding input variables and replaced sub-PTs) we added 
$$\ria(V)\setminus \icon(V)| \icon(V)|\scon(V)$$
(i.e., the partitioning of $\raexc(V)$), and $ \outt(V)$ as well as $\rho^*(\scon(V)\cup \outt(V))$ and $\rho^*(\ra(V)\cup \outt(V))$.
Thus, the space-time exponent of this structure is simply the max over the last two quantities (component wise), i.e., $(1,2) = (\rho^*(D), \rho^*(LFEDG))$ achieved for example at $G$ in the middle PT.

\begin{figure}
\begin{minipage}[b]{0.3\textwidth}
\begin{figure}[H]
\centering
\begin{tikzpicture}[font=\footnotesize]

                \foreach \angle/\label in {0/K, -36/I, -72/G, -108/E, -144/C, -180/A}
                {
                    \node[inner sep=1pt] (\label) at (\angle:2) {$\label$};
                }
        
                \foreach \angle/\label in {-18/J, -54/H, -90/F, -126/D, -162/B}
                {
                    \node[inner sep=1pt] (\label) at (\angle:1) {$\label$};
                }

                \node[inner sep=1pt] (L) at (0,0) {$L$};

                \draw (A) -- (B);
                \draw (A) -- (C);
                \draw (A) -- (D);
                \draw (B) -- (C);
                \draw (B) -- (D);
                \draw (B) -- (E);
                \draw (C) -- (D);
                \draw (C) -- (E);
                \draw (D) -- (E);
                \draw (D) -- (F);
                \draw (D) -- (G);
                \draw (E) -- (F);
                \draw (E) -- (G);
                \draw (E) -- (H);
                \draw (F) -- (G);
                \draw (F) -- (H);
                \draw (G) -- (H);
                \draw (I) -- (G);
                \draw (I) -- (H);
                \draw (I) -- (J);
                \draw (I) -- (K);
                \draw (J) -- (G);
                \draw (J) -- (H);
                \draw (J) -- (K);
                \draw (H) -- (K);
                \draw (L) -- (A);
                \draw (L) -- (B);
                \draw (L) -- (C);
                \draw (L) -- (D);
                \draw (L) -- (E);
                \draw (L) -- (F);
                \draw (L) -- (G);
                \draw (L) -- (H);
                \draw (L) -- (I);
                \draw (L) -- (J);
                \draw (L) -- (K);
            \end{tikzpicture}
    \caption{Query $Q()$}
    \label{fig:app:rptQuery}
\end{figure}
\end{minipage}
\hspace{-1cm}
\begin{minipage}[b]{0.59\textwidth}
\begin{figure}[H]
  \begin{center}
    \begin{tikzpicture}[scale = 0.85, font=\footnotesize]
        
            \node (L) {$L$}
                child { node (C) {$C$}
                    child { node (B) {$B$}
                        child { node (D) {$D$}
                            child {node (A) {$A$}}
                            child  {node (E) {$E$}
                                child [edge from parent/.style={draw,gray}] {node (F) {$\textcolor{gray}{F}$}
                                         child { node (G) {$\textcolor{gray}{G}$}
                                            child [xshift=-2cm] { node (H) {$\textcolor{gray}{H}$}
                                                child {node (I) {$\textcolor{gray}{I}$}
                                                    child {node (J) {$\textcolor{gray}{J}$}
                                                        child { node (K) {$\textcolor{gray}{K}$}
                                                        }
                                                    }
                                                }
                                            }
                                        }
                                    }
                                 child [edge from parent/.style={draw,black,dashed}]  {node (F2) {$F$}
                                     child [edge from parent/.style={draw, solid}] { node (E2) {$E$}
                                        child { node (G2) {$G$}
                                            child {node (D2) {$D$}}
                                            child {node (H2) {$H$}
                                                child [edge from parent/.style={draw,gray}]  {node (I2) {$\textcolor{gray}{I}$}
                                                    child { node (J2) {$\textcolor{gray}{J}$}
                                                        child { node (K2) {$\textcolor{gray}{K}$}
                                                        }
                                                    }
                                                }
                                                child [edge from parent/.style={draw,black,dashed}] {node (I3) {$I$}
                                                    child [edge from parent/.style={draw, solid}] { node (H3) {$H$}
                                                        child { node (J3) {$J$}
                                                            child {node (G3) {$G$}}
                                                            child {node (K3) {$K$}}
                                                        }
                                                    }
                                                }
                                            }
                                        }
                                    }
                                }
                            }
                        }
                    }
                };

            \draw[dotted] (F) -- (F2);
            \draw[dotted] (I2) -- (I3);
            
            \node (L2) at ($(F2) + (1,1)$) {\textcolor{gray}{$L$}};
            \node (L3) at ($(I3) + (1,1)$) {\textcolor{gray}{$L$}};

            \draw[gray] (L2) -- (F2);
            \draw[gray] (L3) -- (I3);

            \node[anchor=west] at ($(L) + (.1,0)$) (EXTDPT) { $||, \emptyset, 0, 1$  };
            \node[anchor=west] at ($(C) + (.1,0)$) (EXTDPT) { $|L|, \emptyset, 0, 1$  };
            \node[anchor=west] at ($(B) + (.1,0)$) (EXTDPT) { $|LC|, \emptyset, 0, 1.5$  };
            \node[anchor=west] at ($(D) + (.1,0)$) (EXTDPT) { $|LCB|, \emptyset, 0, 2$  };
            \node[anchor=east] at ($(A) + (-.1,0)$) (EXTDPT) { $|LCBD|, \emptyset, 0, 2.5$  };
            \node[anchor=west] at ($(E) + (.1,0)$) (EXTDPT) { $|LCBD|, \emptyset, 0, 2.5$  };
            \node[anchor=east] at ($(F) + (-.1,0)$) (EXTDPT) { $|L|DE, \emptyset$  };

            \node[anchor=west] at ($(F2) + (.1,0)$) (EXTDPT) { $|L|, DE, 1, 2$  };
            \node[anchor=west] at ($(E2) + (.1,0)$) (EXTDPT) { $|LF|, DE, 1, 2$  };
            \node[anchor=west] at ($(G2) + (.1,0)$) (EXTDPT) { $|LFE|, D, 1, 2.5$  };
            \node[anchor=east] at ($(D2) + (-.1,0)$) (EXTDPT) { $|LFEG|, D, 1, 2.5$  };
            \node[anchor=west] at ($(H2) + (.1,0)$) (EXTDPT) { $|LFEG|, \emptyset, 0, 2.5$  };
            \node[anchor=east] at ($(I2) + (-.1,0)$) (EXTDPT) { $|L|GH, \emptyset$  };

            \node[anchor=west] at ($(I3) + (.1,0)$) (EXTDPT) { $|L|, GH, 1, 2$  };
            \node[anchor=west] at ($(H3) + (.1,0)$) (EXTDPT) { $|LI|, GH, 1, 2$  };
            \node[anchor=west] at ($(J3) + (.1,0)$) (EXTDPT) { $|LIH|, G, 1, 2.5$  };
            \node[anchor=east] at ($(G3) + (-.1,0)$) (EXTDPT) { $|LIHJ|, G, 1, 2.5$  };
            \node[anchor=west] at ($(K3) + (.1,0)$) (EXTDPT) { $|LIHJ|, \emptyset, 0, 2.5$  };

            \node at ($(L) + (0,1)$) (EXTDPT) { $(\ria\setminus \icon) | \icon| \scon, \outt, \rho^*(\scon \cup \outt), \rho^*(\ra \cup \outt)$  };

    \end{tikzpicture}
\end{center} 
  \caption{A RPT}
    \label{fig:app:rptExample}
\end{figure}
\end{minipage}
\end{figure}

For the roots $V$ of the replaced sub-PTs we added the partitioning of $\raexc(V)$ and $\out(V)$ as these will be important for the algorithm.
Note that these sets would be the same, no matter in which order the variables would be on the path.
Further, the replaced PTs can always be made to be a path.
Note that the partitioning $\raexc(V)$ and $\out(V)$ are the same as the ``new'' roots (i.e., $F$, $I$) except that the stored contexts $\scon$ moved to the output $\outt$.
We see the non-replaced parts (drawn in black) as the real vertices of the RPT.

Next, we discuss notation and conventions used in the subsequent proofs.

\subsection{Some Notation and Conventions for RPTs}

Recall the notation used in the recursive definition of RPTs.
The intuition of the recursive construction is to \textit{replace} the subtrees $\descc(A_i)$ of $P$ with $\calP_i$.
We then define the vertices of $\calP$ as the disjoint union (note that a single variable $A$ may appear multiple times in a RPT)
$${\boldsymbol{V}}(\calP) = ({V}(P)\setminus (\iv \cup \bigcup_{i} \descc(A_i))) \sqcup \bigsqcup_{i} \boldsymbol{V}(\calP_i).$$

Next, we aim to prove that Algorithm \ref{alg:app:rpt} runs in time $O(|D|^{t(\calP,\calC,\calI)})$ and uses space $O(|D|^{s(\calP,\calC,\calI)})$.
Thus, this then proves Theorem \ref{thm:algo:rpt}.
But first, let us first clarify some ambiguities.
To that end, let $((P,(\calP_i)_i), (C,(\calC_i)_i), (\iv, (\caliv_i)_i))$ be a sub-RPT that appears in $(\calP,\calC,\calI)$ where the subtree rooted in $A'_i\in V(P)$ was replaced by the RPT $(\calP_i, \calC_i, \caliv_i)$.
Further, let $(Q^{P}(\boldsymbol{X}),\boldsymbol{I})$ be the query to PTCR $(P,C,\boldsymbol{I})$.
\begin{itemize}
    \item Then, $(\calP_i, \calC_i, \caliv_i)$ has to be a RPT for the query 
    $$(Q^{P}_i[\descc(A'_i)\cup\raexc(A'_i)](\scon(A'_i)\cup\outt(A'_i)), \ria(A'_i)).$$
    \item To distinguish multiple copies of the same variable in $\boldsymbol{V}(\calP)$, for a $A\in V(P)$, we denote with $PT(A):=P$ the PT $P$ the variable $A$ stems from.
    \item With $\Root(\calP)$ we mean the root of the top PT in $\calP$.
    \item Note that $V(P)\setminus (\boldsymbol{I} \cup \bigcup \descc(A'_i))\subseteq \boldsymbol{V}(\calP)$
    \item Let $A\in \boldsymbol{V}(\calP)$ and $PT(A)=P$. Then with $child(A)$ we mean the children in $\boldsymbol{V}(\calP)$. Thus, if $A'_i$ is a child of $A$ in the PT $P(A)=P$, then $\Root(\calP_i)$ (ignoring $\caliv_i$) is considered a child of $A$ in $\calP$ (written $\Root(\calP_i)\in child(A)$) and $A'_i$ is not considered a child of $A$ in $\calP$ (written $B'\not\in child(A)$).
    We write $\Root(\calP_i)=A_i$ replaced $A'_i$ (see line \ref{line:app:replaced} of Algorithm~\ref{alg:app:rpt}).
    \item However, we leave $\ra, \raexc, \ria, \icon, \scon, \con, \anc, \ancc, \desc, \descc, \out, \outt$ unchanged, i.e., they are as in $PT(A)$.
    Thus, these sets of variables always stem from the same PT.
    \item When projecting a tuple $\boldsymbol{x}\in \dom^{\anc(A)}$ onto some variables $\boldsymbol{Y}\subseteq V(P)$ of a PT $P$ that is different from $PT(A)$ ($P\neq PT(A)$), we identify the variables.
    I.e., $\boldsymbol{x}[\boldsymbol{Y}]$ is not necessarily the empty tuple.
    \item To avoid overcounting, we have to make sure that the annotations of every relation $R_j(\boldsymbol{X}_j)$ of $Q^{P}(\boldsymbol{X})$ is only used once.
    To that end, in a call to $\textsc{solve}(A,\boldsymbol{x})$, $\textsc{solveReplaced}(A,A',\boldsymbol{x}),$ and $\textsc{fillCache}(A,\yanc)$ whether the relation $R(\boldsymbol{X}_j)$ or $\supp(R_j[\boldsymbol{X}_j])$ is used in $Q$ depends on $PT(A)$.
    This is defined recursively.
    I.e., for variables $V\in \boldsymbol{V}(\calP_i)$, we only use the annotations for relation $R_j(\boldsymbol{X}_j)$ where $\boldsymbol{X}_j\cap \descc(A'_i)\neq \emptyset$.
    Thus, these are the relations that would be been considered further down in $(P,C,\boldsymbol{I})$ and, thus, now have to be considered in a $(\calP_i, \calC_i, \caliv_i)$ as this replaced $(P,C,\boldsymbol{I})$.
    With this, it is always implicitly clear which version of $Q$ we are currently using.
\end{itemize}

{\small

\begin{algorithm}[t]
\caption{RPT Algorithm}
\label{alg:app:rpt}
\begin{algorithmic}[1]
    \Statex \textbf{Input:} $\text{Query } Q, \text{RPT } (\calP,\calC,\calI)$
    \Statex \textbf{Output:} $\sem{Q}$
    \For{$A\in \boldsymbol{V}(\calP)$}
        \State $M_A \gets \emptyset, \ria_A \gets \bot$ 
    \EndFor  
    \State \Return \Call{solve}{$root(\calP),()$}

    \Function{solve}{$A,{\boldsymbol{x}}$} 
        \If{${\boldsymbol{x}}[\ria(A)] = \ria_A $} \Comment{Cache hit}
            \State \Return $M_A({\boldsymbol{x}}[\con(A')])$ 
        \EndIf
        \State $\Call{fillCache}{A,\yanc}$ \label{line:rpt1} \Comment{Cache miss}
        \State \Return $M_A({\boldsymbol{x}}[\con(A)])$
    \EndFunction
    \Function{solveReplaced}{$A, A',{\boldsymbol{x}}$} 
        \If{${\boldsymbol{x}}[\ria(A)] = \ria_A $} \Comment{Cache hit}
            \State \Return $M_A({\boldsymbol{x}}[\con(A')])$ 
        \EndIf
        \State $\ria_A \gets \boldsymbol{x}[\ria(A)]$ \label{line:rpt2} \Comment{$\ria(A)=\ria(A')$, Cache miss}
        \State $\yanc \gets \boldsymbol{x}[\anc(A)] $
            \State $\Call{fillCache}{A,\yanc}$
            \State $\yinst\gets \boldsymbol{x}[\icon(A')]$ \Comment{$\icon(A')=\icon(A)=\con(A)$}
            \State $\text{OUT} \gets M_A(\boldsymbol{x}[\yinst])$
            \State $M_A \gets \{(\yinst, \ysto) \mapsto \text{OUT}[\outt(A')|\yinst, \ysto] \mid \ysto\in \sem{Q[\scon(A')|{\ria_A}]} \}$ \label{line:rptcache}
            \State \Return $M_A({\boldsymbol{x}}[\con(A')])$
    \EndFunction
    \Function{fillCache}{$A,\yanc$} 
        \For{$\ysto\in \sem{Q[{\scon(A)}|{\ria_A}]}$}
            \State $\boldsymbol{y} \gets (\yanc,\ysto)$
            \State OUT $\gets \{\boldsymbol{z} \mapsto \0 \mid \boldsymbol{z}\in\dom^{\outt(A)}\}$  
            \For{$a\mapsto s \in \sem{Q[A|{(\ria_A,\ysto)}]}$}
                    \State $\boldsymbol{y}'\gets (\boldsymbol{y}, a)$
                    \State TMP $\gets \{a[\outt(A)\cap \{A\}] \mapsto s\}$ 
                    \For{$B\in child(A)$} \label{line:rptrc1}
                        \If{$B$ replaced $B'$} \label{line:app:replaced}
                            \State TMP $\gets$ TMP $\otimes$ \Call{solveReplaced}{$B,B',\boldsymbol{y}'$}
                        \Else
                            \State TMP $\gets$ TMP $\otimes$ \Call{solve}{$B,\boldsymbol{y}'$} \label{line:rptrc2}
                        \EndIf  
                    \EndFor     
                    \State OUT $\gets$ OUT $\oplus$ TMP
            \EndFor
            \State $M_A \gets M_A \cup \{\boldsymbol{y}[\con(A)] \mapsto $ OUT$\}$
        \EndFor
        
    \EndFunction
\end{algorithmic}

\end{algorithm}
}

To prove Theorem \ref{thm:algo:rpt}, we essentially have to re-prove the lemmas considered in Section \ref{sec:app:ptcr}.
As these worked by induction, we only have to re-prove the induction steps for when an $A$ replaced an $A'$.

\subsection{Lemmas for Theorem \ref{thm:algo:rpt}}

\begin{lemma}
\label{lem:app:rpt1}
    Let $(\calP,\calC,\calI)$ be a RPT of the SPQ $Q$ and $A\in \boldsymbol{V}(\calP)$ arbitrary. 
    Let $Q(\boldsymbol{X})\leftarrow \bigotimes_i R_i(\boldsymbol{X}_i)$ be a SPQ and $(\calP,\calC,\calI)$ a RPT of $Q(\boldsymbol{X})$.
    Then, for a run of Algorithm \ref{alg:ptcr}, the function $\textsc{solve}(A,\boldsymbol{x})$ is only called for arguments $A\in \boldsymbol{V}(\calP)$ that do not replace a $A'$ and $\boldsymbol{x}\in \dom^{\anc(A)}$ where $\boldsymbol{x}[\raexc(A)]\in \sem{Q[{\raexc(A)}]}$.
    Further, the function $\textsc{solveReplaced}(A,A',\boldsymbol{x})$ is only called for arguments $A\in \boldsymbol{V}(\calP)$ that replaced $A'$, $\boldsymbol{x}\in \dom^{\anc(A')}$ where $\boldsymbol{x}[\raexc(A')]\in \sem{Q[{\raexc(A')}]}$ and $\boldsymbol{x}[\raexc(A)]\in \sem{Q[{\raexc(A)}]}$.
\end{lemma}

\begin{proof}    
We prove this by induction on the recursion depth.
To that end, the base case holds trivially due to $\raexc(\Root(\calP))= \emptyset$.

For the induction step for $\textsc{solve}$, let $A\in \boldsymbol{V}(\calP)$ be such that the induction hypothesis holds.
I.e., all calls are made for $\boldsymbol{x}\in \dom^{\anc(A)}$ where $\boldsymbol{x}[\raexc(A)]\in \sem{Q[{\raexc(A)}]}$.
Now, let $\boldsymbol{x}$ be such that is incurs a cache miss.
Then, the call takes $\yanc := \boldsymbol{x}[\anc(A)\setminus \scon(A)]$ and iterates through $(\ysto,a)\in \sem{Q[{\scon(A)\cup \{A\}}|{\boldsymbol{x}[\ria(A)]}]}$.
Consequently, 
$$(\boldsymbol{x}[\ria(A)], \ysto,a)\in \sem{Q[{\ria(A)\cup \scon(A)\cup \{A\}}]}.$$
Put differently, for $\boldsymbol{y}':=(\yanc, \ysto, a)$, we get 
$$\boldsymbol{y}'[\ra(A)] = (\boldsymbol{x}[\ria(A)], \ysto,a)\in \sem{Q[{\ria(A)\cup \scon(A)\cup \{A\}}]} = \sem{Q[{\ra(A)}]}.$$
Then, let $B\in child(A)$.

\begin{itemize}
    \item If $PT(B)=PT(A)\colon$ we know, $\raexc(B)\subseteq \ra(A)$ and hence, $\boldsymbol{y}'[\raexc(B)]  \in \sem{Q[{\raexc(B)}]}$.
    Thus, the recursive call $\textsc{solve}(B,\boldsymbol{y}')$ is as required.

    \item If $PT(B)\neq PT(A)\colon$ we have $\anc(B) = \raexc(B) = \ria(B)=\ria(A)$. 
    Thus, $\boldsymbol{y}'[\anc(B)] = \boldsymbol{y}'[\raexc(B)] \in \sem{Q[{\raexc(B)}]}.$ 
    Further, let $B'$ be replaced by $B$.
    Then $\boldsymbol{y}'[\raexc(B)]  \in \sem{Q[{\raexc(B)}]}$ due to the same argument as in the previous case.
\end{itemize}

For the induction step for $\textsc{solveReplaced}$, let $A\in \boldsymbol{V}(\calP)$ be such that the induction hypothesis holds and $A'$ be replaced by $A$.
I.e., all calls are made for $\boldsymbol{x}\in \dom^{\anc(A')}$ where $\boldsymbol{x}[\raexc(A)]\in \sem{Q[{\raexc(A)}]}$ and $\boldsymbol{x}[\raexc(A')]\in \sem{Q[{\raexc(A')}]}$.
Now, let $\boldsymbol{x}$ be such that is incurs a cache miss.
Then, the call takes $\yanc := \boldsymbol{x}[\anc(A)]$ and iterates through $a\in \sem{Q[{A}|{\yanc}]}$.
Note that $\scon(A)=\emptyset$ and $\anc(A)=\ria(A)$.
Consequently, $(\yanc,a)\in \sem{Q[{\ria(A) \cup \{A\}}]}$.
Put differently, for $\boldsymbol{y}':=(\yanc, a)$, we get 
$$\boldsymbol{y}' \in  \sem{Q[{\ra(A)}]}.$$
Then, for the children $B\in child(A)$ we can argue as we did for $\textsc{solve}$.

Note that all calls made to the functions $\textsc{solve}, \textsc{solveReplace}$ (except the very first) arise in this form.
Hence, this complete the induction step and with it the proof.
\end{proof}

\begin{lemma}
\label{lem:app:rpt2}
    Let $(\calP,\calC,\calI)$ be a RPT of the SPQ $Q$ and $A\in \boldsymbol{V}(\calP)$ arbitrary. 
    Then, during the execution of Algorithm \ref{alg:ptcr}, the function call $\textsc{solve}(A,\boldsymbol{x})$ returns 
    $$\sem{Q[{\descc(A)}|{\boldsymbol{x}[\con(A)]}](\outt(A))}.$$
    Furthermore, after the call to $\textsc{fillCache}(A,\boldsymbol{x}[\anc(A)\setminus \scon(A)])$ the cache $M_A$ is of the form
    \begin{align}
    \left\{(\yinst,\ysto)\mapsto \sem{Q[{\descc(A)}|{\yinst,\ysto}](\outt(A))}\mid \ysto\in \sem{Q[{\scon(A)}|{\boldsymbol{x}[\ria(A)]}]}\right\}.\label{eq:app:rptcache}   
    \end{align}
    where $\yinst=\boldsymbol{x}[\icon(A)]$.
    Further, the function call $\textsc{solveReplaced}(A,A',\boldsymbol{x})$ returns 
    $$\sem{Q[{\descc(A')}|{\boldsymbol{x}[\con(A')]}](\outt(A'))}.$$
    Furthermore, after Line \ref{line:rptcache}, the cache $M_A$ is of the form
    \begin{align}\{(\yinst,\ysto)\mapsto \sem{Q[{\descc(A')}|{\yinst,\ysto}](\outt(A'))}\mid \ysto\in \sem{Q[{\scon(A')}|{\boldsymbol{x}[\ria(A')]}]}\}.\label{eq:app:rptcache2}   
    \end{align}
    where $\yinst=\boldsymbol{x}[\icon(A)]=\boldsymbol{x}[\icon(A)]$.
\end{lemma}

Note that $\anc(A)\setminus \scon(A) = \anc(A)$ is $A$ replaced $A'$.

\begin{proof}
We prove both claims in a simultaneous induction on the recursion depth.
To that end, the base case are $A$ such that $child(A) = \emptyset$.
We have seen the proof of this base case for $\textsc{fillCache}(A,\yanc)$ and $\textsc{solve}(A,\boldsymbol{x})$ in the proof of Lemma \ref{lem:app:ptcr2}.
Thus, let us immediately go to the base claim for $\textsc{solveReplaced}(A,A',\boldsymbol{x})$.
There are two cases to consider:
\begin{itemize}
    \item $ \boldsymbol{x} [ \ria(A) ] = \ria_A \colon $  Then, there was a previous call with $\textsc{solveReplaced}(A,A',\boldsymbol{z})$ such that the cache was filled and $\boldsymbol{z}[\ria(A)]=\boldsymbol{x}[\ria(A)]$.
    Thus, $M_A$ is of the form given in Eq.~\eqref{eq:app:rptcache2}.
    Furthermore, due to Lemma \ref{lem:app:rpt1}, $\boldsymbol{x}[\raexc(A')]\in \sem{Q[{\raexc(A')}]}$.
    Moreover, $\boldsymbol{x}[\scon(A')]\in \sem{Q[{\scon(A')}|{\boldsymbol{x}[\ria(A')]}]}$.
    Hence, 
    \begin{align*}
        M_A(\boldsymbol{x}[\con(A')]) & = M_A(\boldsymbol{x}[\icon(A')],\boldsymbol{x}[\scon(A')]) \\
        & = \sem{Q[{\descc(A)}|{\boldsymbol{x}[\icon(A')],\boldsymbol{x}[\scon(A')]}](\outt(A'))} \\
        &= \sem{Q[{\descc(A')}|{\boldsymbol{x}[\con(A')]]}(\outt(A'))}.
    \end{align*}
    \item $ \boldsymbol{x} [ \ria(A) ] \neq \ria_A \colon $ Then, $\textsc{fillCache}(A,\boldsymbol{x}[\anc(A)\setminus \scon(A)])$ will be called and after the call, $M_A$ will be of the form given in Eq. \eqref{eq:app:rptcache}.
    Then, consider the transformation performed in Line \ref{line:rptcache}.
    That is, we take the cache in the form of Eq. \eqref{eq:app:rptcache}, group by $\con(A')\neq \con(A)$, and let $\boldsymbol{z}\in \dom^{\con(A')}$ point to the sub-relation that extends $\boldsymbol{z}$.
    Note that, $\scon(A')\cup \outt(A') = \outt(A)$.
    Then, further note that $Q[{\descc(A)}|{\boldsymbol{x}[\icon(A)]}](\outt(A))$ only uses the annotations of relations that use variables from $\descc(A')$.
    Thus, there is no ``overcounting'' and after line~\ref{line:rptcache}, the cache is in the form of Eq. \eqref{eq:app:rptcache2}.
\end{itemize}

Now for the induction step.
To that end, let $child(A) \neq \emptyset$.
Then consider a $B\in child(A)$ such that $B$ replaced $B'$.
Then, note that $\textsc{solveReplace}(B,B',\boldsymbol{x})$ returns the same $\textsc{solve}(B',\boldsymbol{x})$ of Algorithm~\ref{alg:ptcr}.
Thus, also for the induction step of $\textsc{solve}$ and $\textsc{fillCache}$, we can use what we have seen in the proof of Lemma \ref{lem:app:ptcr2}.
Further, for the inductive claim of $\textsc{solveReplaced}(A,A',\boldsymbol{x})$ we can simply reuse the argument for the base claim.
Note that we never used the fact that $A$ has no children there.
\end{proof}

\begin{lemma}
\label{lem:app:rpt3}
    Let $(\calP,\calC,\calI)$ be a RPT of the SPQ $Q$ and $A\in \boldsymbol{V}(\calP)$ arbitrary. 
    Then, for a run of Algorithm \ref{alg:ptcr}, let $\textsc{solve}(A,\boldsymbol{x}_1), \textsc{solve}(A,\boldsymbol{x}_2)$ be two calls where the first happens before the latter.
    Then, $\boldsymbol{x}_1[\ria(A)]$ comes lexicographically before $\boldsymbol{x}_2[\ria(A)]$ (or they are equal).
    Further, let $\textsc{solveReplaced}(A,A',\boldsymbol{x}'_1), \textsc{solveReplaced}(A,A',\boldsymbol{x}_2)$ be two calls where the first happens before the latter.
    Then, $\boldsymbol{x}'_1[\ria(A)]$ comes lexicographically before $\boldsymbol{x}'_2[\ria(A)]$ (or they are equal).
\end{lemma}

\begin{proof}    
Notice that $\ria(A)=\ria(A')$ and that the structure of recursive calls made by $\textsc{solve}(A,\boldsymbol{x})$ and $\textsc{solveReplaced}(A,A',\boldsymbol{x})$ are the same.
Thus, we can essentially reuse the proof of Lemma~\ref{lem:app:ptcr3}.
\end{proof}

\begin{lemma}
\label{lem:app:rpt4}
    Let $(\calP,\calC,\calI)$ be a RPT of the SPQ $Q$ and $A\in \boldsymbol{V}(\calP)$ arbitrary. 
    Then for a run of the Algorithm \ref{alg:app:rpt}, the collective time spent in function calls $\textsc{solve}(A,\boldsymbol{x})$ and $\textsc{solveReplaced}(A,A',\boldsymbol{x})$ (excluding cache hits - counting starts from lines \ref{line:rpt1} and \ref{line:rpt2} - and excluding recursive calls - lines~\ref{line:rptrc1} to~\ref{line:rptrc2}) over all values for $\boldsymbol{x}$ is bounded by $O(|D|^{\rho^*(\ra(A)\cup \out(A))})$.
\end{lemma}

\begin{proof}    
Again, notice that $\ria(A)=\ria(A')$ and that the work done by calls $\textsc{solve}(A,\boldsymbol{x})$ and $\textsc{solveReplaced}(A,A',\boldsymbol{x})$ is fundamentally the same.
Thus, we can essentially reuse the proof of Lemma \ref{lem:app:ptcr4}.
\end{proof}

Now, let us prove Theorem~\ref{thm:algo:rpt}

\subsection{Proof for Theorem \ref{thm:algo:rpt}}

\thmalgorpt*

\begin{proof}

The correctness follows directly by applying Lemma \ref{lem:app:rpt2} to the call $\textsc{solve}(root(\calP),())$ which consequently returns $\sem{Q(\boldsymbol{X})}$ as $\descc(root(\calP))=\var(Q)$ and $\outt(root(\calP))=\boldsymbol{X}$.

To bound the runtime of Algorithm \ref{alg:app:rpt}, we can attribute each step past lines \ref{line:rpt1} or \ref{line:rpt2} (cache miss) to the current call of the function $\textsc{solve}(A,\boldsymbol{x})$, or $\textsc{solve}(A, A',\boldsymbol{x})$ and the steps before lines \ref{line:rpt1} and \ref{line:rpt2}  to the calling function.
As the number of steps before lines \ref{line:rpt1} and \ref{line:rpt2} is constant, we can actually simply ignore them.
To that end, for each variables $A$ we sum up the time over all values for $\boldsymbol{x}$.
Due to Lemma \ref{lem:app:rpt4}, this can be bound by $O(|D|^{\rho^*(\ra(A)\cup \out(A))})$.
Thus, by taking the maximum over all $A\in \boldsymbol{V}(\calP)$, we arrive at a overall bound of $O(|D|^{t(\calP,\calC,\calI)})$ for Algorithm \ref{alg:app:rpt}.

Lastly, to bound the space consumption of Algorithm \ref{alg:app:rpt} is suffice to look at Lemma \ref{lem:app:rpt2} that describes the caches.
That is, for variables $A$, the cache $M_A$ is essentially a subset of 
$$\sem{Q[{\descc(A)\cup \scon(A)}|{\boldsymbol{x}[\ria(A)]}](\scon(A)\cup \outt(A))}$$
no matter if looking at Eq. \eqref{eq:app:rptcache} or \eqref{eq:app:rptcache2}.
The size of the cache $M_A$ can, therefore, never exceed $O(|D|^{\rho^*(\scon(A)\cup\outt(A))})$.   
Hence, the size of all caches $M_A$ together never exceed $O(|D|^{s(P,C)})$.
\end{proof}

Now, let us proceed towards Theorem \ref{thm:rptcr-tdrptcr}.

\subsection{Proof of Theorem \ref{thm:rptcr-tdrptcr}}

\thmrptcrtdrptcr*

\begin{proof}

Let $\Pi=((T,\chi), \pi)$ be a query plan in $\tdrptcr$ of the query $Q$.
We prove that (for $|V(T)|\geq 2$) there is a query plan $\Pi'=((T',\chi'),\pi')$ in $\tdrptcr$ of the query $Q$ with $|T'|< |T|$ and such that $\Pi'\preceq \Pi$.
Thus, by induction, there must also be a single node TD based plan $\Pi'$ that is essentially a RPT in $\rptcr$.

To that end, we explain how to transform $\Pi=((T,\chi), \pi)$ into $\Pi'=((T',\chi'),\pi')$, 
To that end, consider a leaf $c\in V(T)$ and its parent $p\in V(T)$.
Then, first let us make some simplifying assumption that the TD is not sub-optimal.
That is, w.l.o.g., $\forall A\in \boldsymbol{Z}:=\chi(c)\cap \chi(p)\neq \emptyset$ we assume $\exists B\in \chi(c)\setminus \chi(p)$ such that $A,B\in \boldsymbol{X}_i$ for some relation $R_i(\boldsymbol{X}_i)$.
I.e., there are no unnecessary variables in $\chi(c)$.

Then, we know $\pi^c=(\calP^c,\calC^c,\caliv^c)$ is a RPT of (see Eq.~\eqref{eq:tdst}) 
\begin{align}
  (Q^c(\boldsymbol{Z}) \leftarrow & \bigotimes_j R^c_j(\bm X_j \cap \chi(c)), \emptyset) \label{eq:app:sub:query}
\end{align}
and $\pi^p=(\calP^p,\calC^p,\caliv^p)$ ``covers'' $Q^c(\boldsymbol{Z})$, i.e., $Q^c(\boldsymbol{Z})$ is part of the query $Q^p(\boldsymbol{Z}^p)$.
Thus, let us ``search'' for $Q^c(\boldsymbol{Z})$.
To that end, let $(\calP^p,\calC^p,\caliv^p) = ((P,(\calP_i)_i), (C,(\calC_i)_i), (\emptyset, (\caliv_i)_i))$ and let $A_1,\dots,A_k\in V(P)$ be such that $(\calP_i, \calC_i, \caliv_i)$ is an RPT for 
    $$(Q^p[{\desc(A_i)\cup \ra(A_i)}](\scon(A_i)\cup \outt(A_i)), \ria(A_i)).$$
Then, $\boldsymbol{Z}\subseteq \ancc(A)$ for $A=\min( \boldsymbol{Z})\in V(P)$. %
Now, for the general case, assume $A\in \descc(A_i)$ for some~$i$.
Thus, 
$$\boldsymbol{Z}\subseteq \descc(A_i)\cup \con(A_i)\subseteq \descc(A_i)\cup \raexc(A_i)$$
and $(\calP_i, \calC_i, \caliv_i)$ ``covers'' $Q^c(\boldsymbol{Z})$.
I.e., we can continue the search recursively.

We can assert, there is a (sub)-RPT $((P',(\calP'_i)_i), (C',(\calC'_i)_i), (\iv', (\caliv'_i)_i))$ of $(\calP^p,\calC^p,\caliv^p)$ together with $A'_1,\dots,A'_k\in V(P')$ such that $(\calP'_i, \calC'_i, \caliv'_i)$ is an RPT for 
    $$(Q[{\desc(A'_i)\cup \ra(A'_i)}](\scon(A'_i)\cup \outt(A'_i)), \ria(A'_i))$$
but $\descc(A'_i)\cap \boldsymbol{Z} = \emptyset$.
Then, $\boldsymbol{Z}\subseteq \ancc(A)$ and $A=\min(\boldsymbol{Z}) \in V(P')\setminus \{\iv'\cup \bigcup_i \descc(A'_i)\}$.
Then, let us append $\chi(c)\setminus \chi(p)\neq \emptyset$ as a path to $A$.
I.e., we add 1 additional child $B$ to $A'$ which in turn has 1 child, and so on until all of $\chi(c)\setminus \chi(p)$ are added.
Note that $(\chi(c)\setminus \chi(p))\cap \boldsymbol{Z}^p = \emptyset$
Furthermore, we extend $C'$ to assign $0$ to all of $\chi(c)\setminus \chi(p)$ except $C'(B)=|B|$.
Then, in the newly constructed PT, $\con(B)=\scon(B)=\raexc(B)=\boldsymbol{Z}$ while $\ria(B)=\icon(B)=\emptyset=\outt(B)$ and we can replace (the newly introduced sub-PT rooted in) $B$ by $(\calP^c,\calC^c,\caliv^c)$.
The altered structure $\pi'^{p}=(\calP'^p,\calC'^p,\caliv'^p)$ together with the original RPTs $\pi'^{v}:=\pi^v$ for $c\neq v\neq p$ give you a structure $((T',\chi'),\pi')$ where $(T',\chi')$ is obtained from $(T,\chi)$ by simply removing $c$. 
Then, over $(T',\chi')$ the RPT $(\calP'^p,\calC'^p,\caliv'^p)$ is an RPT of 
\begin{align}
  (Q^p(\chi(p)\cap \chi(parent(p))) \leftarrow & \bigotimes_j R^p_j(\bm X_j \cap (\chi(c)\cup \chi(p))) \otimes \bigotimes_{w\in child(p)}Q^w(\boldsymbol{Z}^w),\emptyset ), \label{eq:app:sub:query:1}
\end{align}
where no longer $c\not\in child(p)$.
Furthermore, by construction,
\begin{align}
    s(\calP'^p,\calC'^p,\caliv'^p)&\leq \max(s(\calP^p,\calC^p,\caliv^p), s(\calP^c,\calC^c,\caliv^c))\\
    t(\calP'^p,\calC'^p,\caliv'^p)&\leq \max(t(\calP^p,\calC^p,\caliv^p), t(\calP^c,\calC^c,\caliv^c))
\end{align}
This completes the construction and proves in $\rptcr\preceq \tdrptcr$.
In particular, due to Theorem~\ref{thm:tdgj-pt}, even $\rptcr\equiv \tdrptcr$.
\end{proof}

Lastly, we prove Theorem \ref{thm:rptcr-tdptcr}.

\subsection{Proof of Theorem \ref{thm:rptcr-tdptcr}}

\thmrptcrtdptcr*

\begin{proof}
To prove $\rptcr \preceq \tdptcr$, we simply have to combine Theorem \ref{thm:rptcr-tdrptcr} together with the trivial relationship $\tdrptcr \preceq \tdptcr$.

To prove $\tdptcr \not\preceq \rptcr$, take the scalar query $Q()$ depicted in Fig.~\ref{fig:app:rptQuery}.
It can be evaluated by a plan in $\rptcr$ with space-time exponent $(1,2.5)$ (see discussion in Section \ref{sec:app:rptexample}) by using the RPT depicted in Fig.~\ref{fig:app:rptExample}.
However, using a computer-assisted exhaustive search of the plan space $\ptcr$, we verified that there is no plan PTCR $(P,C)$ with a space exponent of $s(P,C)\leq 1$ and time exponent of $t(P,C)\leq \nicefrac{5}{2}$. 
Since there is also no linear separator in this query, this implies that the same holds for $\tdptcr$. 
\end{proof}

\fi 
\ifArxiv

\section{Appendix for Section \ref{sec:beyond}}

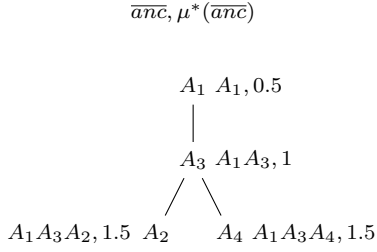
\begin{figure}
\begin{minipage}[b]{0.49\textwidth}
    \begin{center}
        \begin{tikzpicture}[scale = 0.65, font=\footnotesize]
            \node (A_1) {$A_1$}
                child {node (A_3) {$A_3$}
                    child {node (A_2) {$A_2$} }
                    child {node (A_4) {$A_4$} }
                    };

        \node[anchor=west] at ($(A_1) + (.25,0)$) (EXTDPT) { $A_1, 0.5$  };
        \node[anchor=west] at ($(A_3) + (.25,0)$) (EXTDPT) { $A_1A_3, 1$  };
        \node[anchor=west] at ($(A_2) - (3.2,0)$) (EXTDPT) { $A_1A_3A_2, 1.5$  };
        \node[anchor=west] at ($(A_4) + (.25,0)$) (EXTDPT) { $A_1A_3A_4, 1.5$  };

        \node at ($(A_1) + (0,1.5)$) (EXTDPT) { $\ancc, \mu^*(\ancc)$  };
        \end{tikzpicture}
    \end{center} 
    \caption{PT of $Q_{A_1}()$}
    \label{fig:app:square_ptcr_heavy_example}
\end{minipage}
\end{figure}

\fourcycle*
\begin{proof}
    To achieve the submodular width, we first partition the input relation, then we solve the query for different combinations of partitions, using different query plans for each combination. To define the partitioning, we first define a \textit{heavy} join value as one which appears in at least $\sqrt{|E_{i}|}$ tuples of $E_{i}$. All other join values are \textit{light}. Note, by the pigeon-hole principle, there cannot be more than $\sqrt{|E_{i}|}$ heavy values. Given this, we partition $E_{i}$ into $E_{i}^{h,\_},E_{i}^{l,h},E_{i}^{l,l}$ where $h$ and $l$ denote whether the values in the left and right attribute are heavy or light, respectively. 
    This partitioning can be done without using additional space by ordering $E_i$ by heaviness of the first attribute -- this then gives you $E_{i}^{h,\_}$ as the top part of $E_i$ -- and then ordering the $E_i\setminus E_{i}^{h,\_}$ by heaviness of the second attribute -- this then gives you $E_{i}^{l,h}$ as the top part of $E_i\setminus E_{i}^{h,\_}$.
    When handling a heavy attribute, we know that the domain size is bounded by $\sqrt{|E_{i}|}$. When handling a light attribute, we can leverage its low degree.
    
    Consider the following query:
    \begin{align*}
        Q_{A_1}() \leftarrow E_{1}^{h,\_}(A_1,A_2)\otimes E_{2}(A_2,A_3)\otimes E_{3}(A_3,A_4)\otimes E_{4}(A_1,A_4)
    \end{align*}
    This query can be solved by the PT defined in Fig.~\ref{fig:app:square_ptcr_heavy_example} with a space-time complexity of $(0,1.5)$. However, to see that, we need to measure the time and space at each node in a manner that incorporates our new knowledge about the bounded domain size. Formally, this is done by using the polymatroid bound \cite{DBLP:conf/pods/Khamis0S17}, which we denote with $\mu^*$, instead of the AGM bound. Take, for example, the node $A_2$ whose time exponent is $1.5$. This bound is valid because the number of iterations on that branch cannot exceed $|E_{i}(A_2,A_3)\times \supp(E_{i}^{h,\_}[A_1])|\leq |D|^{1.5}$, equally for $A_4$. 
    For any combination of partitions where at least one attribute is heavy, an equivalent PT can be constructed.
    Note that this takes care of all but 1 of the $3^4$ combinations to consider.
    
    Therefore, the only remaining interesting case is the combination of partitions where all attributes are light.
    \begin{align*}
    Q_{LIGHT}() \leftarrow  E_{1}^{l,l}(A_1,A_2)\otimes E_{2}^{l,l}(A_2,A_3)\otimes E_{3}^{l,l}(A_3,A_4)\otimes E_{4}^{l,l}(A_1,A_4)
    \end{align*}
    In this case, we do not see how to simply apply one of the structures considered in this paper, i.e., not even an RPT.
    However, we can proceed as follows:
    We assume the relations are ordered lexicographically in th following way:
    $$E_{1}^{l,l}(A_1,A_2),\quad E_{2}^{l,l}(A_2,A_3),\quad E_{4}^{l,l}(A_1,A_4),\quad E_{3}^{l,l}(A_4,A_3).$$
    That is, $E^{l,l}_1, E^{l,l}_2, E^{l,l}_4$ are ordered in the natural way while $E^{l,l}_3$ is ordered by the attribute $A_4$ and, then, for every $a_4\in \dom^{A_4}$ we order $E^{l,l}_4[A_3]$.
    Then, we aim to perform a merge join on the fly using essentially the decomposition
    $$\bigoplus_{a_1,a_3}(\bigoplus_{a_2} E_{1}^{l,l}(a_1,a_2) \otimes E_{2}^{l,l}(a_2,a_3))\otimes ( \bigoplus_{a_4} E_{4}^{l,l}(a_1,a_4) \otimes E_{3}^{l,l}(a_4,a_3)).$$
    We explain how to iterate through $E_{1}^{l,l}(A_1,A_2) \otimes E_{2}^{l,l}(A_2,A_3)$ projected to $A_1,A_3$ in lexicographic order.
    To that end,, we iterate through $a_1\in \supp(E^{l,l}_1[A_1])\cap \supp(E^{l,l}_4[A_1])$ at the top level.
    Then, individually, we compute $\supp(E^{l,l}_1[A_2|a_1])$ and $\supp(E^{l,l}_4[A_4|a_1])$.
    After that, for every $a_2\in \supp(E^{l,l}_1[A_2|a_1])$ we go through $a_3\in\supp(E^{l,l}_2[A_3|a_2])$ in parallel.
    That is, every $a_2$ spawns a separate process -- hence we use $O(|D|^{\frac{1}{2}})$ additional space as $a_1$ is light, i.e., $|\supp(E^{l,l}_1[A_2|a_1])|\leq |\sqrt{E_1}|$.
    However, we go over pairs $a_1,a_3$ in lexicographic order and we compute 
    $$s = \bigoplus_{a_2\in \supp(E^{l,l}_1[A_2|a_1]) \colon a_3\in \supp(E^{l,l}_2[A_3|a_2])}E_1(a_1,a_2)\otimes E_2(a_2,a_3).$$
    Thus, we enumerate 
    $$(a_1,a_3)\mapsto s \in \sem{\bigoplus_{a_2} E_{1}^{l,l}(A_1,a_2) \otimes E_{2}^{l,l}(a_2,A_3)}.$$
    All of this is possible in time $|\supp(E^{l,l}_1)\Join \supp(E^{l,l}_2)|\leq O(|D|^{\mu^*(A_1A_2A_3)})=O(|D|^{\frac{3}{2}})$.

    The same argumentation allows us to enumerate 
    $$(a_1,a_3)\mapsto s' \in \sem{\bigoplus_{a_4} E_{4}^{l,l}(A_1,a_4) \otimes E_{3}^{l,l}(a_4,A_3)}.$$
    in lexicographic order in time $|\supp(E^{l,l}_1)\Join \supp(E^{l,l}_2)|\leq O(|D|^{\mu^*(A_1A_2A_3)})=O(|D|^{\frac{3}{2}})$.
    Thus, merge joining on the fly allows us to compute $\sem{Q_{LIGHT}()}$ in $O(|D|^{\frac{3}{2}})$ time and using $O(|D|^{\frac{1}{2}})$ additional space (we execute up to $O(|D|^{\frac{1}{2}})$ processes in parallel).
\end{proof}

\fi

\end{document}